\documentclass[final]{siamltex}
\usepackage{amsmath,amssymb, amsfonts}
\usepackage{graphicx,color}
\usepackage{diagbox}
\usepackage[ruled]{algorithm2e}
\usepackage{algorithmic}
\usepackage{enumitem}
\usepackage{multirow}
\usepackage{subfig}
\usepackage{epstopdf}
\usepackage{comment}

\newcommand{\Tr}{\mathrm{Tr}}

\newcommand{\mc}[1]{\mathcal{#1}}
\newcommand{\mf}[1]{\mathfrak{#1}}
\newcommand{\mr}[1]{\mathrm{#1}}

\newcommand{\ud}{\,\mathrm{d}}

\newcommand{\norm}[1]{\lVert#1\rVert}
\newcommand{\average}[1]{\left\langle#1\right\rangle}
\newcommand{\wt}[1]{\widetilde{#1}}

\newcommand{\EE}{\mathbb{E}}

\newcommand{\RR}{\mathbb{R}}

\newtheorem{remark}[theorem]{Remark}

\global\long\def\R{\mathbb{R}}

\global\long\def\vp{\varphi}
\global\long\def\ra{\rightarrow}

\global\long\def\symm{\mathcal{S}^N}

\global\long\def\pd{\mathcal{S}^N_{++}}

\global\long\def\Tr{\mathrm{Tr}}

\numberwithin{equation}{section}
\numberwithin{figure}{section}
\newtheorem{thm}{\protect\theoremname}

\newtheorem{notation}[thm]{Notation}
\numberwithin{thm}{section}

\providecommand{\corollaryname}{Corollary}
\providecommand{\lemmaname}{Lemma}
\providecommand{\propositionname}{Proposition}
\providecommand{\remarkname}{Remark}
\providecommand{\theoremname}{Theorem}

\allowdisplaybreaks

\title{Bold Feynman diagrams and the Luttinger-Ward formalism via Gibbs
measures \vspace{1.5 mm} \\Part I:  Perturbative approach}
\author{
Lin Lin\thanks{Department of Mathematics, University of California, Berkeley, Berkeley, CA 94720 and Computational Research Division, Lawrence Berkeley National Laboratory, Berkeley, CA 94720. Email: \texttt{linlin@math.berkeley.edu}}
\and Michael Lindsey\thanks{Department of Mathematics, University of California, Berkeley, Berkeley, CA 94720. Email: \texttt{lindsey@math.berkeley.edu}}
}

\begin{document}

\maketitle

\begin{abstract}
Many-body perturbation theory (MBPT) is widely used in quantum physics,
chemistry, and materials science. At the heart of MBPT is the Feynman
diagrammatic expansion, which is, simply speaking, an elegant way of
organizing the combinatorially growing number of terms of a certain
Taylor expansion. In particular, the construction of the `bold Feynman
diagrammatic expansion' involves the partial resummation to infinite
order of possibly divergent series of diagrams.  This procedure demands
investigation from both the combinatorial (perturbative) and the
analytical (non-perturbative) viewpoints.  In Part I of this two-part
series, we illustrate how the combinatorial properties of Feynman
diagrams in MBPT can be studied in the simplified setting of Gibbs
measures (known as the Euclidean lattice field theory in the physics
literature) and provide a self-contained explanation of Feynman diagrams in this
setting.  We prove the combinatorial validity of the bold diagrammatic
expansion, with methods generalizable to several types of field theories and
interactions.  Our treatment simplifies the presentation and numerical
study of known techniques in MBPT such as the self-consistent
Hartree-Fock approximation (HF), the second-order Green's function
approximation (GF2), and the GW approximation. 
The bold diagrams are closely related to the Luttinger-Ward (LW) formalism,
which was proposed in 1960 but whose analytic properties have not been
rigorously established. The analytical study of the LW formalism in the
setting of Gibbs measures will be the topic of Part II.
\end{abstract}



\pagestyle{myheadings}
\thispagestyle{plain}

\section{Introduction}\label{sec:intro}
In quantum many-body physics, the computational complexity of obtaining
the numerically exact solution to the many-body Schr\"odinger equation
generally scales exponentially with respect to the number of particles
in the system. Hence a direct approach to the quantum many-body problem is intractable for all but very small systems.
Many-body perturbation theory (MBPT) formally treats the Coulomb
interaction between electrons as a small quantity and provides useful
approximations to many quantities of physical interest with significantly
reduced computational cost.  
MBPT has been demonstrated to be quantitatively useful, and sometimes
indispensable, in a wide range of scientific applications. These range
from the
early description of helium atoms and the uniform electron
gas~\cite{FetterWalecka2003} to modern theories of photovoltaics and
the optical excitation of
electrons~\cite{AryasetiawanGunnarsson1998,OnidaReiningRubio2002,BernardiPalummoGrossman2013}.
Even for `strongly correlated'
systems~\cite{GeorgesKotliarKrauthEtAl1996,BullaCostiPruschke2008} where
a perturbation theory is known to be unsuitable, MBPT still provides the
basic building blocks used in many successful
approaches~\cite{KotliarSavrasovHauleEtAl2006}.

MBPT is usually formulated in the language of second quantization.  For
many problems of interest, the Hamiltonian can be split into a
single-particle term and a two-particle term, which are respectively
quadratic and quartic in the creation and annihilation
operators~\cite{FetterWalecka2003}. These operators are defined on a
Fock space whose dimension scales exponentially with respect to the
system size. Nonetheless, in the special case of `non-interacting'
systems, in which the Hamiltonian contains only the quadratic term, quantities of interest can be obtained exactly. Hence the
non-interacting system is naturally taken as a reference system.   The
remaining quartic interaction, which arises from the Coulomb interaction
between electrons and makes the system `interacting,' is responsible for
almost all of the difficulties in quantum many-body physics. MBPT treats
the quartic term as a perturbation to the non-interacting system. 

Feynman diagrams arise naturally in MBPT as a bookkeeping device for the
coefficients of perturbative series, though they can be further endowed
with physical interpretation~\cite{FetterWalecka2003}. Initially these
diagrams involve contributions from the so-called non-interacting
Green's function (alternatively known as the bare propagator) that
specifies the non-interacting reference problem, as well the quartic
interaction.  Virtually all physical quantities of interest can be
represented perturbatively via such a \emph{bare} Feynman diagrammatic
expansion. Remarkably, the bare Feynman diagrammatic expansions of certain quantities can be 
simplified into \emph{bold} Feynman diagrammatic
expansions~\cite{StefanucciVanLeeuwen2013}. Such an 
expansion is obtained from a bare expansion via a `partial
resummation' procedure, which selects certain pieces of bare diagrams
and sums their contribution to infinite order.  In the bold diagrams,
the role of the bare propagator is assumed by the interacting Green's
function, alternatively known as the bold propagator. This procedure,
referred to as `dressing' or `renormalizing' the propagator, is \emph{a
priori} valid only in a formal sense.  Although it may be initially
motivated as an attempt to simplify the diagrammatic expansion, the
passage from bare to bold diagrams has significant implications. In
particular, most Green's function methods for the theoretical and
numerical investigation of quantum many-body physics, such as the
self-consistent Hartree-Fock approximation, the second-order Green's
function approximation (GF2), the GW
approximation~\cite{Hedin1965}, the dynamical mean-field theory
(DMFT)~\cite{GeorgesKotliarKrauthEtAl1996,Potthoff2006}, the GW+DMFT
method~\cite{BiermannAryasetiawanGeorges2003}, the dynamical cluster
approximation~\cite{StaarMaierSchulthess2013}, and bold diagrammatic Monte
Carlo methods~\cite{ProkofevSvistunov2008,LiLu2017}, can be derived via summation 
over some (possibly infinite) subset of the bold diagrams. 

All of these methods, as well as the bold diagrammatic expansion itself, can be 
viewed as resting on a foundation known as the Luttinger-Ward (LW)\footnote{The Luttinger-Ward formalism
is also known as the Kadanoff-Baym
formalism~\cite{BaymKadanoff1961} depending on the context. In this paper we always use
the former.} 
formalism~\cite{LuttingerWard1960} that originated in 1960. This formalism 
has found widespread usage in physics
and chemistry~\cite{DahlenVanVon2005,Ismail-Beigi2010,BenlagraKimPepin2011,RentropMedenJakobs2016}. 
However, the LW formalism, which is based on a functional of the same
name, is defined only formally.  This is a serious
issue both in theory and in practice. Indeed, even the very existence
of the LW functional in the context of fermionic systems is under debate, with numerical evidence to the
contrary appearing in the past few
years~\cite{KozikFerreroGeorges2015,Elder2014,TarantinoRomanielloBergerEtAl2017,GunnarssonRohringerSchaeferEtAl2017}
in the physics community.

\subsection{Contributions}
This paper is the first of a two-part series and expands on the work 
in~\cite{LinLindsey2018}. In Part I, we provide a 
self-contained explanation of MBPT in the setting of a \emph{Gibbs
model} (alternatively, following the physics literature, 
`Euclidean lattice field theory'). The perturbation theory of this
model, with a specific form of quartic interaction that we refer to as
the \emph{generalized Coulomb interaction}, enjoys a correspondence with
the Feynman diagrammatic expansion for the quantum 
many-body problem with two-body interaction~\cite{NegeleOrland1988,AmitMartin-Mayor2005,AltlandSimons2010}. 
The model is also of interest in its own right and includes, e.g., the
(lattice) $\varphi^{4}$
theory~\cite{AmitMartin-Mayor2005,Zinn-Justin2002}, as a special case. 
In the setting of the Gibbs model, one is interested in the computation of expectation values with respect to possibly 
high-dimensional Gibbs measures.  While the exact computation
of such high-dimensional integrals is generally intractable,
important exceptions are the Gaussian integrals, i.e., integrals for the moments of 
a Gaussian measure, which can be evaluated exactly. Hence the Gaussian measure plays the
role of the reference system. One can construct 
perturbation series using Feynman diagrams, which correspond to
moments of Gaussian measures, to evaluate quantities of interest. 

The main contribution of Part I of this series is the rigorous justification of the
bold diagrammatic expansion in the Gibbs measure setting. Although the
basic idea of the passage from bare to bold diagrams can be intuitively
perceived, the validity of this procedure actually relies on subtle
combinatorial arguments, which to the extent of our knowledge, have not
appeared in the literature. We remark that the arguments appearing in this paper
regarding these manipulations are just as applicable to the quantum
many-body problem as they are to the Gibbs model.  Furthermore, these
arguments clarify why certain quantities such as the self-energy
admit a bold diagrammatic expansion, while other quantities, such as
the free energy, do not.

In fact, bosonic and fermionic field theories (which can in particular be 
derived from the non-relativistic quantum many-body problem via the coherent 
state path integral formalism~\cite{AltlandSimons2010,NegeleOrland1988}) can be viewed formally as infinite-dimensional Gibbs measures 
over complex and Grassmann fields, respectively, in contrast to the real 
`fields' considered in this work. The diagrammatic expansions 
for such theories yield propagator lines with a direction (indicated by
an arrow), due to the distinction between creation and annihilation
operators.  In the setting of the two-body interaction, this additional structure significantly reduces the 
symmetry of the Feynman diagrams, and in fact the self-energy and single-particle Green's function 
diagrams all have a symmetry factor of one. This greatly simplifies the proof of the bold 
diagrammatic expansion in these settings (although a proof of the 
unique skeleton decomposition as in Proposition \ref{prop:selfenergycompose} is still necessary). 
However, we view this simplification as largely accidental because it does not extend 
to interactions beyond the two-body interaction, where more sophisticated arguments 
are necessary (and indeed, to our knowledge, bold diagrams have not yet been considered). 
By contrast the tools introduced here 
can be applied with minimal modification to more complicated interaction forms.

As an auxiliary tool for carefully establishing diagrammatic expansions (both bare and bold), 
we have also found it necessary to introduce definitions of the various 
flavors of Feynman diagrams (as well as associated notions of isomorphism, automorphism, etc.)
in a way that is new, as far as we know. Most of this perspective (which views Feynman 
diagrams as data structures with half-edges as the fundamental building block) is conveyed 
in section \ref{sec:diagDef}. We have also aimed to make this framework 
durable enough not just for the developments of this paper, but also to allow us to pursue 
further (and more sophisticated) diagrammatic manipulations, 
such as the development of the bold screened diagrams, Hedin's
equations~\cite{Hedin1965}, and the Bethe-Salpeter
equation~\cite{OnidaReiningRubio2002}, in future work.

\subsection{Related works}

The construction of Feynman diagrams in the setting of the Gibbs measure
is well known, particularly in quantum field
theory~\cite{PeskinSchroeder1995,AltlandSimons2010}.
To our knowledge, this setting is mostly discussed as a prelude to the setting of quantum field theory 
(in particular, via the coherent state path integral, the quantum
many-body setting)
~\cite{AltlandSimons2010,NegeleOrland1988}, or to 
more general mathematical settings
arising in geometry and
topology~\cite{DeligneEtingofFreedEtAl1999,Polyak2004}.

\subsection{Perspectives}

MBPT is known to be difficult to work with, both analytically and
numerically. In fact, even learning MBPT can be difficult without a
considerable amount of background knowledge in physics.
Hence, more broadly, we hope that this paper will serve as an introduction to 
bare and bold Feynman diagrams that is self-contained, rigorous, and 
accessible to a mathematical audience without a background in
quantum physics.  The
prerequisites for understanding this part of the two-part series are
just multivariable calculus and some elementary combinatorics. 

In fact, our perspective is that the Gibbs model can be used as a point of departure 
(especially for mathematicians) for the study of the many-body problem
in three senses: (1) theoretically, (2) numerically, and (3)
pedagogically. We shall elaborate on these three points presently.

(1) Virtually all of the important concepts of MBPT for the quantum many-body
problem---such as Green's functions, the self-energy, the bare and bold
diagrams, and the Luttinger-Ward formalism, to name a few---have analogs
in the setting of the Gibbs model. The same can be said of virtually all
Green's function methods, including all of the methods cited above.
Furthermore, there is an analog of the impurity
problem, which is fundamental in quantum embedding
theory~\cite{SunChan2016}. 

When rigorous theoretical understanding of the quantum many-body
problem becomes difficult, a lateral move to the more tractable Gibbs model may
yield interesting results.  Headway in this direction is reported Part
II of this series, in which the Luttinger-Ward formalism is
established rigorously for the first time, and 
further progress will be reported in future work. Moreover, studying the
extent to which results \emph{fail} to translate between settings, given
the formal correspondence between the two, may yield interesting
insights.

(2) Numerical methods in MBPT are known to be difficult to implement,
and the calculations are often found to be difficult to converge, time-consuming, or both. Given that virtually all Green's function methods of
interest translate to the Gibbs model, this setting can serve as a
sandbox for the evaluation and comparison of methods in various regimes.
Indeed, one can benchmark these methods by obtaining essentially exact approximations of the
relevant integrals via Monte Carlo techniques. We
hope that the Gibbs model can provide new insights into a number of
difficult issues in MBPT, such as the role of self-consistency in the GW
method, the appropriate choice of vertex correction beyond the GW
method, and the study of embedding methods.

(3) For a mathematical reader, the literature of MBPT can be 
difficult to digest.  In our view the consideration of the Gibbs model
offers perhaps the simplest introduction to the major concepts of MBPT. This
work, together with a familiarity with second quantization and the
basics of many-body Green's functions, should equip the reader to follow
the development of the various approximations and methods found in the
literature, by distilling these concepts via the Gibbs
model. We have attempted to respect this goal by maintaining a
pedagogical style of exposition, with many examples provided throughout
for concreteness.

\subsection{Outline}

In section \ref{sec:prelim} (`Preliminaries') we formally introduce the Gibbs model as well 
as its associated physical quantities such as the partition function, the free energy, and the Green's function.
Here we prove the classic formula (Theorem \ref{thm:wick}) attributed to Isserlis and Wick for computing the 
moments of a Gaussian measure, which is the basis for all Feynman 
diagrammatic expansions. We also quickly recover the Galitskii-Migdal formula 
(Theorem \ref{thm:migdal}) from quantum many-body physics in this
setting using a scaling argument.

In section \ref{sec:feynman}, we introduce various flavors of Feynman diagrams and 
use them to compute diagrammatic expansions for the partition function (Theorem \ref{thm:Zexpand}), 
the free energy (Theorem \ref{thm:logZexpand}), and the Green's function (Theorem \ref{thm:Gexpand}). 
Then, motivated by the prospect of simplifying the perturbative computation of the Green's function, 
we introduce the self-energy and the Dyson equation, which can be used to recover the Green's function 
once the self-energy is known, and then compute the diagrammatic expansion of the self-energy
(Theorem \ref{thm:Sigmaexpand}).

In section \ref{sec:feynmanBold}, the main goal is to formulate and prove the bold diagrammatic expansion of the 
self-energy (Theorem \ref{thm:boldExpand}). This result is only 
a fact about formal power series, but in Part II of this series, we will show that the 
bold diagrammatic expansion admits an analytical interpretation as an asymptotic series, 
in a sense that we preview in Remark \ref{rem:boldExpand}. Theorem \ref{thm:boldExpand}, 
which is a combinatorial result, is in fact used in establishing the
analytical fact in Part II. 
In section \ref{sec:greenMethod} we provide an overview of Green's function methods, 
including a diagrammatic derivation of the GW method and a discussion of 
a property known as $\Phi$-derivability.
In section \ref{sec:LWpre}, we provide a preview of the Luttinger-Ward 
formalism from the diagrammatic perspective and explain 
how the LW functional relates to the free energy.

Finally, in section \ref{sec:numer} we consider a few basic 
numerical experiments with Green's function methods for the Gibbs model.

\section{Preliminaries}\label{sec:prelim}
Before discussing Feynman diagrams in proper, we discuss various preliminaries, including the basics of Gaussian integration.
For simplicity we restrict our attention to real matrices, though analogous results can be obtained in the complex Hermitian case.

\subsection{Notation}

First we recall some basic facts from calculus. For a real symmetric positive definite matrix 
$A\in \RR^{N\times N}$, we define
\begin{equation}
  Z_{0} := \int_{\RR^{N}} e^{-\frac12 x^T A x} \ud x  =
  (2\pi)^{\frac{N}{2}}\left(\det(A)\right)^{-\frac12}.
  \label{eqn:gaussian2}
\end{equation}
The two-point correlation function $G^{0}$ is an $N\times N$ matrix with entries
\begin{equation}
  G^{0}_{ij} := \frac{1}{Z_{0}} \int_{\RR^{N}} x_{i} x_{j} e^{-\frac12
  x^{T} A x}\ud x = (A^{-1})_{ij},
  \label{eqn:G0def}
\end{equation}
i.e., $G^0 = A^{-1}$. (We place the `$0$' in the superscript merely to accommodate 
the use of indices more easily in the notation.)
Note that $G^{0}$ is the covariance
matrix $\EE (XX^{T})=A^{-1}$ of the $N$-dimensional Gaussian random variable 
$X\sim \mc{N}(0,A^{-1})$. 

Now consider a more general $N$-dimensional integral, called the
partition function, given by 
\begin{equation}
  Z = \int_{\RR^{N}} e^{-\frac12 x^{T} A x - U(x)}\ud x,
  \label{eqn:partition}
\end{equation}
where $U(x)$ is called the interaction term. Throughout Part I, we take $U$ to be the following quartic polynomial:
\begin{equation}
  U(x) = \frac{1}{8} \sum_{i,j=1}^{N} v_{ij} x_{i}^2 x_{j}^2.
  \label{eqn:Uterm}
\end{equation}
Without loss of generality we assume that $v_{ij}=v_{ji}$, since
otherwise we can always replace $v_{ij}$ by $(v_{ij}+v_{ji})/2$ without
changing the value of $U(x)$. The factor of $8$ comes from the fact that we do not
distinguish between the $i$ and $j$ indices (due to the symmetry of the
$v$ matrix), nor do between the two terms $x_i x_{i}$ and
$x_{j} x_{j}$. Each will contribute a symmetry factor of $2$ in the developments that follow, 
and this convention
simplifies the bookkeeping of the constants in the diagrammatic series. 
Quartics of the form~\eqref{eqn:Uterm} arise from the discretization of the
 $\varphi^{4}$ theory~\cite{AmitMartin-Mayor2005} and moreover as a classical analog of the interaction 
arising in quantum many-body settings, such as the Coulomb
interaction of electronic structure theory and the interaction term in simplified condensed 
matter models such as the Hubbard model~\cite{MartinReiningCeperley2016}.  With some abuse of notation, we will refer to
the interaction~\eqref{eqn:Uterm} as the \emph{generalized Coulomb
interaction}.

Our results generalize quite straightforwardly to other interactions. 
We will comment in section
\ref{subsec:otherInt} on diagrammatic
developments for other interactions . But for concreteness of the  
example expressions and diagrams throughout, it is simpler to stick to the 
generalized Coulomb interaction as a reference.

Let $\mathcal{S}^{N}$, $\mathcal{S}^{N}_+$, and $\mathcal{S}^{N}_{++}$
denote respectively the sets of symmetric, symmetric positive
semidefinite, and symmetric positive definite $N\times N$ real matrices. 
We also require that 
\begin{equation}
  \frac12 x^{T} A x + U(x) \to +\infty, \quad \norm{x}\to +\infty
  \label{eqn:potrequire}
\end{equation}
for any $A\in \mathcal{S}^{N}$ and moreover that the growth in
Eq.~\eqref{eqn:potrequire} that the integral~\eqref{eqn:partition} is
well-defined. Here $\norm{\cdot}$ is the vector $2$-norm. 
Note that Eq.~\eqref{eqn:potrequire} does not require $A$ to be positive
definite. For instance, if $N=1$, then Eq.~\eqref{eqn:partition} becomes
\begin{equation}
  Z = \int_{\RR} e^{-\frac12 a x^2 - \frac18 v x^4}\ud x, 
  \label{eqn:partition1D}
\end{equation}
and the expression in \eqref{eqn:partition1D} is well-defined as long as $v>0$, regardless
of the sign of $a$. Nonetheless, we assume that $A \in
\mathcal{S}^{N}_{++}$, as this assumption is necessary for the
construction of a perturbative series in the interaction strength. In
Part II of this series of 
papers, Eq.~\eqref{eqn:partition1D} will help us understand the behavior of bold diagrammatic
methods when $A$ is indefinite.

For general $N \geq 1$, there is a natural condition on the matrix $v$
that ensures that integrals like \eqref{eqn:partition1D}
are convergent, namely that the matrix $v$ is positive definite. Indeed, this assumption ensures in particular that $U$
is a nonnegative polynomial, strictly positive away from $x=0$. Since $U$
is homogeneous quartic, it follows that $U\geq C^{-1}\vert x\vert^{4}$
for some constant C
sufficiently large, so for any $A$, the expression 
$\frac{1}{2} x^T A x + U(x)$ goes to $+\infty$ quartically as $\Vert x\Vert \ra \infty$.
Another sufficient assumption is that the entries of $v$ are nonnegative and moreover that the 
diagonal entries are strictly positive.  We will explore the implication
of such conditions in future work.

To simplify the notation, for any function $f(x)$, we define 
\begin{equation}
  \average{f} = \frac{1}{Z} \int_{\RR^{N}} f(x) e^{-\frac12 x^{T} A x -
  U(x)}\ud x, \quad \average{f}_{0} = \frac{1}{Z_{0}} \int_{\RR^{N}}
  f(x) e^{-\frac12 x^{T} A x}\ud x. 
  \label{}
\end{equation}
Throughout the paper we are mostly interested in computing two quantities. The
first is the \textit{free energy}, defined as the negative logarithm of the
partition function: 
\begin{equation}
  \Omega = -\log Z.
  \label{eqn:grandpotential}
\end{equation}
The second is the two-point correlation function (also called the
\textit{Green's function} by analogy with the 
quantum many-body literature),
which is the $N\times N$ matrix 
\begin{equation}
  G_{ij} = \frac{1}{Z} \int_{\RR^{N}} x_{i} x_{j} e^{-\frac12 x^{T} A x
  - U(x)}\ud x =: \average{x_{i}x_{j}}.
  \label{eqn:green}
\end{equation}
It is important to recognize that 
\begin{equation}
  G\in \mc{S}^{N}_{++}.
  \label{eqn:Gdomain}
\end{equation}
In fact, as we shall see in Part II, this
constraint defines the domain of `physical' Green's functions, in a
certain sense. In the discussion below, $G$ is also called the
interacting Green's function, in contrast to the \emph{non-interacting Green's
function} $G^{0}=A^{-1}$. The non-interacting and interacting 
Green's functions are also often called the \emph{bare} 
and \emph{bold propagators}, respectively, especially in the context of 
diagrams.

\subsection{Scaling relation}

The homogeneity of the quartic term $U(x)$ allows for the derivation of a \textit{scaling
relation} for the partition function. Define the $\lambda$-dependent partition
function as
\begin{equation}
  Z_{\lambda} = 
  \int_{\RR^{N}} e^{-\frac12 x^{T} A x - \lambda U(x)}\ud x.
  \label{eqn:Zlambda}
\end{equation}
Then by an change of variable $y=\lambda^{\frac14} x$, 
we have
\begin{equation}
  Z_{\lambda} = \lambda^{-\frac{N}{4}}
  \int_{\RR^{N}} e^{-\frac{1}{2\sqrt{\lambda}} y^{T} A y - U(y)}\ud y.
  \label{eqn:Zlambdascale}
\end{equation}
The scaling relation allows us to represent 
other averaged quantities using the two-point correlation function.  One example
is given in Theorem~\ref{thm:migdal}, which is analogous to the
computation of a quantity called the (internal) energy using the Galitskii-Migdal formula in quantum
physics~\cite{MartinReiningCeperley2016}.

\begin{theorem}[Galitskii-Migdal]
  The internal energy 
  \begin{equation}
    E := \average{\frac12 x^{T} A x + U(x)} = 
    \frac{1}{Z} \int_{\RR^{N}} \left(\frac12 x^{T} A x + U(x)\right) e^{-\frac12 x^{T} A x - U(x)}\ud x
    \label{eqn:Energy}
  \end{equation}
  can be computed using the two point correlation function $G$ as
  \begin{equation}
    E = \frac{1}{4}\Tr[AG+I],
    \label{eqn:migdal}
  \end{equation}
  where $I$ is the $N\times N$ identity matrix.
  \label{thm:migdal}
\end{theorem}
\begin{proof}
  By the definition of $G$, we have
  \begin{equation}
    \average{\frac12 x^{T} A x} = \frac12 \Tr\left[ A
    \average{xx^{T}} \right] = \frac12 \Tr[AG].
    \label{eqn:migdal_t1}
  \end{equation}
  In order to evaluate $\average{U(x)}$, we consider the
  $\lambda$-dependent partition function in Eq.~\eqref{eqn:Zlambda},
  and we have
  \begin{equation}
    -\frac{\ud Z_{\lambda}}{\ud \lambda} \Big\vert_{\lambda=1}= 
    \int_{\RR^{N}} U(x) \, e^{-\frac12 x^{T} A x - U(x)}\ud x.
    \label{eqn:migdal_t2}
  \end{equation}
  Using the scaling relation in Eq.~\eqref{eqn:Zlambdascale},
  we have
  \begin{equation}
    -\frac{\ud Z_{\lambda}}{\ud \lambda} \Big\vert_{\lambda=1}= 
    \frac{N}{4} Z - \int_{\RR^{N}} \frac{1}{4} y^{T} A y \ e^{-\frac12 y^{T} A y -  U(y)}\ud y.
    \label{eqn:migdal_t3}
  \end{equation}
  Combining Eq.~\eqref{eqn:migdal_t1} to~\eqref{eqn:migdal_t3} we have
  \[
  E = \average{\frac12 x^{T} A x + U(x)} = \frac14 \average{x^{T} A x +
  N} = \frac14 \Tr[AG+I].
  \]
\end{proof}

\subsection{Wick theorem}

We first introduce the following notation.  For an even number $m$, we denote by $\mc{I}_{m}$ a set of integers $\{1,\ldots,m\}$. For 
$i\ne j\in\mc{I}_{m}$, we call $(i,j)$ a pair.
A \textit{pairing} $\sigma$ on $\mc{I}_{m}$ is defined to be a partition of
$\mc{I}_{m}$ into $k$ disjoint pairs. For example, the set of all possible pairings of
the set $\mc{I}_{4}=\{1,2,3,4\}$ is  $\{(1,2)(3,4),(1,3)(2,4),(1,4)(2,3)\}$.
Note that a pairing $\sigma$ can be viewed as an element of the permutation
group $\mathrm{Sym}(\mc{I}_{m})$, such that $\sigma^2=1$ and whose action on $\mc{I}_m$
has no fixed points. Under this interpretation $\sigma$ maps any element
$i\in\mc{I}_{m}$ to
the element $\sigma(i)$ of the pairing containing $i$.
For a given pairing $\sigma$, we define the set 
\[
\mc{I}_{m}/\sigma := \{ i\in \mc{I}_{m} \,\vert \, i < \sigma(i)\}
\]
to be the collection of indices corresponding to the `first element'
of each pair.
Denote by $\Pi(\mc{I}_{m})$ the set of all
possible pairings. Observe that there are 
\[
|\Pi(\mc{I}_{m})|=\frac{m!}{2^{m/2} (m/2)!}
\]
pairings in total.

Now Wick's theorem (Theorem~\ref{thm:wick}), also known as Isserlis' theorem~\cite{Isserlis1918}
in probability theory, is the basic tool for deriving the Feynman rules
for diagrammatic expansion. For completeness we give a proof, but since this is a classic result, 
it is provided in Appendix \ref{sec:appWick}.

\begin{theorem}[Isserlis-Wick]\label{thm:wick}
  For integers $1\le \alpha_{1},\ldots,\alpha_{m}\le N$,
  \begin{equation}
    \average{x_{\alpha_{1}}\cdots x_{\alpha_{m}}}_{0} = \begin{cases}
      0,& \text{$m$ is odd},\\
      \sum_{\sigma \in \Pi(\mc{I}_{m})} \prod_{i \in \mc{I}_{m}/\sigma}
      G^{0}_{\alpha_{i},\alpha_{\sigma(i)} }, &\text{$m$ is even}.
    \end{cases}
    \label{}
  \end{equation}
\end{theorem}

In Theorem~\ref{thm:wick}, the indices $\alpha_{i}$ do not need to
be distinct from one another. For example, for $N=4$, 
\[
  \average{x_{1}x_{2}x_{3}x_{4}}_{0} = G^{0}_{12}G^{0}_{34} +
  G^{0}_{13}G^{0}_{24} + G^{0}_{14}G^{0}_{23},
\]
and
\[
  \average{x_{1}^2 x_{3}x_{4}}_{0} = G^{0}_{11}G^{0}_{34} +
  G^{0}_{13}G^{0}_{14} + G^{0}_{14}G^{0}_{13} =
  G^{0}_{11}G^{0}_{34} + 2 G^{0}_{14}G^{0}_{13}.
\]
Similarly
\[
  \average{x_{1}^4}_{0} = G^{0}_{11}G^{0}_{11} +
  G^{0}_{11}G^{0}_{11} + G^{0}_{11}G^{0}_{11} =
  3 G^{0}_{11}G^{0}_{11}.
\]

\section{Feynman diagrams}\label{sec:feynman}

Let us now consider the expansion of quantities such
as $Z$, $\Omega$, and $G$ with respect to a (small) interaction term. For the case currently under 
consideration, in which $U$ is of the form \eqref{eqn:Uterm}, the
size of the interaction
term is measured by the magnitude of the coefficients $v_{ij}$. Equivalently, we can consider a
$\lambda$-dependent interaction as in the definition of the $\lambda$-dependent partition 
function $Z_{\lambda}$ and expand in the small parameter $\lambda$. 
This motivates us to expand $e^{-U(x)}$ using a Taylor series, i.e.
\begin{equation}
  Z = \int_{\RR^{N}} \sum_{n=0}^{\infty} \frac{1}{n!} (-U(x))^{n}
  e^{-\frac12 x^{T} A x}\ud x  \sim
  \sum_{n=0}^{\infty} \frac{1}{n!} \int_{\RR^{N}} (-U(x))^{n}
  e^{-\frac12 x^{T} A x}\ud x.
  \label{eqn:Ztaylor}
\end{equation}
The `$\sim$' indicates that interchanging the order of integration of summation
leads only to an asymptotic series expansion with respect to the
interaction strength, also called the coupling
constant~\cite{MartinReiningCeperley2016}. This can be
readily seen for the example with $n=1$ of Eq.~\eqref{eqn:partition1D}, where
\begin{equation}
  Z_{\lambda} \sim \int
  \sum_{n=0}^{\infty} \frac{1}{n!} \left(-\frac18 \lambda x^4\right)^{n} e^{-\frac12 x^2}
  \ud x = \sum_{n=0}^{\infty} \frac{(-1)^n
  \lambda^{n}}{n!} 2^{-n+\frac12}
  \Gamma\left(2n+\frac12\right).
  \label{}
\end{equation}
Here $\Gamma(\,\cdot\,)$ is the Gamma-function. It is clear that the series
has zero convergence radius, and the series is only an \emph{asymptotic series} in the sense
that the error of the truncation to $n$-th order  is $O(\lambda^{n+1})$ as $\lambda\to 0^+$.

One might have guessed that the radius of convergence must be zero by the following heuristic argument: evidently $Z_{\lambda} = +\infty$ for
any $\lambda<0$, which suggests that the radius of convergence cannot be positive at 
$\lambda = 0$.

In general since $U$ is a quartic polynomial in $x$, the
$n$-th term in Eq.~\eqref{eqn:Ztaylor} can be expressed as the linear
combination of a number of $4n$-point correlation functions for a Gaussian
measure.  These can be readily evaluated using the Wick theorem.

To motivate the need for Feynman diagrams, we first compute the first 
few terms of the expansion for the partition function `by hand'.

The $0$-th order term in~\eqref{eqn:Ztaylor} is clearly
$Z_{0}$. 
Using the Wick theorem, the first-order contribution to $Z_{0}$
contains two terms as
\begin{equation}
  - Z_{0} \sum_{i,j} \frac{1}{8} v_{ij}
  \average{x_{i}^2x_{j}^2}_{0} = - Z_0 \sum_{i,j}
  v_{ij} \left(\frac{1}{8} G^{0}_{ii} G^{0}_{jj} + \frac{1}{4} G^{0}_{ij}
  G^{0}_{ij} \right).
  \label{eqn:Zfirst}
\end{equation}

The second-order contribution, however, can be seen with some effort to contain $8$ distinct terms
as 
\begin{multline}
  Z_{0} \frac{1}{2!} \sum_{i_1,j_1,i_2,j_2} \frac{1}{8^2} v_{i_1j_1} v_{i_2j_2}
  \average{x_{i_1}^2 x_{j_1}^2 x_{i_2}^2 x_{j_2}^2}_{0} \\
  =\  \ Z_{0} \sum_{i_1,j_1,i_2,j_2}
  v_{i_1j_1} v_{i_2j_2} \bigg[ \bigg(\frac{1}{2! \cdot 8^2}
  G^{0}_{i_1i_1}G^{0}_{j_1j_1}G^{0}_{i_2i_2}G^{0}_{j_2j_2} \\ 
  \ \ +\  \frac{1}{2!\cdot 4^2}
  G^{0}_{i_1j_1}G^{0}_{i_1j_1} G^{0}_{i_2j_2}G^{0}_{i_2j_2} + \frac{1}{4\cdot 8}G^{0}_{i_1i_1}G^{0}_{j_1j_1} G^{0}_{i_2j_2}G^{0}_{i_2j_2}  
  \bigg) \\
  + \left( 
\frac{1}{2!\cdot 8} G^{0}_{i_1i_1}G^{0}_{i_2i_2}G^{0}_{j_1j_2}G^{0}_{j_1j_2} 
   +
 \frac{1}{2\cdot 2}
  G^{0}_{i_1j_1}G^{0}_{i_2j_2}G^{0}_{i_1i_2}G^{0}_{j_1j_2} 
 +\frac{1}{4} G^{0}_{i_1j_1}G^{0}_{i_1i_2}G^{0}_{j_1i_2}G^{0}_{j_2j_2} 
  \right) \\
  + \left( \frac{1}{2!\cdot 8} G^{0}_{i_1i_2}G^{0}_{i_1i_2}G^{0}_{j_1j_2}G^{0}_{j_1j_2} 
  + \frac{1}{2!\cdot 4}G^{0}_{i_1i_2}G^{0}_{j_1i_2}G^{0}_{i_1j_2}G^{0}_{j_1j_2}\right)
  \bigg].
  \label{eqn:Zsecond}
\end{multline}
The form in which this expression has been written  (in particular, the form of the denominators of the pre-factors) will become clear later on.

Following the same principle, one can derive higher-order contributions to $Z$.  However, the number of distinct terms in each order grows
combinatorially with respect to $n$. The number of distinct terms, as
well as the associated pre-constants, are already non-trivial in the
second-order expansion.
Feynman diagrams provide a graphical way to systematically organize such terms.

\subsection{Motivation}

In fact it is helpful to view $-v_{ij} x_i^2 x_j^2$ as the contraction
of the fourth-order tensor $-u_{ikjl} x_i x_j x_k x_l$, where $u_{ikjl}
= v_{ij} \delta_{ik} \delta_{jl}$. (Notice that $u_{ikjl}$ is invariant
under the exchange of the first two indices with one another, of the
last two indices with one another, and of the first two indices with the
last two indices. This yields an eightfold redundancy that will become
relevant later on.) Using this insight we can expand the $n$-th term in
the series of Eq.~\eqref{eqn:Ztaylor} as
\begin{equation}
\label{eqn:nthExpand}
\frac{Z_0}{8^n n!} \sum_{i_1,j_1,k_1,l_1,\ldots,i_n,j_n,k_n,l_n = 1}^N \left( \prod_{m=1}^n -v_{i_m j_m} \delta_{i_m k_m} \delta_{j_m l_m} \right) \average{ \prod_{m=1}^n  x_{i_m} x_{j_m} x_{k_m} x_{l_m}}_0.
\end{equation}
One can then use the Wick theorem to express this quantity as a sum over pairings of $\mc{I}_{4n}$. However, it is easier to represent the pairings graphically in the following way. For each $m=1,\ldots,n$, we draw one copy of Fig. \ref{fig:feynman1} (b), i.e., a \emph{wiggled 
line}  known as the interaction line, with four dangling \emph{half-edges} labeled $i,j,k,l$.

\begin{figure}[h]
  \begin{center}
    \includegraphics[width=0.6\textwidth]{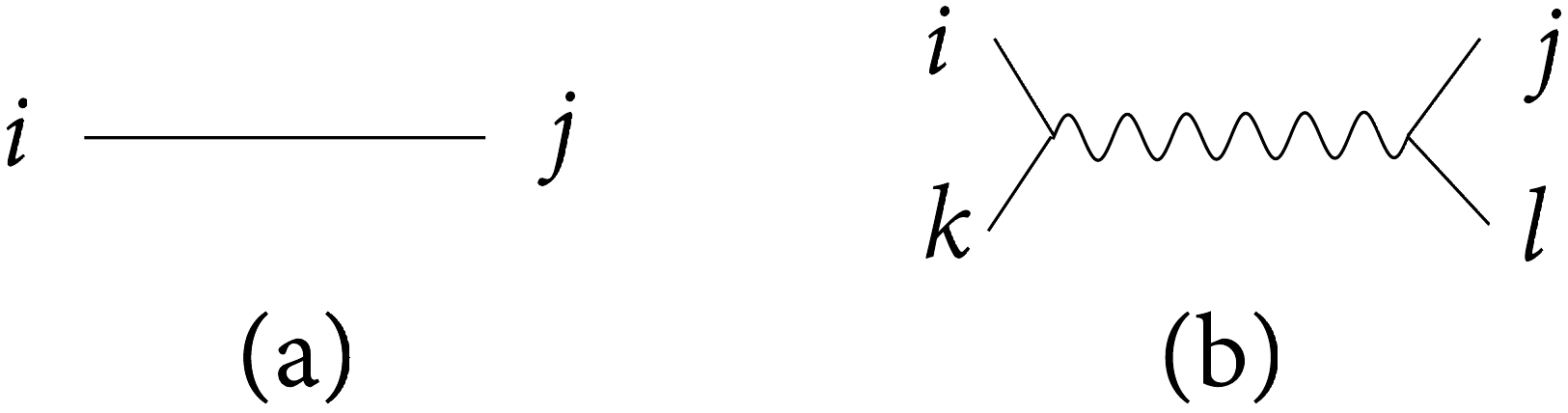}
  \end{center}
  \caption{(a) the bare propagator, $G^{0}_{ij}$\,.\ \ (b) the interaction, $-v_{ij} \delta_{ik}\delta_{jl}$\,.}
  \label{fig:feynman1}
\end{figure}

We can then 
number each interaction line as $1,\ldots,n$ and indicate this by adding an appropriate subscript to thel labels $i,j,k,l$ associated to this vertex. (For the first-order 
terms, since there is only one interaction line, we may skip this step.) The $4n$ half-edges, each with a unique label, 
represent the set on which we consider pairings. We depict a pairing by linking the paired half-edges with a \emph{straight line}, 
which represents the bare propagator $G^0$. The resulting figure is a (labeled, closed) Feynman diagram of order $n$. 
An example of order 2 is depicted in Fig. \ref{fig:linkingExample}.

\begin{figure}[h]
  \begin{center}
    \includegraphics[width=0.25\textwidth]{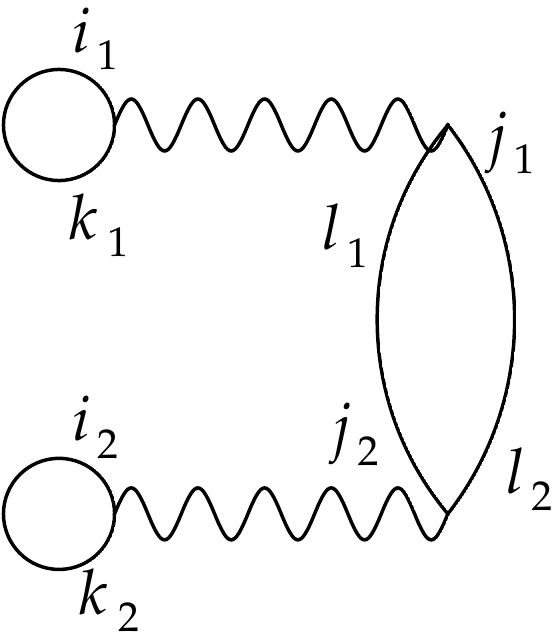}
  \end{center}
  \caption{A labeled closed Feynman diagram of order $2$.}
  \label{fig:linkingExample}
\end{figure}

The quantity associated via Wick's theorem with the pairing represented by such a diagram can then be computed by 
taking a product over all propagators and interaction lines of the associated quantities indicated in Fig. \ref{fig:feynman1} (a) 
and (b), respectively. 
For instance, a line between half-edges $i_1$ and $k_2$ would yield the factor $G^0_{i_1 k_2}$. Meanwhile, the contribution of the
interaction lines altogether is $\prod_{m=1}^n v_{i_m j_m} \delta_{i_m k_m} \delta_{j_m l_m}$. The resulting product is then 
summed over the indices $i_1,j_1,k_1,l_1,\ldots,i_n,j_n,k_n,l_n$. For the example depicted Fig. \ref{fig:linkingExample}, 
this procedure yields the sum

\begin{multline}
\sum_{i_1,j_1,k_1,l_1,i_2,j_2,k_2,l_2} v_{i_1 j_1} \delta_{i_1 k_1} \delta_{j_1 l_1} v_{i_2 j_2} \delta_{i_2 k_2} \delta_{j_2 l_2} 
G^0_{i_1 k_1} G^0_{j_1 l_2} G^0_{l_1 j_2} G^0_{i_2 k_2}  \\ 
= \sum_{i_1,j_1,i_2,j_2} 
v_{i_1 j_1} v_{i_2 j_2} 
G^0_{i_1 i_1} G^0_{i_2 i_2}  G^0_{j_1 j_2} G^0_{j_1 j_2}.
\end{multline}

In summary, we can graphically represented the sum over pairings furnished by Wick's theorem as a sum 
over such diagrams. It is debatable whether we have really made any progress at this point; keeping in mind that 
the diagrams we have constructed distinguish labels, there are as many
diagrams to sum over as there are pairings of $\mc{I}_{4n}$.
Nonetheless, we can use our new perspective to group similar diagrams and mitigate the proliferation of 
terms at high order.

Indeed, many diagrams yield the same contribution. In Fig.~\ref{fig:feynmanZ1label}, the labeled 
first-order diagrams are depicted. Fig.~\ref{fig:feynmanZ1label} (b) and (b') differ only by a relabeling 
that swaps $j$ and $l$ and so yield the same contribution after indices are summed over. From another 
point of view, after removing labels these diagrams become
`topologically equivalent', or isomorphic in some sense.

\begin{figure}[h]
  \begin{center}
    \includegraphics[width=0.7\textwidth]{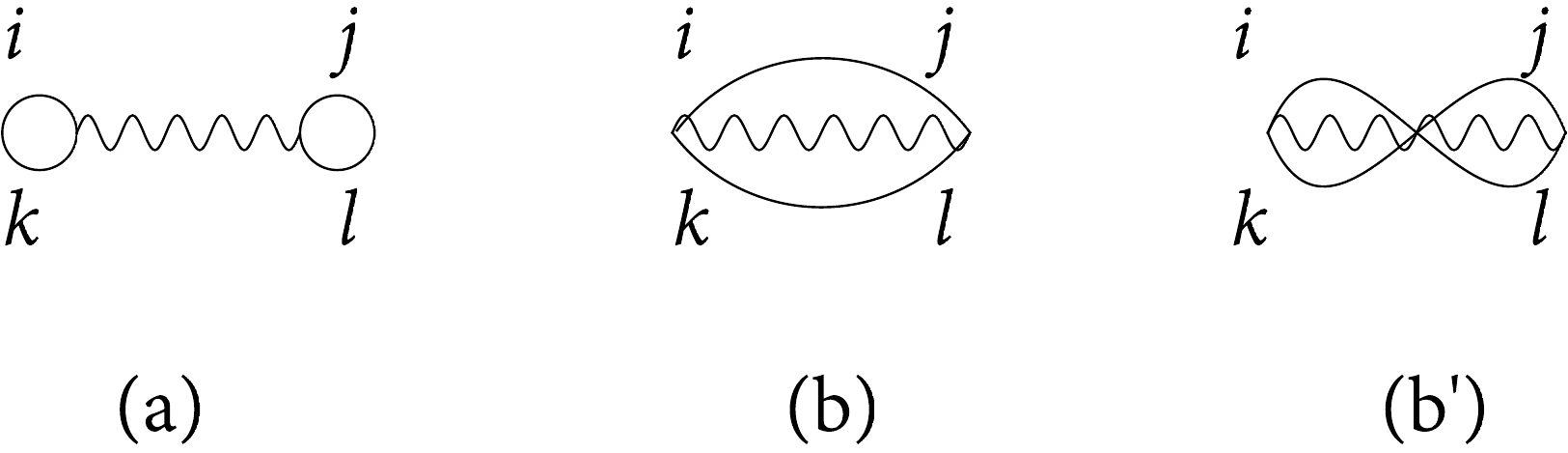}
  \end{center}
  \caption{First order expansion for $Z$ with labeled diagrams. (a) (b) correspond to the
  first and second term in Eq.~\eqref{eqn:Zfirst}. (b') gives an
  equivalent term to (b) and should not be counted twice.}
  \label{fig:feynmanZ1label}
\end{figure}

Our goal is to remove this redundancy in our summation 
by summing only over \emph{unlabeled} diagrams. One expects 
that the `amount' of redundancy of each unlabeled diagram is measured by its symmetry 
in a certain sense. Before making these notions precise, we provide careful definitions of 
labeled and unlabeled closed Feynman diagrams.

\subsection{Labeled and unlabeled diagrams}
\label{sec:diagDef}
We begin with a definition of unlabeled closed Feynman diagrams, and then define labeled diagrams as 
unlabeled diagrams equipped with extra structure. Given $n$ unlabeled interaction lines, each with 
four dangling half-edges, intuitively speaking we produce an unlabeled closed Feynman diagram by 
linking half-edges according to a pairing on all $4n$ of them.
By linking together the half-edges dangling
from a single interaction line, 
one can produce only the two `topologically distinct' diagrams shown in Fig. \ref{fig:feynmanLink1}. 
By applying the linking procedure to two interaction lines, one obtains
the diagrams in Fig.~\ref{fig:feynmanLink2}. 

\begin{figure}[h]
  \begin{center}
    \vspace{4 mm}
    \includegraphics[scale=0.5, trim = 0mm 25mm 0mm 0mm, clip]{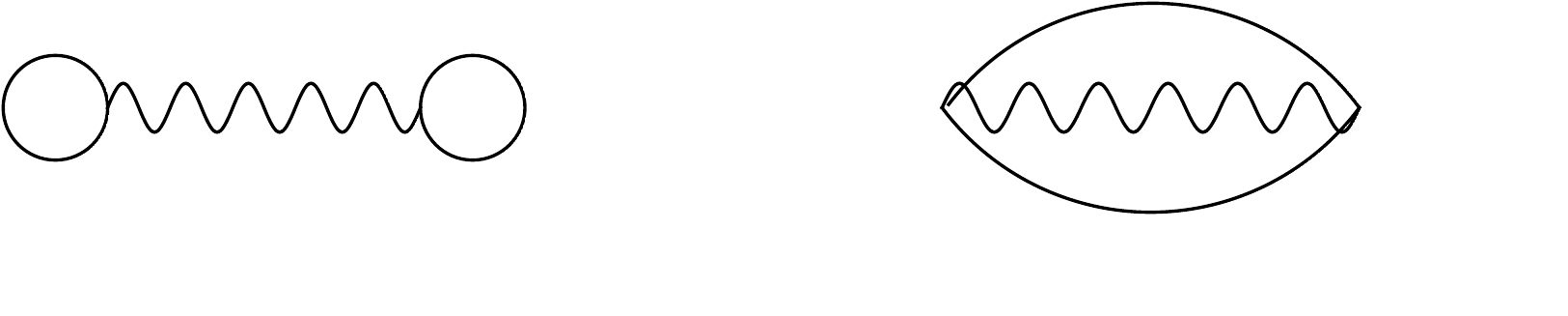}
  \end{center}
  \caption{Unlabeled closed Feynman diagrams of order 1. 
  In many-body perturbation theory, the left-hand diagram corresponds to the
  `Hartree' term and is often referred to as the `dumbbell'
  diagram. The right-hand diagram corresponds to the `Fock exchange' term
  and is often referred to as the `oyster' diagram. }
  \label{fig:feynmanLink1}
\end{figure}

\begin{figure}[h]
  \begin{center}
    \includegraphics[width=0.7\textwidth]{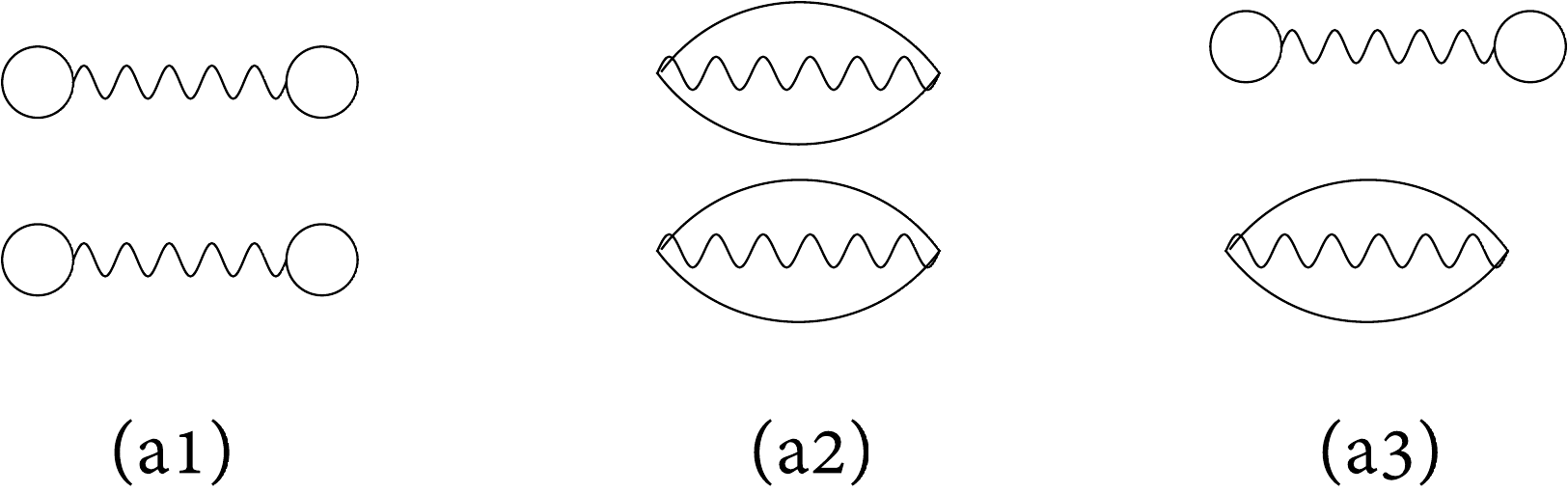}
    
    \vspace{4 mm}
    
    \includegraphics[width=0.7\textwidth]{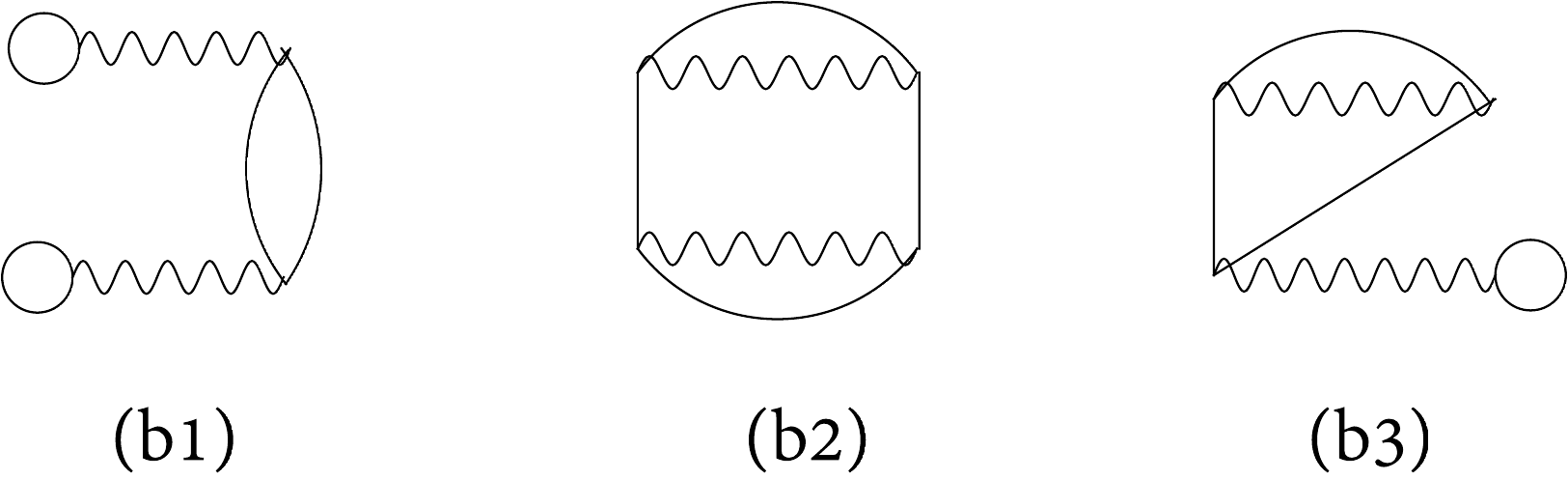}
    
    \vspace{4 mm}
    
    \includegraphics[width=0.5\textwidth]{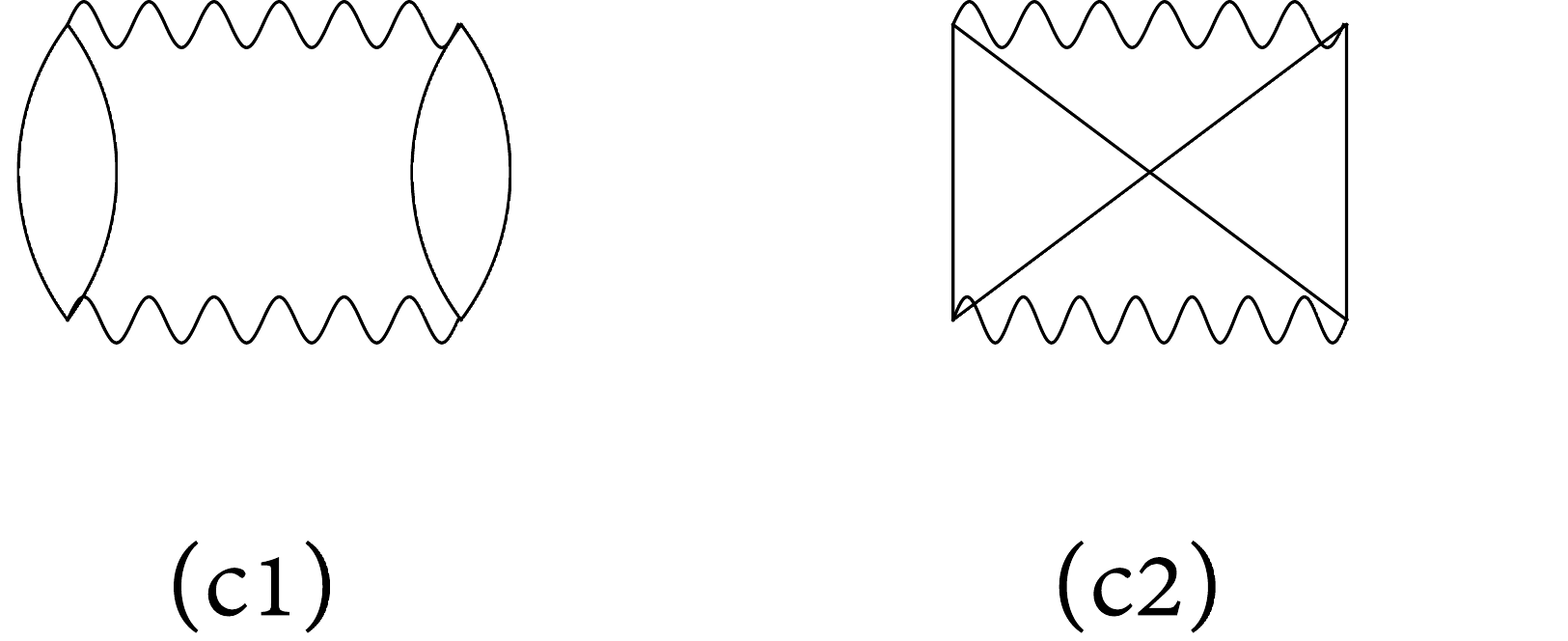}
    
  \end{center}
  \caption{Unlabeled closed Feynman diagrams of order 2.}
  \label{fig:feynmanLink2}
\end{figure}

 Observe that via this linking procedure, each 
interaction line can be viewed as a \emph{vertex} of
degree $4$ in an undirected graph with some \emph{additional structure}, in particular a partition of the four half-edges 
that meet at the vertex into two pairs of half-edges (separated by the wiggled line). Half-edges from the same
interaction line may be linked, so in fact the resulting graph may have self-edges (or loops). (In an undirected graph 
with self-edges, each self-edge contributes 2 to the degree of the vertex, so that the degree indicates the number of half-edges emanating from a vertex.)

In fact it is more natural to view closed Feynman diagrams as being specified via 
the linking of half-edges than it is to view them as undirected graphs specified by vertex and edge sets $(V,E)$. We now provide careful definitions.

\begin{definition}
  An unlabeled closed Feynman diagram $\Gamma$ of order $n$ consists of a vertex set $V$ with 
  $\vert V \vert = n$ and the following extra structure. To the vertices $v\in V$ there are associated 
  disjoint sets $H_1(v)$ and $H_2(v)$ each of cardinality $2$. The union $H(v) := H_1(v) \cup H_2(v)$ is the 
  `half-edge set' of the vertex (or `interaction') $v$, and the partition $\{H_1(v), H_2(v)\}$ reflects the 
  separation of the half-edges into two pairs separated by a wiggled line. The (disjoint) union $\bigcup_{v\in V} H(v)$ 
  is equipped with a partition $\Pi$ into $2n$ pairs of half-edges.\footnote{Intuitively 
  speaking, these data specify a recipe for linking up half-edges to form a connected undirected graph of 
  degree 4, but the previously specified data are a more natural representation of the diagram, especially once 
  labels are introduced.} In total we can view the unlabeled diagram $\Gamma$ as the tuple $\Gamma = (V,H_1,H_2,\Pi)$. 
  For any half-edge $h \in \bigcup_{v\in V} H(v)$, let the unique vertex $v$ associated with this half-edge be denoted 
  by $v = v(h)$.
\end{definition}

\begin{notation}
As a matter of notation going forward, we stress that we maintain a careful 
distinction in the notation between \emph{sets} or \emph{pairs}
$\{ \,\cdot\, , \,\cdot\, \}$, e.g., of half-edges, in which the order of the terms does not 
matter, and  \emph{ordered pairs}
$(\,\cdot\, , \,\cdot\,)$, e.g., of half-edges, in which the order matters.
\end{notation}

We will often refer to different flavors of Feynman diagrams simply as diagrams when 
the context is clear. However, if not otherwise specified, diagrams should be understood
to be unlabeled.

The reader may notice that our depictions of unlabeled diagrams do not distinguish the sides 
of each interaction line from one another by the labels `$1$' and `$2$,' while the definition seems 
to do so. This labeling should indeed not be important when we decide whether or not 
two unlabeled diagrams are `the same.' One could have instead defined an unlabeled diagram to have 
each vertex equipped merely with a partition of its four half-edges into two disjoint pairs, but such a 
definition would be a bit
cumbersome to accommodate notationally without making use of the labels `$1$' and `$2$' anyway later on. 
What is really more important is to define an equivalence relation (a notion of isomorphism) between 
unlabeled diagrams that only cares about the partition of the half-edge set at each vertex, not the labeling 
of the pairs in the partition. Of course such a notion must be introduced regardless of our choice of 
definition:
\begin{definition}
\label{def:unlabeledIsom}
Two unlabeled closed Feynman diagrams  $\Gamma = (V,H_1,H_2,\Pi)$
and  $\Gamma' = (V',H_1',H_2',\Pi')$ are isomorphic if there exists a bijection $\vp : V\ra V'$ and bijections 
$\psi_v : H(v) \ra H'(\vp(v))$ for all $v\in V$, such that
\begin{enumerate}
\item $\psi_v (H_1(v)) = H_1'(\vp(v))$ or $\psi_v (H_1(v)) = H_2'(\vp(v))$ for all $v\in V$.
\item for every $v_1,v_2\in V$, $h_1 \in H(v_1)$, $h_2 \in H(v_2)$, we have $\{h_1, h_2\} \subset \Pi$ if and only if 
  $\{ \psi_{v_1} (h_1), \psi_{v_2} (h_2) \} \subset \Pi '$.
\end{enumerate}
\end{definition}
We will often denote by $\psi$ a bijection between the 
\emph{entire} half-edge sets of two diagrams. Note that the 
$\psi_v$ can be obtained directly from the map $\psi$.

Now we defined the \emph{labeled} closed Feynman diagrams that were introduced informally earlier, as 
well as an appropriate notion of isomorphism for such diagrams.

\begin{definition}
  A labeled closed Feynman diagram $\Gamma$ is specified by an unlabeled closed Feynman diagram 
  $(V,H_1,H_2,\Pi)$, together with a bijection $\mathcal{V} : V \ra \{1,\ldots,n \}$, viewed as a `labeling' of the 
  vertices, as well as a bijection $\mathcal{H}_v : H(v) \ra \{i,j,k,l \}$ for every $v\in V$ 
  which sends $H_1(v)$ to either 
  $\{i,k\}$ or $\{j,l\}$, where
  $i,j,k,l$ are understood as symbols or distinct letters, not numbers. We will denote the collection of these 
  bijections, viewed as labelings of the half-edges associated to each vertex, by $\mathcal{H}$, so in total we 
  can view the labeled diagram $\Gamma$ as the tuple $\Gamma = (V,H_1,H_2,\Pi, \mathcal{V}, \mathcal{H})$. The 
  data $(\mathcal{V},\mathcal{H})$ will be called a labeling of the unlabeled diagram $(V,H_1,H_2,\Pi)$.
\end{definition}

\begin{definition}
Two closed labeled Feynman diagrams  $\Gamma = (V,H_1,H_2,\Pi, \mathcal{V}, \mathcal{H})$
and  $\Gamma' = (V',H_1',H_2',\Pi', \mathcal{V}', \mathcal{H}')$ are isomorphic if they are isomorphic as unlabeled Feynman 
diagrams via maps $\vp$ and $\psi_v$ as in Definition
\ref{def:unlabeledIsom}, which additionally satisfy
\begin{enumerate}
\item $\mathcal{V}(v) = \mathcal{V}'(\vp(v))$ for all $v\in V$, and 
\item $\mathcal{H}_v (h) = \mathcal{H}_{\vp(v)}' ( \psi_v (h) )$ for all $v\in V$, $h\in H(v)$.
\end{enumerate}
\end{definition}

\begin{remark}
\label{rem:labeled}
We can think of two labeled closed Feynman diagrams are isomorphic when they represent the same 
pairing on the set $\{i_1, j_1, k_1, l_1 ,\ldots, i_n, j_n, k_n, l_n\}$ of labels. In other words, the new perspective on 
labeled diagrams as unlabeled diagrams with extra structure is compatible with the old 
perspective on labeled diagrams as pairings, represented graphically by drawing $n$ interaction lines 
as in Fig.~\ref{fig:linkingExample} (b) on the page and then linking their dangling half-edges. The 
definition ensures that the labels $\{i, k\}$ and $\{j, l\}$ appear on opposite sides of 
the $p$-th interaction line in order to ensure this correspondence.
\end{remark}
\begin{remark}
\label{rem:labeled2}
Note that there is only one possible way for two labeled diagrams 
to be isomorphic, since an isomorphism must send each vertex in the 
one to its equivalently labeled vertex in the other, and it must 
send all half-edges associated to a given vertex in one to 
the equivalently labeled half-edges associated to the corresponding 
vertex in the other. This 
completely determines maps $\vp$ and $\psi_v$, so one need only 
to check whether or not these maps define an isomorphism of 
unlabeled diagrams.
\end{remark}

Refer again to Fig.~\ref{fig:feynmanZ1label} for a depiction of labeled closed diagrams. 
Recall that one can assign a numerical value to a labeled diagram by 
taking a formal product of the factors for each edge and each vertex indicated by 
Fig. \ref{fig:feynman1} and then summing over all half-edge labels. In fact, the 
value so obtained is independent of the choice of labeling, hence can be associated 
with the underlying unlabeled diagram as well.

\begin{definition}
The numerical value associated with a labeled or unlabeled diagram $\Gamma$ as in the preceding 
discussion is called the Feynman amplitude of $\Gamma$, denoted $F_\Gamma$.
\end{definition}

For
instance, Fig.~\ref{fig:feynmanZ1label} (a) should be interpreted as
\begin{equation}
  \sum_{i,j,k,l} (-v_{ij}) \delta_{ik}\delta_{jl}
  G^{0}_{ik}G^{0}_{jl} = -\sum_{i,j} v_{ij} G^{0}_{ii}G^{0}_{jj}.
  \label{}
\end{equation}
Comparing with the first term in Eq.~\eqref{eqn:Zfirst}, we see that we are 
missing only the pre-constant $\frac{1}{8}$. In fact the factor $8$ in this 
denominator has a significance that can be understood in terms of
the structure of Feynman diagrams. It is known as the
\textit{symmetry factor} for the Feynman diagram of Fig.~\ref{fig:feynmanZ1label} (a). 

More generally the symmetry factor of any Feynman diagram, which we
shall define shortly, allows us to likewise compute the pre-constants of
the associated term in our series expansion for the partition function.
Roughly speaking, the symmetry factor counts the number of different
labelings of a given labeled Feynman diagram that maintain its structure.  In
particular, after relabeling, two connected half-edges should remain
connected.

To define the symmetry factor more precisely, we first describe more carefully what is 
meant by a `relabeling.' Consider the permutation group $S_4$ on the four letters 
$\{i,j,k,l\}$. Denote by $R$ the subgroup of order $8$ generated by $(i, k)$, $(j, l)$, and $(i,j)(k,l)$. 
(In fact $R$ is isomorphic to the dihedral group of order $8$.)
Observe that the group $\mathbf{R}_n := S_n \times R^n$ 
acts in a natural way on the set of labelings of any 
fixed unlabeled diagram. Here $S_{n}$ acts on the permutation of
$n$ vertices, while $R^{n}$ permutes the associated half-edges.
In other words, $g = (\sigma, \tau_1, \ldots, \tau_n) \in S_n \times R^n$ acts on labelings
by permuting the vertex labelings according to $\sigma$ and by permuting the half-edge labelings at the $p$-th vertex 
according to $\tau_p$. We may think of each such $g$ as a `relabeling.'

\begin{definition}
  An automorphism of a labeled closed Feynman diagram $\Gamma$ of
  degree $n$ is a relabeling $g \in \mathbf{R}_n$ such that $g \cdot \Gamma$ 
  is isomorphic to $\Gamma$ (as a labeled Feynman diagram). The set of all 
  automorphisms of $\Gamma$ forms a subgroup $\mathrm{Aut}(\Gamma)$ 
  of $\mathbf{R}_n$, called the 
  automorphism group of $\Gamma$. The size $|\mathrm{Aut}(\Gamma)|$ of 
  the automorphism group is called the symmetry factor of $\Gamma$ and 
  denoted $S_\Gamma$. (Note that $S_\Gamma$ is independent of the 
  labeling of $\Gamma$, i.e., depends only on the structure of $\Gamma$ 
  as an unlabeled diagram.)
\end{definition}

\begin{remark}
Any relabeling $g \in \mathbf{R}_n$ of $\Gamma$ determines 
maps $\vp$ and $\psi$ from the vertex and half-edge sets of 
$\Gamma$, respectively, to themselves. The map 
$\vp$ is obtained by mapping the vertices of $\Gamma$ to the 
equivalently labeled vertices of $g\cdot \Gamma$, and the 
map $\psi$ is obtained by mapping the half-edges associated 
to each vertex in $\Gamma$ to the equivalently labeled 
half-edges of the equivalently labeled vertex of $\Gamma$. 
Conversely, any such maps $\vp$ and $\psi$ determine 
a relabeling $g \in \mathbf{R}_n$ of $\Gamma$. For any 
$g \in \mathbf{R}_n$, we denote the associated maps
by $\vp_g$ and $\psi_g$

Recalling Remark \ref{rem:labeled2}, it follows 
that $g\cdot \Gamma$ and $\Gamma$ are isomorphic 
as labeled diagrams (i.e., $g \in \mr{Aut}(\Gamma)$) if and only 
if the associated maps $\vp_g$ and $\psi_g$ define an isomorphism 
from $\Gamma$ to itself as an \emph{unlabeled} diagram. In other words, 
automorphisms, which have been defined via actions on labelings, are really just equivalent to 
self-isomorphisms of unlabeled diagrams. However, the perspective 
of labeled diagrams is valuable to retain for the application 
of Wick's theorem.
\end{remark}

For example, Fig.~\ref{fig:feynmanZ1auto} depicts all of the
automorphisms of the diagram in Fig.~\ref{fig:feynmanZ1label} (b), so the 
symmetry factor of this diagram is $4$.
One may readily verify that $S_\Gamma =8$ for the diagram in 
Fig.~\ref{fig:feynmanZ1label} (a).

\begin{figure}[h]
  \begin{center}
    \includegraphics[width=0.5\textwidth]{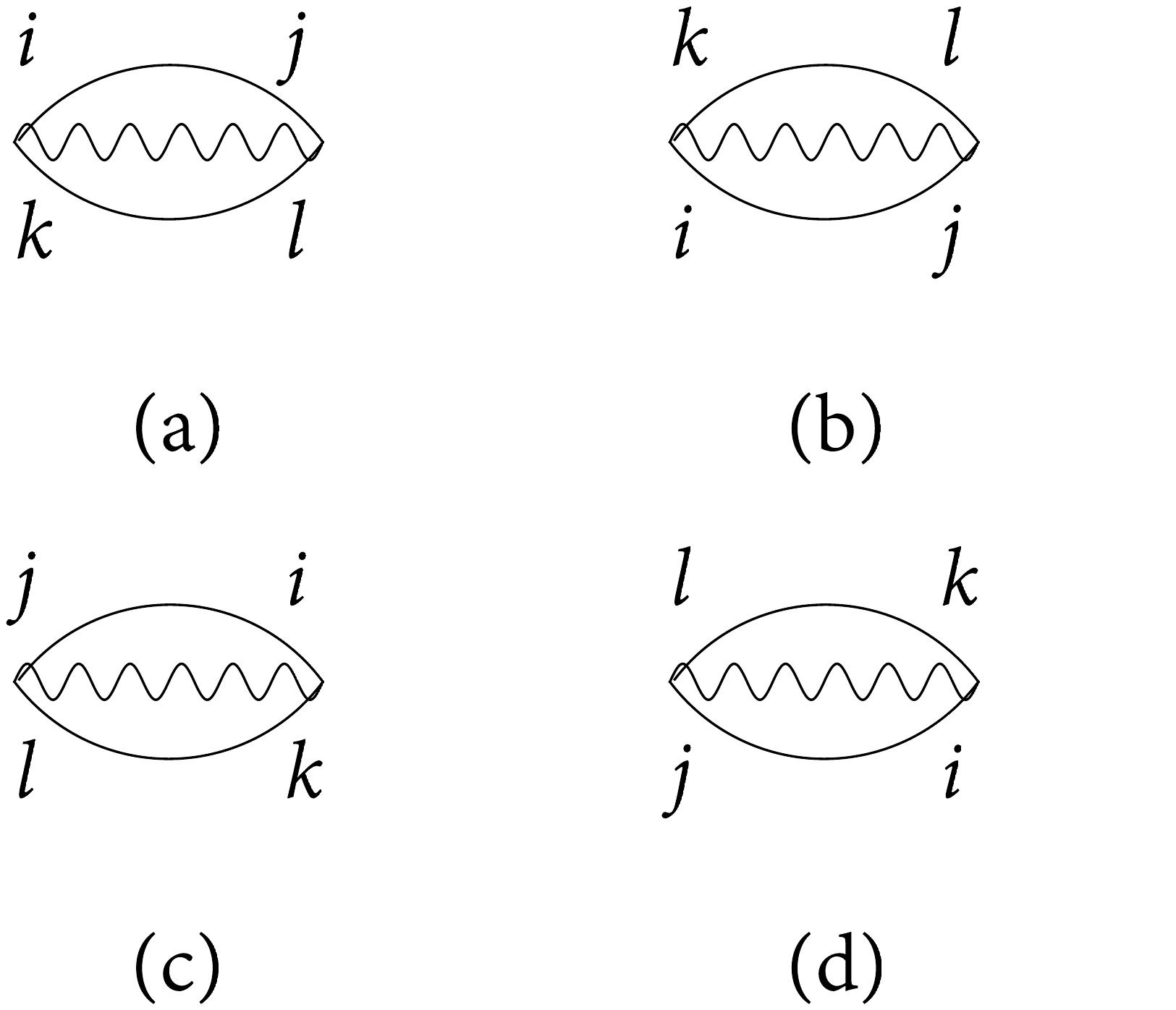}
  \end{center}
  \caption{All automorphisms for Fig.~\ref{fig:feynmanZ1label}
  (b).}
  \label{fig:feynmanZ1auto}
\end{figure}

These symmetry factors recover the pre-factors from our first-order expansion 
of the partition function. This correspondence will be established in general in 
Theorem \ref{thm:Zexpand}.

Before moving on, we comment that two non-isomorphic labeled diagrams can be isomorphic as unlabeled 
diagrams. In this case, the numerical values associated with both are nonetheless the same. 
For instance, Fig.~\ref{fig:feynmanZ1label} (b) represents
\[
\sum_{i,j,k,l} (-v_{ij}) \delta_{ik}\delta_{jl}
  G^{0}_{ij}G^{0}_{kl} = -\sum_{i,j} v_{ij} G^{0}_{ij}G^{0}_{ij},
\]
while (b') represents 
\[
\sum_{i,j,k,l} (-v_{ij}) \delta_{ik}\delta_{jl}
  G^{0}_{il}G^{0}_{kj} = -\sum_{i,j} v_{ij} G^{0}_{ij}G^{0}_{ij},
\]
i.e., the same term.
When we ultimately sum over (isomorphism classes of) unlabeled diagrams in our series expansion 
for the partition function, 
(b) and (b') will \textit{not}
be counted as distinct diagrams. Therefore we record the following definition:

\begin{definition}
The set of (isomorphism classes of) unlabeled closed Feynman diagrams is 
denoted $\mf{F}_{0}$.
\end{definition}

In our new terminology, Fig.~\ref{fig:feynmanLink1} and  Fig.~\ref{fig:feynmanLink2} 
depict all isomorphism classes of unlabeled closed diagrams of first and second order, respectively.
Summation over the unlabeled diagrams, 
as opposed to the labeled diagrams, significantly simplifies the effort of bookkeeping, at the cost 
of computing symmetry factors for each diagram.

\subsection{Feynman rules for $Z$}
\label{sec:feynmanZ}

We are now ready to state and prove the so-called `Feynman rules' 
for the diagrammatic expansion of the partition function, i.e., the recipe for 
producing the Taylor expansion via the enumeration of unlabeled diagrams.

\begin{theorem}\label{thm:Zexpand}
  The asymptotic series expansion for $Z$ is given by 
  \begin{equation}
    Z = Z_{0} \sum_{\Gamma\in \mf{F}_{0}} \frac{F_{\Gamma}}{S_{\Gamma}},
    \label{eqn:ZFeynman}
  \end{equation}
  i.e., the $n$-th term in the series of Eq.~\eqref{eqn:Ztaylor} is given by the sum of $Z_0 \frac{F_\Gamma}{S_\Gamma}$ 
    over isomorphism classes of unlabeled Feynman diagrams $\Gamma$ of order $n$.
\end{theorem}

\begin{remark}
We remind the reader that for a diagram $\Gamma$ of order $n$, the Feynman amplitude $F_\Gamma$ can be computed as follows:
\begin{enumerate}
    \item Assign a dummy index to each of the $4n$ half-edges.
    \item Each edge with half-edge indices $a,b$ yields a factor $G^{0}_{ab}$.
    \item Each interaction line with half-edge indices $a,b,c,d$ yields a factor
      $-v_{ab}\delta_{ac}\delta_{bd}$.
  \item Multiply all factors obtained via steps 2 and 3, and sum over all dummy indices from $1$ to $N$.
  \end{enumerate}
\end{remark}

\begin{proof}
Recall Eq.~\eqref{eqn:nthExpand}, i.e., that we can write the $n$-th term in the series of Eq.~\eqref{eqn:Ztaylor} as
\[
\frac{Z_0}{8^n n!} \sum_{i_1,j_1,k_1,l_1,\ldots,i_n,j_n,k_n,l_n = 1}^N \left( \prod_{m=1}^n - v_{i_m j_m} \delta_{i_m k_m} \delta_{j_m l_m} \right) \average{ \prod_{m=1}^n  x_{i_m} x_{j_m} x_{k_m} x_{l_m}}_0.
\]
By our preceding discussions (see Remark \ref{rem:labeled}) this quantity can be written as 
\[
\frac{Z_0}{8^n n!} \sum_{\Gamma\ \mathrm{labeled,}\,\mathrm{order}\,n } F_{\Gamma}.
\]
We wish to replace the sum over (isomorphism classes of) labeled diagrams 
with a sum over unlabeled diagrams. The question is then: to any unlabeled diagram $\Gamma$ 
of order $n$, how many distinct labeled diagrams 
can be obtained by labeling $\Gamma$? To answer this question first assign an arbitrary labeling 
to obtain a labeled diagram which we shall also call $\Gamma$. Then the set of all labelings is the orbit of 
$\Gamma$ under the group $\mathbf{R}_n$. By the orbit-stabilizer theorem, the size of this orbit 
is equal to $\vert \mathbf{R}_n \vert / \vert (\mathbf{R}_n)_\Gamma \vert$, where $(\mathbf{R}_n)_\Gamma$ 
is the stabilizer subgroup of $\mathbf{R}_n$ with respect to $\Gamma$. But this subgroup is precisely 
$\mathrm{Aut}(\Gamma)$ and $\vert \mathbf{R}_n \vert = 8^n n!$, so the number of distinct labeled diagrams associated with the underlying 
unlabeled diagram is $\frac{8^n n!}{S_\Gamma}$. Therefore the $n$-th term in the series of Eq.~\eqref{eqn:Ztaylor} is
in fact 
\[
Z_0 \sum_{\Gamma\ \mathrm{unlabeled,}\,\mathrm{order}\,n } \frac{F_{\Gamma}}{S_{\Gamma}},
\]
as was to be shown.
\end{proof}

We now apply Theorem \ref{thm:Zexpand} to compute the second-order part of the 
expansion for $Z$.
We can represent the $8$ terms in the
second-order part via the $8$ (isomorphism classes of) 
unlabeled closed Feynman diagrams depicted in
Fig.~\ref{fig:feynmanLink2}, applying Theorem~\ref{thm:Zexpand} to compute
the pre-factor of each term.  The terms are organized into three groups
according to the three groups of terms in Eq.~\eqref{eqn:Zsecond}. 
The diagrammatic approach facilitates the enumeration of these terms
 and allows us to classify the terms more intuitively. 
The first group of
diagrams (a1)--(a3) in Fig.~\ref{fig:feynmanLink2} are simply the diagrams 
obtained as `concatenations'
of two disconnected first-order diagrams. When computing
the symmetry factor, we need to take into account the possible
exchange of the
two interaction lines as well as the symmetry factor of each disconnected piece as 
a first-order diagram. Unlike diagrams (a1) and (a2), diagram (a3) is 
not symmetric with respect to the exchange of the two interaction lines, so the 
former contribution is not included. One can readily verify the correspondence between the rest
of diagrams and terms in Eq.~\eqref{eqn:Zsecond}. The distinction between the (b) and (c)
diagrams will be made clear later on in our discussion of the so-called bold diagrams.

\subsection{Comments on other interactions}
\label{subsec:otherInt}
We pause to make some brief comments on the development of 
Feynman diagrams for other interactions besides the generalized Coulomb interaction of 
Eq.~\eqref{eqn:Uterm}.

First, consider an interaction of the form 
\begin{equation}
  U(x) = \frac{1}{4!} \sum_{i,j,k,l} u_{ikjl} \,x_i x_j x_k x_l,
  \label{eqn:Uterm2}
\end{equation}
where $u_{ikjl}$ is a \emph{symmetric} fourth-order tensor (i.e., invariant under 
any permutation of the indices). The inclusion of the factor of $4!$ owes to the fact that the 
symmetry group of the interaction (i.e., the analog of $R$) is now all of $S_4$, which 
is of order $4!$. Then the developments will 
be much the same, but with the role of the interaction line of Fig.~\ref{fig:feynman1} (b) 
assumed by the device shown in Fig.~\ref{fig:feynman4}.

\begin{figure}[h]
  \begin{center}
  \vspace{2 mm}
    \includegraphics[width=0.2\textwidth]{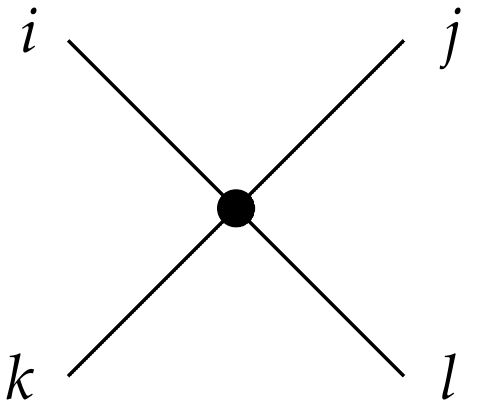}
    \vspace{1 mm}
  \end{center}
  \caption{The interaction $u_{ikjl} $\,.}
  \label{fig:feynman4}
\end{figure}

Since any fourth-order tensor can be symmetrized without changing the 
associated quartic form, why not just consider symmetric interactions? The reason 
is that symmetrizing, e.g., the generalized Coulomb interaction throws away 
its lower-dimensional structure. While there are somewhat fewer diagram topologies
to contend with, each vertex now involves a sum over four indices, not two. 
Moreover, the two-body interaction of quantum many-body physics has a natural 
asymmetry between creation and annihilation operators that is reflected in the 
structure of the Feynman diagrams for the generalized Coulomb interaction.

There is nonetheless a way to generalize the generalized Coulomb interaction 
without destroying its structure. Indeed, simply consider
\begin{equation}
  U(x) = \frac{1}{8} \sum_{i,j,k,l} u_{ikjl} \,x_i x_j x_k x_l,
  \label{eqn:Uterm3}
\end{equation}
where $u_{ikjl}$ is invariant under (1) the exchange $i$ with $k$, 
(2) the exchange of $k$ with $l$, and (3) the simultaneous exchange 
of $i$ with $j$ and $k$ with $l$. In other words, the symmetry group of 
the interaction is $R$ as for the generalized Coulomb interaction.

The developments for interactions of the form \eqref{eqn:Uterm3}---with 
the interaction line of Fig.~\ref{fig:feynman1} (b) now contributing 
the factor $u_{ikjl}$ in the computation of Feynman amplitudes---are 
no different than for interactions of the form \eqref{eqn:Uterm}, 
with the exception of the GW approximation to be discussed in section \ref{sec:GW}. 
For the sake of writing down concrete expressions that correspond 
to various diagrams, we simply assume an interaction of the 
form \eqref{eqn:Uterm}.

\subsection{Feynman rules for $\Omega$}\label{sec:feynmanOmega}

The free energy $\Omega$ is given by the negative logarithm of $Z$ as in
Eq.~\eqref{eqn:grandpotential}, which appears to be difficult to
evaluate in terms of Feynman diagrams.  It turns out that the logarithm 
in fact simplifies the
diagrammatic expansion by removing the disconnected diagrams as in
Fig.~\ref{fig:feynmanLink2} (a). This is the content of 
Theorem~\ref{thm:logZexpand} below, which is called the linked cluster
expansion in physics literature.

Before stating the theorem, we establish some notation. Recall that a closed diagram 
induces an undirected graph of degree four. We say that a closed diagram is \emph{connected}
if the induced graph is connected.
\begin{definition}
The set of all connected closed diagrams is denoted $\mf{F}^{\mr{c}}_{0}\subset \mf{F}_{0}$.
\end{definition}

Similarly we can talk about connected components of a 
Feynman diagram in the obvious way. We can also consider the `union' $\Gamma_1 \cup \Gamma_2$
of diagrams, i.e., the diagram constructed by viewing $\Gamma_1$ and $\Gamma_2$ as disconnected 
pieces of the same diagram.
We leave more careful definitions of these notions to the reader.
We establish a special notation for the union of several copies of the same diagram:
\[
\Gamma^{n} := \bigcup_{j=1}^{n} \Gamma
\]

A general diagram $\Gamma\in \mf{F}_{0}$ can be decomposed as
\begin{equation}
  \Gamma = \bigcup_{i=1}^{K} \Gamma_{i}^{n_{i}}, 
  \label{eqn:gamma_decompose}
\end{equation}
where $\Gamma_{1},\ldots,\Gamma_{K} \in \mf{F}^{\mr{c}}_{0}$ are distinct.

For any diagram $\Gamma$ expressed in the form
of~\eqref{eqn:gamma_decompose}, the Feynman amplitude is
\begin{equation}
  F_{\Gamma} = F^{n_{1}}_{\Gamma_{1}}\cdots
  F^{n_{K}}_{\Gamma_{K}}, 
  \label{}
\end{equation}
and since $\Gamma_{1},\ldots,\Gamma_{K}$ are distinct diagrams, the
symmetry factor is 
\begin{equation}
  S_{\Gamma} = (n_{1}!\cdots n_{K}!)\,S^{n_{1}}_{\Gamma_{1}}\cdots
  S^{n_{K}}_{\Gamma_{K}}.
  \label{}
\end{equation}

It is convenient to define $F_{\Gamma_{\emptyset}}=1$ and $S_{\Gamma_{\emptyset}}=1$ for the `empty' Feynman
diagram $\Gamma_{\emptyset}$ of order zero and moreover to let $\Gamma^0 = \Gamma_{\emptyset}$ for any 
diagram $\Gamma \in \mf{F}_0$.

Using this notation, we can then think of every diagram $\Gamma \in
\mf{F}_0$ as being uniquely specified by a function
$n:\mf{F}^{\mr{c}}_{0} \ra \mathbb{N}$ mapping $\Gamma \mapsto
n_{\Gamma}$, where $\mathbb{N}$ indicates the set of natural numbers
including zero. We denote the set of such functions by
$\mathbb{N}^{\mf{F}^{\mr{c}}_{0} }$.  Indeed, any such function
specifies a diagram $\Gamma(n) := \bigcup_{\Gamma\in
\mf{F}^{\mr{c}}_{0}} \Gamma^{n_{\Gamma}}$.  Moreover, $F_{\Gamma(n)} =
\prod_{\Gamma \in \mf{F}_0^{\mr{c}}} F_{\Gamma}^{n_{\Gamma}}$, and
$S_{\Gamma(n)} = \prod_{\Gamma \in \mf{F}_0^{\mr{c}}} n_{\Gamma}!\,
S_{\Gamma}^{n_{\Gamma}} $.

Now we are ready to state and prove the diagrammatic expansion for the free energy.

\begin{theorem}[Linked cluster expansion for $\Omega$]\label{thm:logZexpand}
  The asymptotic series expansion for $\Omega$ is
  \begin{equation}
    \Omega = \Omega_{0} - \sum_{\Gamma\in \mf{F}^{\mr{c}}_{0}} \frac{F_{\Gamma}}{S_{\Gamma}},
    \label{eqn:logZFeynman}
  \end{equation}
  where $\Omega_{0}=-\log Z_{0}$.
\end{theorem}
\begin{proof}
Exponentiating both sides of Eq.~\eqref{eqn:logZFeynman} motivates the consideration of the following 
expression:
\begin{equation}
\label{eqn:expExpand}
\exp\left( \sum_{\Gamma\in \mf{F}^{\mr{c}}_{0}}
      \frac{F_{\Gamma}}{S_{\Gamma}} \right) = \sum_{K=0}^\infty \frac{1}{K!} 
      \left( \sum_{\Gamma\in \mf{F}^{\mr{c}}_{0}}
      \frac{F_{\Gamma}}{S_{\Gamma}} \right)^K.
\end{equation}
We aim to relate this expansion to our expansion for the partition function from Theorem \ref{thm:Zexpand}. 

We will apply the multinomial theorem to compute the $K$-th power of the sum over $\Gamma\in \mf{F}^{\mr{c}}_{0}$ 
appearing on the right-hand side of Eq.~\eqref{eqn:expExpand}. This yields a sum over $n \in \mathbb{N}^{\mf{F}^{\mr{c}}_{0} }$ 
such that $\sum_{\Gamma \in \mf{F}^{\mr{c}}_{0}} n_{\Gamma} = K$ weighted by the multinomial coefficients 
$\frac{K!}{\prod_{\Gamma \in  \mf{F}_{0}^{\mr{c}} } (n_{\Gamma})! }$, as in
\begin{eqnarray*}
\exp\left( \sum_{\Gamma\in \mf{F}^{\mr{c}}_{0}}
      \frac{F_{\Gamma}}{S_{\Gamma}} \right) & = & \sum_{K=0}^\infty \frac{1}{K!} \ 
      \sum_{n \in \mathbb{N}^ {\mf{F}_{0}^{\mr{c}} } \,:\, \sum_{\Gamma} n_{\Gamma} = K} \ 
      \frac{K!}{\prod_{\Gamma \in  \mf{F}_{0}^{\mr{c}} } (n_{\Gamma} !) }
      \prod_{\Gamma \in \mf{F}_0^{\mr{c}} } \left( \frac{ F_{\Gamma}}{S_{\Gamma}}\right)^{n_{\Gamma}} \\
      & = & 
      \sum_{K=0}^\infty \ 
      \sum_{n \in \mathbb{N}^ {\mf{F}_{0}^{\mr{c}} } \,:\, \sum_{\Gamma} n_{\Gamma} = K} \ 
      \frac{ F_{\Gamma(n)}}{S_{\Gamma(n)}} \\ 
      & = & 
      \sum_{n \in \mathbb{N}^ {\mf{F}_{0}^{\mr{c}} }} \ 
      \frac{ F_{\Gamma(n)}}{S_{\Gamma(n)}},
\end{eqnarray*}
where in the penultimate step we have used our formulas for the Feynman amplitude and symmetry factor
of the diagram $\Gamma(n)$ associated to $n \in \mathbb{N}^{\mf{F}^{\mr{c}}_{0} }$. But since 
$\mathbb{N}^{\mf{F}^{\mr{c}}_{0} }$ is in bijection with $\mf{F}_0$ via $n \mapsto \Gamma(n)$, we have proved:

\begin{eqnarray*}
\exp\left( \sum_{\Gamma\in \mf{F}^{{\mr{c}}}_{0}}
      \frac{F_{\Gamma}}{S_{\Gamma}} \right) = \sum_{\Gamma \in \mf{F}_{0}} \frac{ F_{\Gamma}}{S_{\Gamma}} = \frac{Z}{Z_0}, 
 \end{eqnarray*}
 with the last equality following from Theorem \ref{thm:Zexpand}.
Taking logarithms yields the theorem.
\end{proof}

For example, the second-order contribution to $\Omega$ is 
\begin{multline}
\sum_{i_1,j_1,i_2,j_2}
  v_{i_1j_1} v_{i_2j_2} \bigg[ \bigg( 
\frac{1}{2!\cdot 8} G^{0}_{i_1i_1}G^{0}_{i_2i_2}G^{0}_{j_1j_2}G^{0}_{j_1j_2}  \\
   +\ 
 \frac{1}{2\cdot 2}
  G^{0}_{i_1j_1}G^{0}_{i_2j_2}G^{0}_{i_1i_2}G^{0}_{j_1j_2} 
 +\frac{1}{4} G^{0}_{i_1j_1}G^{0}_{i_1i_2}G^{0}_{j_1i_2}G^{0}_{j_2j_2} 
  \bigg) \\
  + \ \left( \frac{1}{2!\cdot 8} G^{0}_{i_1i_2}G^{0}_{i_1i_2}G^{0}_{j_1j_2}G^{0}_{j_1j_2} 
  + \frac{1}{2!\cdot 4}G^{0}_{i_1i_2}G^{0}_{j_1i_2}G^{0}_{i_1j_2}G^{0}_{j_1j_2}\right)
  \bigg], 
  \label{eqn:Omegasecond}
\end{multline}
and the terms are yielded by Fig.~\ref{fig:feynmanLink2} (b), (c).

\subsection{Feynman rules for $G$}\label{sec:feynmanG}

Our next goal is to obtain a diagrammatic
expansion for the Green's function $G$. First observe that the 
asymptotic series expansion for $ZG$ can be written, similarly to that of 
$Z$, as 
\begin{equation}
  Z G_{ij} 
 \sim
  \sum_{n=0}^{\infty} \frac{1}{n!} \int_{\RR^{N}} x_{i} x_{j}(-U(x))^{n}
  e^{-\frac12 x^{T} A x}\ud x.
  \label{eqn:ZGTaylor}
\end{equation}
Again the interchange between the summation and integration is only
formal. The right hand side of Eq.~\eqref{eqn:ZGTaylor} can be evaluated
using the Wick theorem and a new class of Feynman diagrams. 

Similarly to Eq.~\eqref{eqn:nthExpand}, we see that the $n$-th term in the expansion of
Eq.~\eqref{eqn:ZGTaylor} is given by 
\begin{equation}
\label{eqn:nthExpandG}
\frac{Z_0}{8^n n!} \sum_{i_1,j_1,k_1,l_1,\ldots,i_n,j_n,k_n,l_n = 1}^N \left( \prod_{m=1}^n -v_{i_m j_m} \delta_{i_m k_m} \delta_{j_m l_m} \right) \average{ x_i x_j \prod_{m=1}^n  x_{i_m} x_{j_m} x_{k_m} x_{l_m}}_0.
\end{equation}

One can then use the Wick theorem to express this quantity as a sum over pairings of $\mc{I}_{4n+2}$, but once again
it is easier to represent the pairings graphically. As before, for each $m=1,\ldots,n$, we draw one copy of Fig. \ref{fig:feynman1} (b), 
i.e., an interaction line with four dangling \emph{half-edges} labeled $i,j,k,l$.
We can then 
number each interaction line as $1,\ldots,n$ and indicate this by adding an appropriate subscript to the labels $i,j,k,l$ associated to this vertex. 
Now we also draw two additional freely floating half-edges with labels $i$ and $j$. We can view the half-edges 
as terminating in a vertex indicated by a dot (which will distinguish these diagrams from the so-called `truncated' 
diagrams that appear later on), while the other end of the half-edge is available for linking.
The $4n+2$ half-edges $\{i,j,i_1,\ldots,l_n\}$, each with a unique label, 
represent the set on which we consider pairings. We depict a pairing by linking the paired half-edges with a bare 
propagator. The resulting figure is a labeled Feynman diagram of order $n$. 
An example of order 2 is depicted in Fig. \ref{fig:linkingExampleG}.

\begin{figure}[h]
  \begin{center}
    \includegraphics[width=0.3\textwidth]{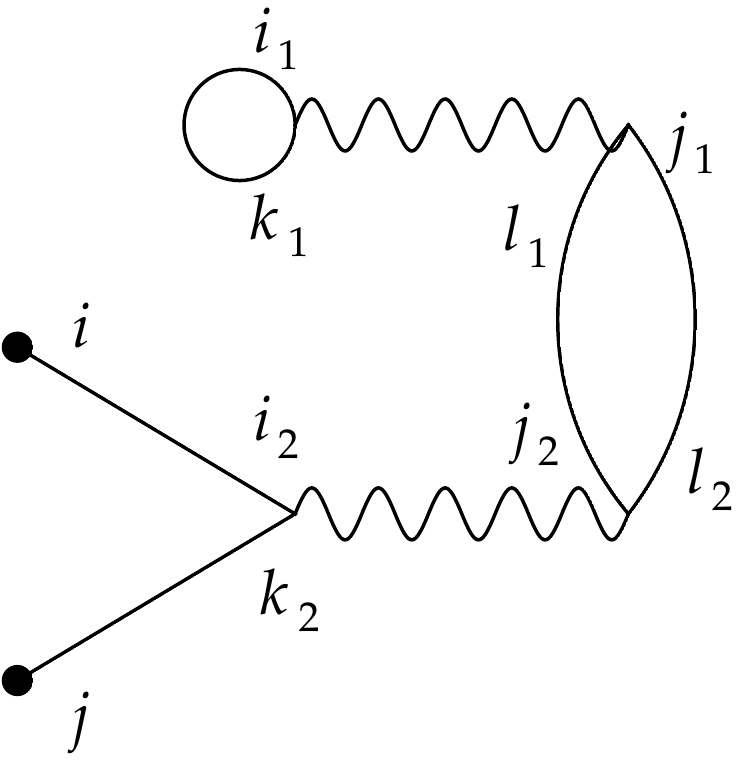}
  \end{center}
  \caption{A labeled closed Feynman diagram of order $2$.}
  \label{fig:linkingExampleG}
\end{figure}

The quantity associated via Wick's theorem with the pairing represented by such a diagram can then be computed by 
taking a formal product over all propagators and interaction lines of the associated quantities indicated in Fig. \ref{fig:feynman1} (a) 
and (b), respectively and then summing over the indices $i_1,j_1,k_1,l_1,\ldots,i_n,j_n,k_n,l_n$. Importantly we \emph{do not} sum 
over the indices $i,j$, as these specify the fixed entry of the Green's function $G_{ij}$ that we are computing via expansion. 
For the example depicted in
Fig. \ref{fig:linkingExampleG}, 
this procedure yields the sum

\begin{multline}
\sum_{i_1,j_1,k_1,l_1,i_2,j_2,k_2,l_2} v_{i_1 j_1} \delta_{i_1 k_1} \delta_{j_1 l_1} v_{i_2 j_2} \delta_{i_2 k_2} \delta_{j_2 l_2} 
G^0_{i_1 k_1} G^0_{j_1 l_2} G^0_{l_1 j_2} G^0_{i i_2} G^0_{j k_2}  \\ 
= \sum_{i_1,j_1,i_2,j_2} 
v_{i_1 j_1} v_{i_2 j_2} 
G^0_{i i_2} G^0_{j i_2} G^0_{i_1 i_1} G^0_{j_1 j_2} G^0_{j_1 j_2}.
\end{multline}

In summary, we can graphically represented the sum over pairings furnished by Wick's theorem as a sum 
over such diagrams, which we call \emph{labeled Feynman diagrams of order $n$ with $2$ external vertices}. 
(Perhaps calling them diagrams with `external half-edges' would be more appropriate, but `external vertices' is the 
conventional terminology.)

One can similarly imagine the natural appearance of Feynman diagrams with $2m$ external vertices 
in the expansion of the $2m$-point propagator $\average{x_{p_1} \cdots x_{p_{2m}}}$.

We can define the \emph{(partially labeled) Feynman diagrams of order $n$ with $2m$ external vertices} to be 
the $\Gamma = (V,H_1,H_2,E,\Pi,\mc{E})$, where $V,H_1,H_2$ are as in the definition of closed diagrams, $E$ is the set 
of $2m$ external half-edges, $\Pi$ is a partition of $E \cup \bigcup_{v\in V} H(v)$ into $2n + m$ pairs of 
half-edges, and $\mc{E}$ is a labeling of the external half-edges only. More precisely, $\mathcal{E}$ is a bijection 
from the external half-edge set $E$ to the set
of symbols $\{p_1,\ldots,p_{2m}\}$. In the case $m=1$ we will instead adopt the convention $\mc{E} : E \ra \{i,j\}$.

Two partially labeled diagrams of order $2$ with $2$ external vertices are depicted in 
Fig.~\ref{fig:partiallyLabeled}. Notice that these diagrams are not isomorphic due to the distinction 
of the external half-edges $i,j$, though they would be isomorphic as `fully unlabeled' diagrams.

\begin{figure}[h]
  \begin{center}
    \includegraphics[width=0.5\textwidth]{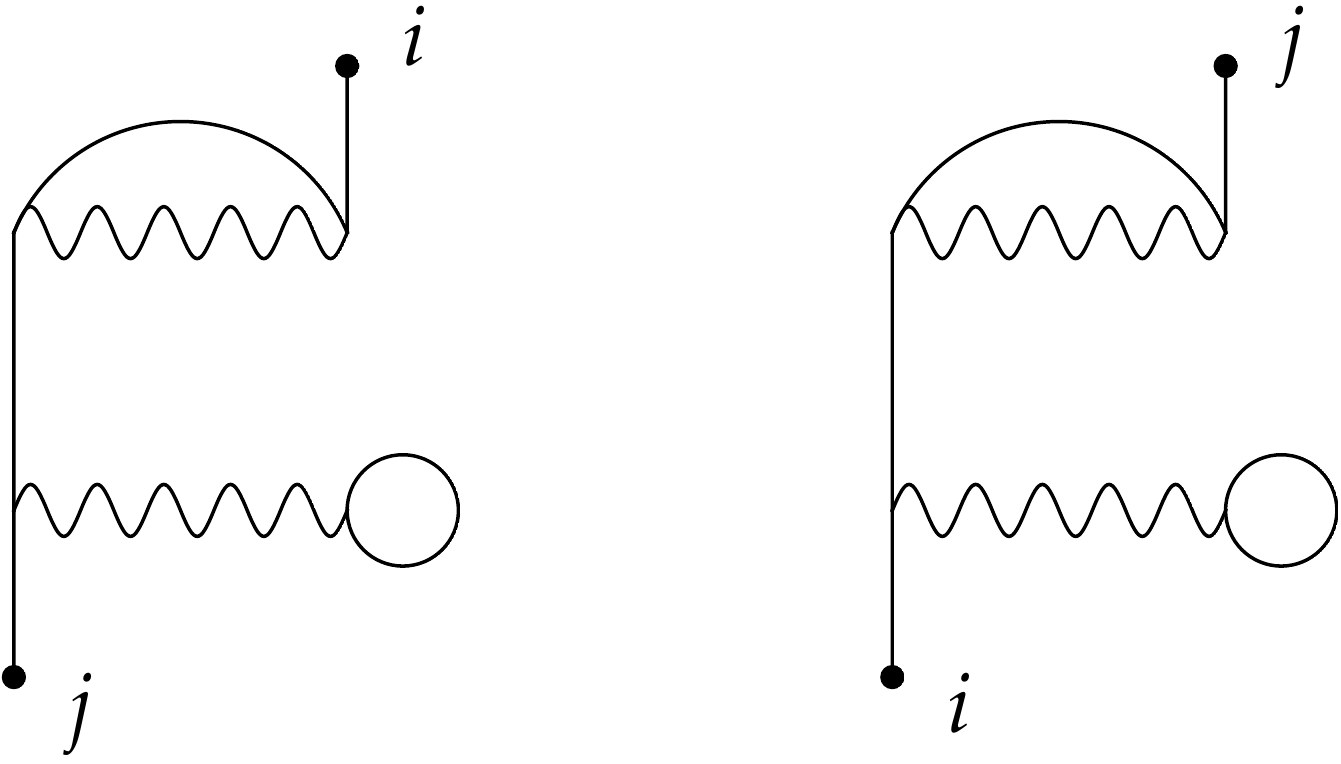}
  \end{center}
  \caption{Non-isomorphic partially labeled diagrams of order $2$ with $2$ external vertices.}
  \label{fig:partiallyLabeled}
\end{figure}

These diagrams can be additionally equipped with $\emph{internal}$ labelings $(\mathcal{V},\mathcal{H})$ to produce 
\emph{(fully) labeled Feynman diagrams of order $n$ with $2m$ external vertices}. Here $\mathcal{V},\mathcal{H}$ are 
defined as before.
More careful definitions of these classes of 
diagrams, as well as definitions of the notions of 
isomorphism for each, follow in the spirit of the analogous definitions for closed diagrams and are left to the reader.

Diagrams with external vertices will be understood to be partially labeled unless otherwise stated. The 
set of partially labeled diagrams (of any order) with $2m$ external vertices is denoted $\mc{F}_{2m}$. 
In the case $m=1$ we often refer to these diagrams as \emph{Green's function diagrams}.
Note that an unlabeled closed Feynman diagram can be viewed equivalently as 
a Feynman diagram with $0$ external vertices.

The group $\mathbf{R}_{n}$ acts naturally as before on internal labelings $(\mathcal{V},\mathcal{H})$ and induces 
a notion of automorphism for fully labeled diagrams with external vertices, as well as a symmetry factor $S_{\Gamma}$ 
defined to be the size of the automorphism group $\mathrm{Aut}(\Gamma)$ of a fully labeled diagram with external vertices 
(or, if $\Gamma$ is only partially labeled, the size of the automorphism group of any full labeling of $\Gamma$).

Moreover, each diagram with $2m$ external vertices yields a Feynman amplitude which is no longer a scalar, 
but in fact a $(2m)$-tensor, $F_{\Gamma}(p_1,\ldots,p_{2m})$ which can be computed as follows
\begin{enumerate}
  \item Assign a dummy index to each of the $4n$ internal half-edges as well as indices $p_1,\ldots,p_{2m}$ to each 
    of the external half-edges according to the labeling $\mc{E}$
    \item Each edge with half-edge indices $a,b$ yields a factor $G^{0}_{ab}$.
    \item Each interaction line with half-edge indices $a,b,c,d$ yields a factor
      $-v_{ab}\delta_{ac}\delta_{bd}$.
  \item Multiply all factors obtained via steps 2 and 3, and sum over all dummy indices from $1$ to $N$ to obtain a 
  tensor in the indices $p_1,\ldots,p_{2m}$.
  \end{enumerate}
For $\Gamma \in \mc{F}_{2}$, i.e., in the case $m=1$, we usually indicate the tensor
arguments by $i,j$ as in $F_{\Gamma}(i,j)$.

Following the same discussion in section~\ref{sec:feynmanZ}, we have
\begin{equation}
\label{eqn:Zcorr}
Z \average{ x_{p_1} \cdots x_{p_{2m}}} = Z_{0} \sum_{\Gamma\in \mf{F}_{2m}} \frac{F_{\Gamma}(p_1,\ldots,p_{2m})}{S_{\Gamma}},
\end{equation}
so in particular
\[
  Z G_{ij} = Z_{0} \sum_{\Gamma\in \mf{F}_{2}} \frac{F_{\Gamma}(i,j)}{S_{\Gamma}}.
  \label{}
\]

Denote by $\mf{F}^{\mr{c}}_{2m}\subset \mf{F}_{2m}$ the set of all diagrams
with $2m$ external vertices for which each connected component of the diagram
contains at least one external half-edge. It is easy to see that 
$\Gamma\in\mf{F}^{\mr{c}}_{2m}$ may have more than one connected component
when $m>1$.  However, when $m=1$, any diagram $\Gamma\in\mf{F}^{\mr{c}}_{2}$ has only
two external half-edges. Each internal vertex has
$4$ half-edges, and each connected component must contain an even
number of half-edges. This implies that $\Gamma$ must contain only one
connected component, so $\mf{F}^{\mr{c}}_{2}$ is in fact the subset of 
diagrams in $\mf{F}_{2}$ that are connected.

Theorem~\ref{thm:Gexpand} below shows, perhaps surprisingly, that the expansion 
for the \textit{correlator} $\average{ x_{p_1} \cdots x_{p_{2m}}}$ removes many
diagrams, and is therefore \emph{simpler} 
than the expansion of $Z \average{ x_{p_1} \cdots x_{p_{2m}}}$.
The combinatorial argument is similar in flavor to that
of the proof of Theorem~\ref{thm:logZexpand}.

\begin{theorem}[Linked cluster expansion for correlators]\label{thm:Gexpand}
  The asymptotic series expansion for $\average{ x_{p_1} \cdots x_{p_{2m}}}$, where $1\le p_1, \ldots, p_{2m} \le N$, is
  \begin{equation}
    \average{ x_{p_1} \cdots x_{p_{2m}}} = \sum_{\Gamma\in \mf{F}^{\mr{c}}_{2m}} \frac{F_{\Gamma}(p_1,\ldots,p_{2m})}{S_{\Gamma}}.
    \label{eqn:corrFeynman}
  \end{equation}
  In particular, the series for $G$ is 
  \begin{equation}
    G_{ij} = \sum_{\Gamma\in \mf{F}^{\mr{c}}_{2}} \frac{F_{\Gamma}(i,j)}{S_{\Gamma}}.
    \label{eqn:GFeynman}
  \end{equation}
\end{theorem}
\begin{proof}
  Any diagram $\Gamma\in \mf{F}_{2m}$ can be decomposed uniquely as
  $\Gamma=\Gamma' \cup \Gamma'' $, where $\Gamma' \in 
  \mf{F}^{\mr{c}}_{2m}$ and $\Gamma'' \in \mf{F}_{0}$, and we allow
  $\Gamma'' $ to be the empty diagram. Hence according to Eq.~\eqref{eqn:Zcorr}, 
  \[
  Z \average{ x_{p_1} \cdots x_{p_{2m}}}= Z_{0} \sum_{\Gamma' \in \mf{F}^{\mr{c}}_{2m}} \sum_{\Gamma''\in
  \mf{F}_{0}} \frac{F_{\Gamma' \cup \Gamma''}(p_1, \ldots, p_{2m})}{S_{\Gamma' \cup \Gamma''}}.
  \]
  Now for $\Gamma' \in 
  \mf{F}^{\mr{c}}_{2m}$ and $\Gamma'' \in \mf{F}_{0}$, 
  \[
  F_{\Gamma' \cup \Gamma''}(p_1, \ldots, p_{2m}) = F_{\Gamma'}(p_1, \ldots, p_{2m}) F_{\Gamma''}.
  \]
  Also  $\Gamma'$ and $\Gamma''$ have different numbers of external vertices, so 
   $\mathrm{Aut}(\Gamma' \cup \Gamma'') = \mathrm{Aut}(\Gamma') \times 
  \mathrm{Aut}(\Gamma'')$, and consequently
  \[
  S_{\Gamma} = S_{\Gamma'} S_{\Gamma''}.
  \]
  Hence
  \begin{eqnarray*}
  Z \average{ x_{p_1} \cdots x_{p_{2m}}} &= & Z_{0} \sum_{\Gamma' \in \mf{F}^{\mr{c}}_{2m}} \sum_{\Gamma''\in
  \mf{F}_{0}} \frac{F_{\Gamma'}(p_1, \ldots, p_{2m}) F_{\Gamma''}}{S_{\Gamma'} S_{\Gamma''}} \\
  &= & 
  Z_{0} \left( \sum_{\Gamma'\in \mf{F}^{\mr{c}}_{2m}} 
  \frac{F_{\Gamma'}(p_1,\ldots,p_{2m})}{S_{\Gamma'}}
  \right)
  \left( \sum_{\Gamma''\in \mf{F}_{0}} 
  \frac{F_{\Gamma''}}{S_{\Gamma''}}
  \right) \\
  & = & Z \sum_{\Gamma'\in \mf{F}^{\mr{c}}_{2m}} 
  \frac{F_{\Gamma'}(p_1,\ldots,p_{2m})}{S_{\Gamma'}},
  \end{eqnarray*}
  where the last equality follows from Theorem \ref{thm:Zexpand}. Dividing by $Z$ completes the proof. 
  \end{proof}

We now discuss the first few terms of the expansion for the Green's function $G$.
The zeroth order expansion for $G_{ij}$ is
$G_{ij}^{0}$.  Fig.~\ref{fig:feynmanG1} depicts the Feynman diagrams for
the first-order contribution to $G_{ij}$, which amounts to the expression
\begin{equation}
  -\frac{1}{2} \sum_{k,l} \left(v_{kl} G^{0}_{ik} G^{0}_{jk} G^{0}_{ll}\right) 
  -\sum_{k,l} \left(v_{kl} G^{0}_{ik} G^{0}_{jl} G^{0}_{kl}\right).
  \label{eqn:Gfirst}
\end{equation}
Note how the symmetry factor of these diagrams is affected by the 
labeling of the external half-edges.
\begin{figure}[h]
  \begin{center}
    \includegraphics[width=0.6\textwidth]{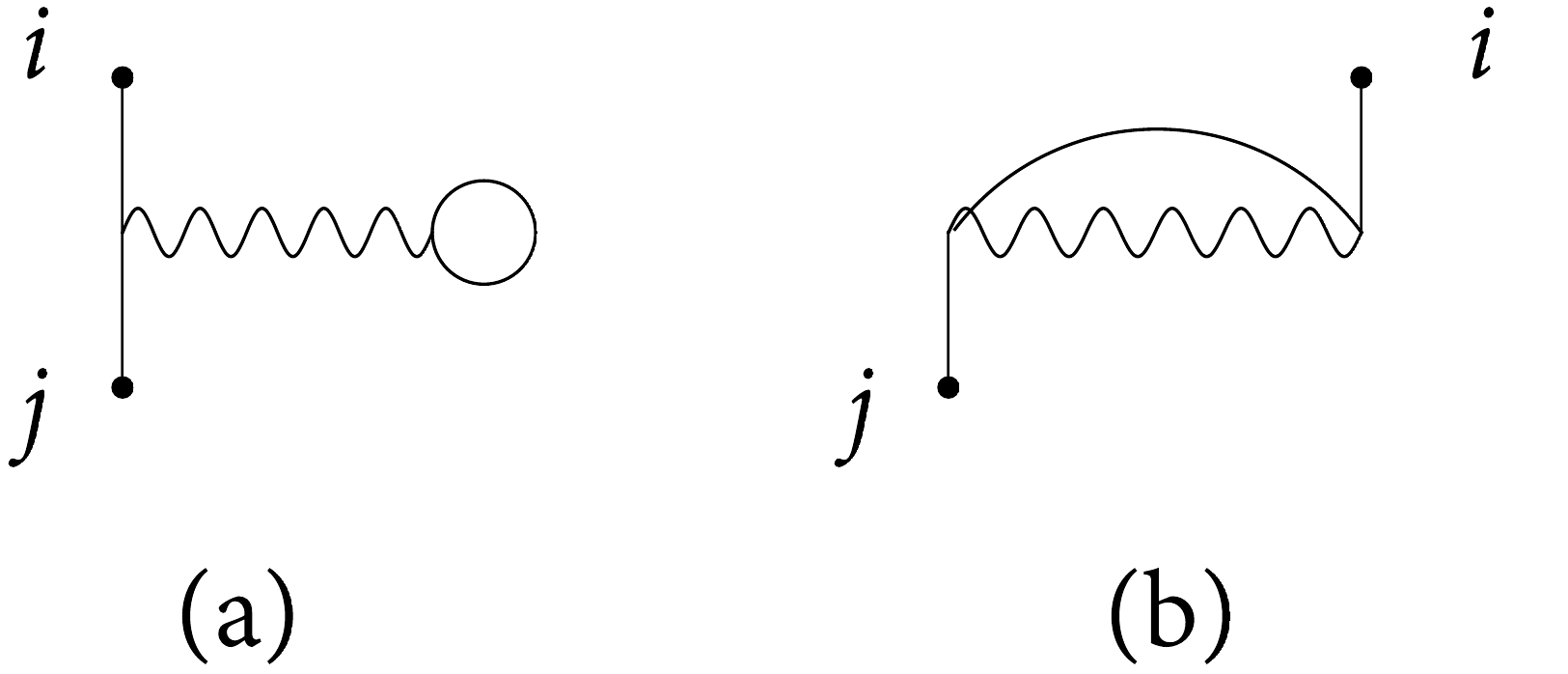}
  \end{center}
  \caption{First-order expansion for $G_{ij}$.}
  \label{fig:feynmanG1}
\end{figure}

Fig.~\ref{fig:feynmanG2}
depicts the Feynman diagrams for the second-order contribution to
$G_{ij}$. These diagrams can be systematically obtained from the 
free energy diagrams of Fig.~\ref{fig:feynmanLink2} (b) and (c) by
cutting a propagator line to yield two external half-edges and then 
listing the non-isomorphic ways of labeling of these external half-edges. Note that
all terms contain only one connected component due to
Theorem~\ref{thm:Gexpand}. For
simplicity we omit the resulting formula for the second-order contribution to
$G_{ij}$.

\begin{figure}[h]
  \begin{center}
    \includegraphics[width=0.9\textwidth]{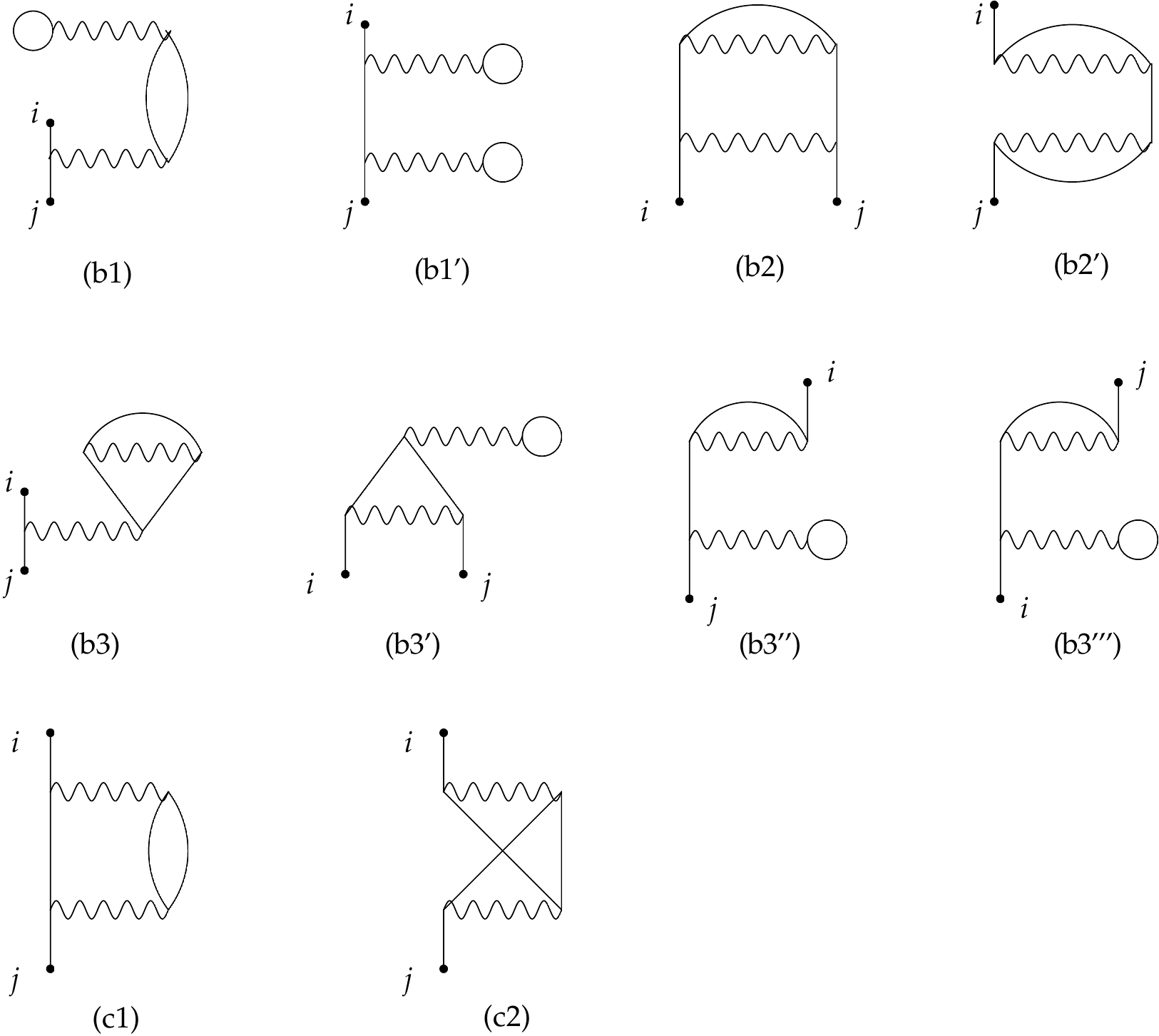}
  \end{center}
  \caption{Second-order expansion for $G_{ij}$. The lettering is obtained from that of the free energy
  diagrams in Fig.~\ref{fig:feynmanLink2} (b),(c), from which the Green's function diagrams may 
  be obtained by cutting lines.}
  \label{fig:feynmanG2}
\end{figure}

\subsection{Why we do not use fully unlabeled Green's function diagrams}
Why not reduce the redundancy of diagrams by considering a notion of fully unlabeled 
Green's function diagrams? One reason is that the notion of symmetry factor would be different, 
yielding an unpleasant extra factor in Theorem \ref{thm:Gexpand}. 
Moreover, in the development of the bold diagrammatic expansion of section \ref{sec:feynmanBold}, 
we will consider an operation in which propagator lines are replaced by Green's function diagrams. 
Since different orientations of such an `insertion' might yield different topologies of the resulting 
diagram, it is good to keep track of non-isomorphic external labelings separately.

Finally, by 
retaining an external labeling, there is a clearer interpretation of each diagram as a matrix 
yielded by contracting out internal indices. Note carefully, however, that the Feynman 
amplitude of a non-symmetric 
diagram, i.e., a diagram whose isomorphism class is changed by a relabeling of the 
external vertices, is in general a non-symmetric matrix. By contrast, $G$ is symmetric. 
Therefore any reasonable truncation of the expansion of Theorem \ref{thm:Gexpand} 
should not include any non-symmetric diagram without also including all diagrams 
obtained by different external labelings.

\subsection{Feynman rules for $\Sigma$}\label{sec:feynmanSigma}
The computation of $G$ by diagrammatic methods can be further simplified via 
the introduction of the notion of the \emph{self-energy}. This notion can be 
motivated diagrammatically as follows. Observe that diagrams such as (b1'), (b2'), 
and (b3'') in Fig.~\ref{fig:feynmanG2} are `redundant' in that they can be 
constructed by `stitching' first-order diagrams together at external vertices.
Such diagrams will be removed in 
the diagrammatic expansion for the self-energy matrix $\Sigma$, defined
as the difference between the inverse of $G$ and that of $G^{0}$ as
\begin{equation}
  \Sigma = \left(G^{0}\right)^{-1} - G^{-1}.
  \label{eqn:Sigmadef}
\end{equation}
Observe that once $\Sigma$ is known, $G$ can be computed simply via Eq.~\eqref{eqn:Sigmadef}.

However, Eq.~\eqref{eqn:Sigmadef} does not clarify the diagrammatic motivation for the 
self-energy. Note that the definition of the self-energy matrix in
Eq.~\eqref{eqn:Sigmadef} is equivalent to
\begin{equation}
  G = G^{0} + G^{0}\Sigma G,
  \label{eqn:dyson}
\end{equation}
which is called the \textit{Dyson equation}. By plugging the formula for $G$ specified by 
the Dyson equation back into the right-hand side of Eq.~\eqref{eqn:dyson} and then 
repeating this procedure \textit{ad infinitum}, one obtains the formal equation 
\begin{equation}
  G = G^{0} + G^{0} \Sigma G^{0} + G^{0} \Sigma G^{0} \Sigma
  G^{0} + \cdots,
  \label{eqn:dysonFormal}
\end{equation}
which suggests a diagrammatic interpretation for $\Sigma$. To wit, in order to avoid 
counting the same Green's function diagram twice in the right-hand side of 
Eq.~\eqref{eqn:dysonFormal}, $\Sigma$ should only
include those diagrams that cannot be separated into two disconnected
components when removing one bare propagator line.
In physics terminology, these are called the 
\textit{one-particle irreducible} (1PI) diagrams.

\begin{figure}[h]
  \begin{center}
    \includegraphics[width=0.65\textwidth]{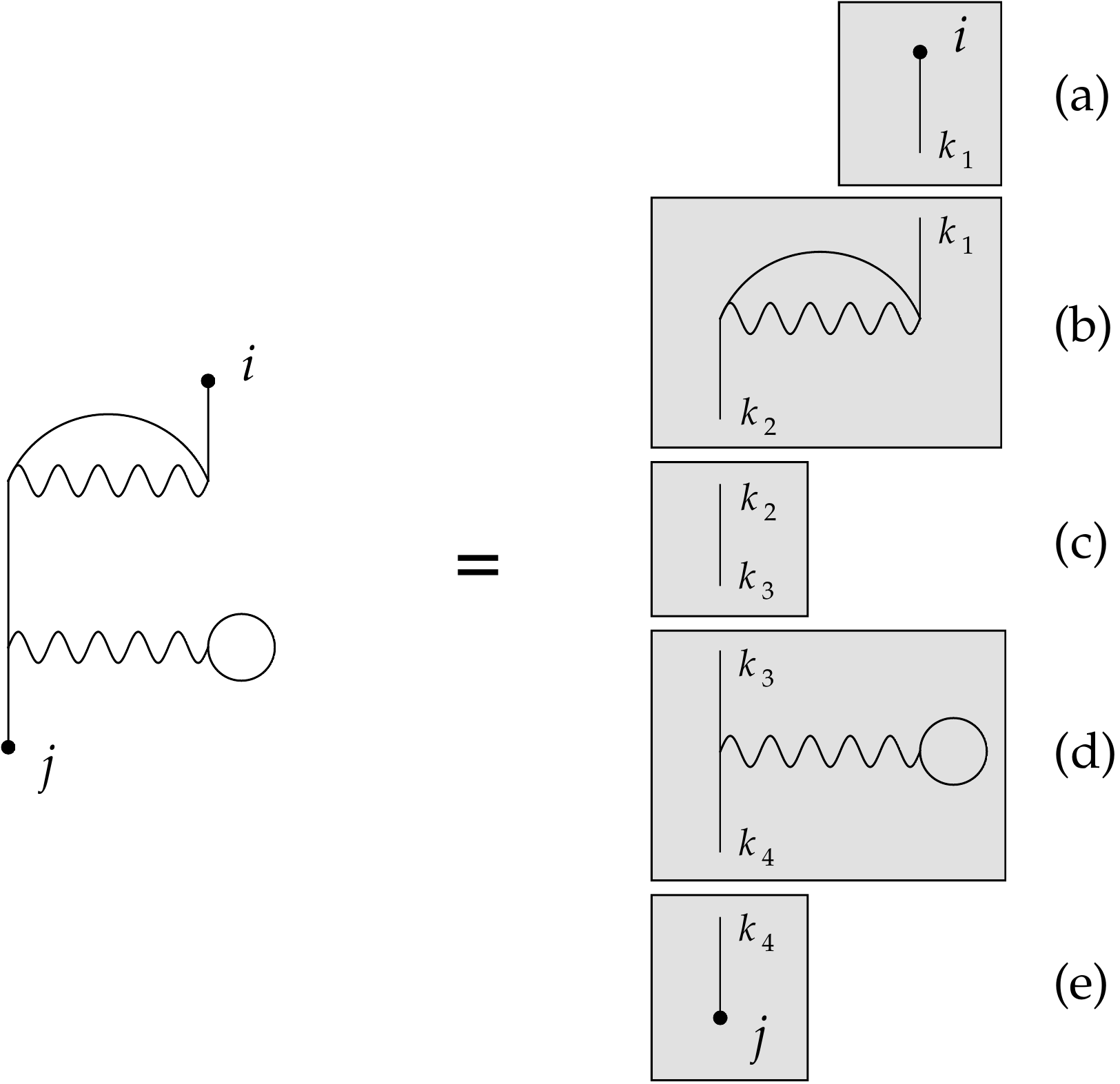}
    \vspace{2 mm}
  \end{center}
  \caption{Decomposing a Green's function diagram into truncated 1PI diagrams 
  and bare propagators.}
  \label{fig:stitching}
\end{figure}

We must be careful about what exactly is meant by such a diagram. 
We want to be able to produce Green's function diagrams by stitching together 
1PI diagrams via a bare propagator line $G^0$, as depicted in Fig.~\ref{fig:stitching}. 
In order to avoid double-counting the propagators at each `stitch,' our self-energy 
diagrams should \emph{not} include a contribution from the propagator where the 
stitch is made. In the example shown in Fig.~\ref{fig:stitching}, we write the matrix 
represented by the diagram on the left-hand side as a product (a)(b)(c)(d)(e) of 
matrices represented by the diagrams on the right-hand side. Here (a), (c), and (e) 
simply represent the propagator $G^0$. Diagrams (b) and (d) are the self-energy diagrams
representing the matrices with $(k_1,k_2)$ entry given by $ v_{k_1 k_2} G^{0}_{k_1 k_2}$ and 
$(k_3, k_4)$ entry given by $\delta_{k_3 k_4} v_{k_3 k_3} G^0_{k_3 k_3}$, respectively.
Since these diagrams are like 
Green's function diagrams, except missing the external propagator contributions, we 
refer to them as \emph{truncated Green's function diagrams}.

\begin{definition}
\label{def:truncatedDiagram}
A truncated Green's function diagram $\Gamma$ is obtained from a Green's 
function diagram $\Gamma'$. The internal half-edges of $\Gamma'$ paired with the 
external half-edges of $\Gamma'$ labeled $i$ and $j$ are referred to as the 
external half-edges of $\Gamma$ and are labeled $i$ and $j$,
respectively. The 1PI diagrams 
are the truncated Green's function diagrams that cannot be disconnected by the removal 
of a single bare propagator line. The set of all truncated Green's function
diagrams is denoted by $\mf{F}^{\mr{c,t}}_{2}$, and the set 
of all 1PI diagrams is denoted by $\mf{F}^{\mr{1PI}}_{2}$. 
The diagrams in 
$\mf{F}^{\mr{1PI}}_{2}$ are alternatively referred to as self-energy diagrams.
\end{definition}

Analogously one can define $\mf{F}^{\mr{c,t}}_{2m}$ and
$\mf{F}^{\mr{1PI}}_{2m}$ for $m>1$, 
but we will not make use of such notions.

As a data structure, a truncated Green's function diagram is really equivalent 
to its `parent' Green's function diagram, but the interpretation is different, and 
we visually distinguish the truncated diagrams from their counterparts by removing the 
dot at the external vertex. In addition, a truncated Green's function diagram has 
a different notion of (matrix-valued) Feynman amplitude $F_{\Gamma}(i,j)$, computed as follows:
\begin{enumerate}
    \item Assign a dummy index to each of the $4n-2$ internal half-edges as well as indices $i,j$ to each 
    of the external half-edges according to the labeling furnished by Definition \ref{def:truncatedDiagram}.
    \item Each \emph{internal} edge with half-edge indices $a,b$ yields a factor $G^{0}_{ab}$.
    \item Each interaction line with half-edge indices $a,b,c,d$ yields a factor
      $-v_{ab}\delta_{ac}\delta_{bd}$.
  \item Multiply all factors obtained via steps 2 and 3, and sum over all dummy indices from $1$ to $N$ to obtain a 
  matrix in the indices $i,j$.
\end{enumerate}
However, the symmetry factor $S_{\Gamma}$ of a truncated Green's function diagram 
is unchanged from that of the underlying Green's function diagram.

We may further introduce the concept of 
two-particle irreducible (2PI) Green's function diagrams as the subset of
diagrams in $\mf{F}_{2}^{\mathrm{1PI}}$ that cannot be disconnected by 
the removal of any two edges. The set of all such diagrams is
denoted by $\mf{F}_{2}^{\mathrm{2PI}}$. The 2PI diagrams will be used
to define the bold diagrams in section~\ref{sec:feynmanBold}.

The first-order self-energy diagrams are depicted in Fig.~\ref{fig:feynmanSigma1}. The only
difference from Fig.~\ref{fig:feynmanG1} is that the external vertices
are removed to produce truncated diagrams. The second-order self-energy 
diagrams are shown in Fig.~\ref{fig:feynmanSigma1}. Note that the Green's function 
diagrams $(\mr{b1}')$, $(\mr{b2}')$,
$(\mr{b3}'')$, and $(\mr{b3}''')$ in Fig.~\ref{fig:feynmanG2}---after removing the external vertices to 
yield self-energy diagrams---are not 1PI, hence not self-energy diagrams.

\begin{figure}[h]
  \begin{center}
    \includegraphics[width=0.5\textwidth]{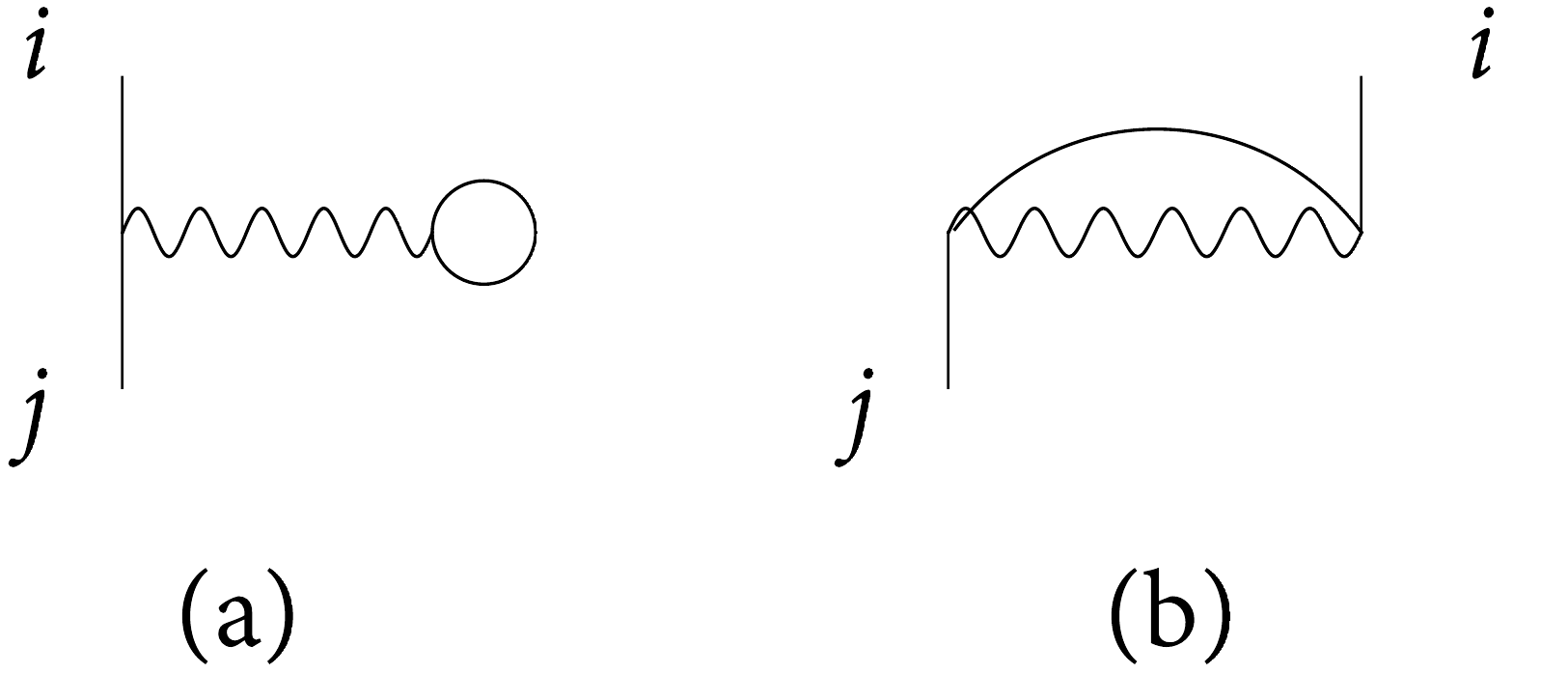}
  \end{center}
  \caption{First-order diagrams for $\Sigma_{ij}$.}
  \label{fig:feynmanSigma1}
\end{figure}

\begin{figure}[h]
  \begin{center}
    \includegraphics[width=0.5\textwidth]{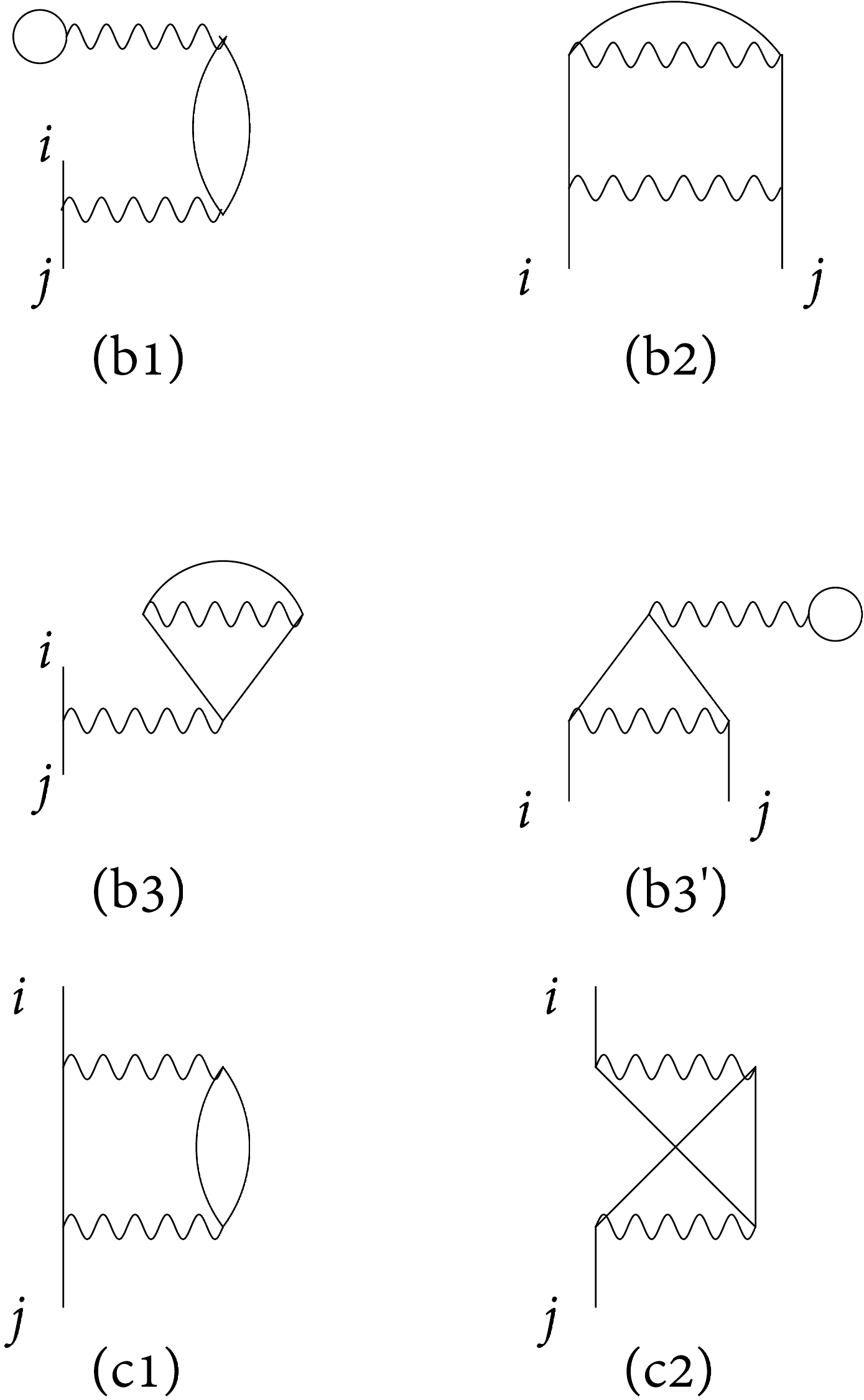}
  \end{center}
  \caption{Second-order diagrams for $\Sigma_{ij}$. The labels correspond
  to those of the `parent' Green's function diagrams in Fig.~\ref{fig:feynmanG2}.}
  \label{fig:feynmanSigma2}
\end{figure}

\begin{theorem}\label{thm:Sigmaexpand}
  The asymptotic series expansion for $\Sigma_{ij}$, where $1\le i,j \le N$, is
  \begin{equation}
    \Sigma_{ij} = \sum_{\Gamma\in \mf{F}^{\mr{1PI}}_{2}} \frac{F_{\Gamma}(i,j)}{S_{\Gamma}}.
    \label{eqn:SigmaFeynman}
  \end{equation}
\end{theorem}

\begin{proof}
One can think of the Dyson equation \eqref{eqn:dyson} 
as an equation of formal power series, where $G$ and $\Sigma$ indicate the asymptotic 
series expansions of $G$ and $\Sigma$. (Recall that we may think of our diagrammatic 
expansions as power series in a parameter $\lambda$ that scales the interaction strength. 
It is not hard to see directly from the definition of $\Sigma$, as for all other quantities we 
consider, that an asymptotic series in $\lambda$ exists in the first place.) Now the series 
for $G$ is known from Theorem \ref{thm:Gexpand}, and we claim that the series for 
$\Sigma$ is the unique formal power series satisfying Eq.~\eqref{eqn:dyson} as an equation 
of formal power series. Indeed, if some power series $\Sigma$ satisfies 
Eq.~\eqref{eqn:dyson}, then $\Sigma$ satisfies Eq.~\eqref{eqn:Sigmadef} as well 
(as an equation of formal power series), but inverses are unique in the ring of formal 
power series, so $\Sigma$ satisfying Eq.~\eqref{eqn:Sigmadef} is uniquely determined.

Thus all we need to show is that Eq.~\eqref{eqn:dyson} holds when we plug in the series 
for $\Sigma$ from
Eq.~\eqref{eqn:SigmaFeynman}. To this end, write, via Theorem \ref{thm:Gexpand}, 
\begin{equation}
\label{eqn:GexpandRepeat}
G = \sum_{\Gamma\in \mf{F}^{\mr{c}}_{2}} \frac{F_{\Gamma}}{S_{\Gamma}},
\end{equation}
where $F_{\Gamma}$ appearing in the summand is a matrix.
Now every $\Gamma \in  \mf{F}^{\mr{c}}_{2}$ that is of order greater than $1$ 
can be decomposed uniquely into a bare propagator line at the external half-edge labeled $i$,  
a self-energy diagram $\Gamma'$ connected to this propagator line at one external half-edge, 
and another Green's function diagram $\Gamma''$ connected to $\Gamma'$ at its other 
external half-edge. (This fact should be clear graphically, though 
a more careful proof is left to the reader.) For example, in Fig.~\ref{fig:stitching}, $\Gamma'$ 
corresponds to (b) and $\Gamma''$ corresponds to (c)(d)(e).

Moreover, we have the equality 
(of matrices) 
\[
F_{\Gamma} = G^0 F_{\Gamma'} F_{\Gamma''}.
\]
Also, due to the fact that $\Gamma$ distinguishes the labels $i,j$, we have that
$S_{\Gamma} = S_{\Gamma'} S_{\Gamma''}$. Indeed, any automorphism of (a fully labeled version of) 
$\Gamma$ must fix the label $i$ and $j$ of the external half-edges, as well as the labels of the internal half-edges 
connected directly to them. Then such an automorphism can only permute labels within the component 
$\Gamma'$; otherwise, the automorphism would induce a graph automorphism of $\Gamma'$ with 
another subgraph of $\Gamma$ containing the external half-edge labeled by $i$ as well as some half-edge 
in $\Gamma''$, which would consequently fail to be one-particle irreducible, contradicting the one-particle 
irreducibility of $\Gamma'$.

Thus from Eq.~\eqref{eqn:GexpandRepeat} we obtain the equality of power series
\[
G = G^0 + G^0 \sum_{\Gamma'\in \mf{F}^{\mr{1PI}}_{2}} 
\sum_{\Gamma''\in \mf{F}^{\mr{c}}_{2}} 
\frac{F_{\Gamma'}}{S_{\Gamma'}} \frac{F_{\Gamma''}}{S_{\Gamma''}} = G^0 + G^0 \left[\sum_{\Gamma'\in \mf{F}^{\mr{1PI}}_{2}} 
\frac{F_{\Gamma'}}{S_{\Gamma'}} \right] G,
\]
as was to be shown.
\end{proof}

\section{Bold diagrams}\label{sec:feynmanBold}

It turns out further redundancy can be removed from the diagrammatic series 
for the self-energy by consideration of the so-called bold diagrams. 
Note that so far all diagrammatic series are defined
using the non-interacting Green's function $G^{0}$ (alternatively the bare propagator),
which can be viewed as the non-interacting 
counterpart to the \emph{interacting} Green's function $G$ (alternatively the 
`dressed' or `renormalized' propagator). What if
we replace all of the $G^{0}$ in our self-energy expansion by $G$? 
Accordingly let us introduce the convention of
a \emph{doubled line} (also called a \emph{bold line}) to denote $G$.
After replacing all thin lines by bold lines in a diagram, the
resulting diagram is called a \textit{bold diagram}. (Topologically the diagram is 
not altered by this procedure, but the interpretation and Feynman amplitude, 
as well as our visual representation of the diagram, are changed.)
A bold diagram can be
understood as a shorthand for an infinite sum of bare diagrams by swapping each  
bold line out for the bare diagrammatic expansion of $G$. 
An example of a bold diagram and its representation as a sum of bare diagrams is 
provided in Fig.~\ref{fig:feynmanSigmaG1}. Note that this representation is 
considered as an equality \emph{only at the level of formal power series}.

\begin{figure}[h]
  \begin{center}
    \vspace{2mm}
    \includegraphics[width=0.95\textwidth]{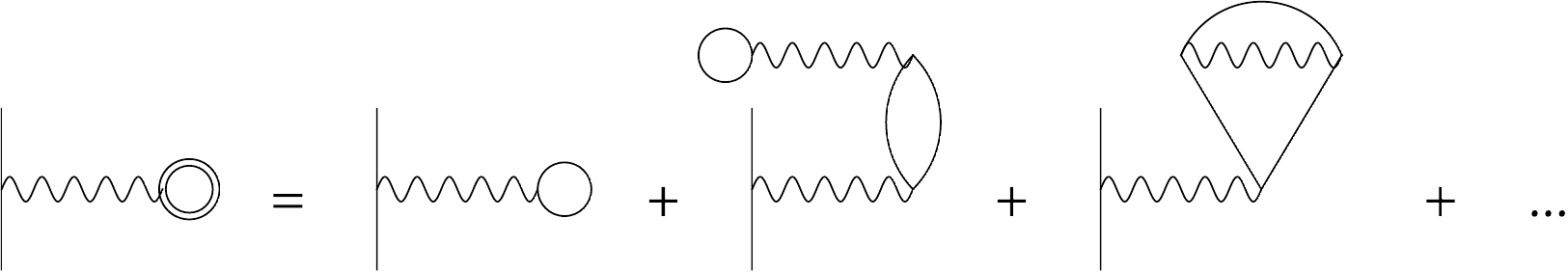}
    \vspace{4mm}
  \end{center}
  \caption{A bold self-energy diagram (the dumbbell, or Hartree, diagram), together with
  its expansion as a series of bare diagrams. Here we omit the labels $i,j$ in the 
  self-energy diagram, leaving a dangling half-edge (without a dot) to indicate
  their existence.}
  \label{fig:feynmanSigmaG1}
\end{figure}

If one were to replace all self-energy diagrams by their bold counterparts, the 
resulting bold diagram expansion would overcount many of the original bare self-energy 
diagrams. Indeed, notice that the second and third terms on the right-hand side of 
Fig.~\ref{fig:feynmanSigmaG1} account for the bare self-energy diagrams (b1) and (b3) 
of Fig.~\ref{fig:feynmanSigma2}. The bold versions of (b1) and (b3) would also count 
these terms, so as a result these contributions would be double-counted. Therefore if we 
can concoct a successful bold diagrammatic expansion for the self-energy, it should 
involve `fewer' diagrams than the bare expansion. From a certain perspective, passage 
to the bold diagrams can then be thought of as
a means to further economize on diagrammatic bookkeeping.

Which self-energy diagrams should be left out of the bold expansion? Notice that the 
disqualifying feature of diagrams (b1) and (b3) of Fig.~\ref{fig:feynmanSigma2} is that 
they contain Green's function diagram insertions---for short, simply 
\emph{Green's function insertions} or even \emph{insertions} when the context is clear.
In other words, we can disconnect each of these diagrams 
into two separate diagrams by cutting two propagator lines. The resulting component 
\emph{not} containing the external half-edges of the original diagram is itself a 
Green's function diagram with external half-edges at the cut locations. (For now 
we are being a bit casual about the distinction 
 between truncated and non-truncated diagrams because there is essentially no 
 topological difference.) In the
component that does contain 
the external half-edges of the original diagram, the two half-edges that have been left 
dangling due to the cuts can be sewn together 
with a bold line to yield the `parent' bold diagram. The insertion procedure yielding 
diagram (b3) Fig.~\ref{fig:feynmanSigma2} is 
depicted in Fig.~\ref{fig:SigmaInsertion}.

In general a bare self-energy diagram 
may contain many such Green's function insertions, possibly viewed as being nested 
within one another. 
However, it will soon pay to introduce a notion of a maximal Green's function insertion, 
or \emph{maximal insertion} for short. This is a Green's function insertion that is not contained within 
any other insertion. Then we will find that any bare self-energy diagram can be 
represented uniquely via its set of maximal insertions.

\begin{figure}[h]
  \begin{center}
    \includegraphics[width=0.75\textwidth]{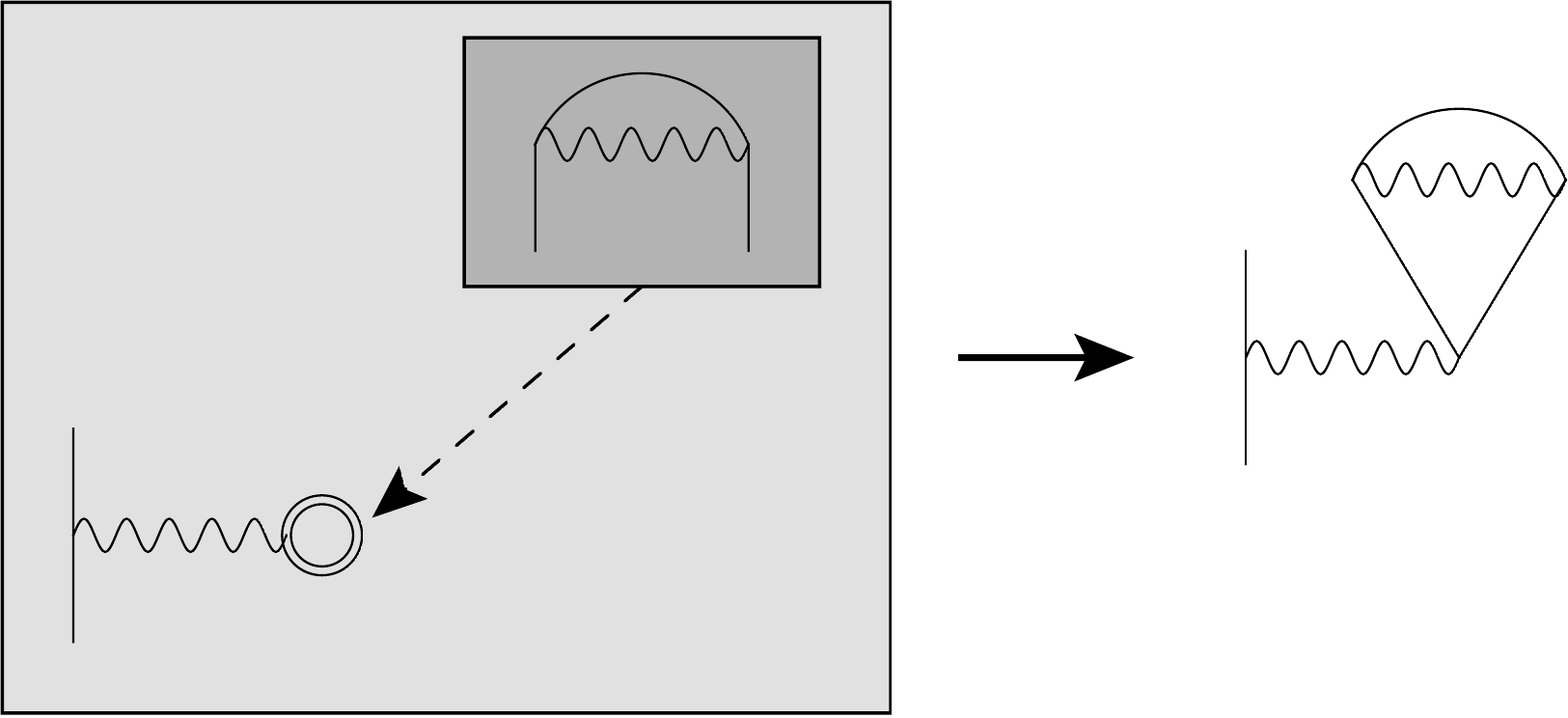}
  \end{center}
  \caption{Green's function insertion yielding diagram (b3) of Fig.~\ref{fig:feynmanSigma2}.}
  \label{fig:SigmaInsertion}
\end{figure}

Note that a diagram admits a Green's function insertion if and only if it 
can be disconnected by removing two propagator lines. Then the candidates 
for the bold self-energy diagrams are the self-energy 
diagrams with Green's function insertions; namely, the 2PI self-energy 
diagrams introduced earlier, though now considered with bold lines. 
This set will be denoted $\mf{F}_{2}^{\mathrm{2PI}}$ as before. We distinguish 
diagrams a bold via the notation for the Feynman amplitude as $\mathbf{F}_{\Gamma}$, 
as opposed to $F_{\Gamma}$.
We call diagrams in $\mf{F}_{2}^{\mathrm{2PI}}$ \textit{skeleton
diagrams} and diagrams in $\mf{F}_{2}^{\mathrm{1PI}}\backslash
\mf{F}_{2}^{\mathrm{2PI}}$ \emph{non-skeleton self-energy diagrams}, 
or simply \emph{non-skeleton diagrams} for short.

The idea of the bold diagrammatic expansion is to write 
\begin{equation}
\label{eqn:bold}
\Sigma_{ij} = \sum_{\Gamma_{\mr{s}} \in \mf{F}_2^{\mr{2PI}}} \frac{\mathbf{F}_{\Gamma_{\mr{s}}}(i,j)}{S_{\Gamma_{\mr{s}}}}. 
\end{equation} This equation must be interpreted rather carefully to
yield a rigorous statement. 
However, formally speaking for now, note that in order for the bold diagram expansion~\eqref{eqn:bold} to match 
the bare diagram expansion~\eqref{eqn:SigmaFeynman} for the self-energy, $S_{\Gamma_{s}}^{-1}$ should be the 
right guess of the pre-factor for the diagram $\Gamma_{s}$ in~\eqref{eqn:bold}. 
Indeed, if we formally substitute the bare expansion for the Green's function in for each bold propagator line  
of a bold diagram $\Gamma_s$, then the first term in the 
resulting expansion for the bold Feynman amplitude of $\Gamma_s$ will be the 
Feynman amplitude of $\Gamma_s$ interpreted as a \emph{bare} diagram, which should indeed be counted with the 
pre-factor $S_{\Gamma_{s}}^{-1}$ as in the bare expansion 
for the self-energy.
But then we have to establish that the rest of the bare
self-energy diagrams (i.e., those with non-skeleton topology) are counted with the appropriate pre-factors.
This turns out to be non-trivial and constitutes the major task of this section.
Our efforts culminate in Theorem \ref{thm:boldExpand} of section
\ref{sec:finallyBold} below, in which we give precise meaning to and
prove Eq.~\eqref{eqn:bold}.  We recommend readers to skip to
section~\ref{sec:finallyBold} for the applications of the bold
diagrammatic expansion first and then to return to the intervening details later.

\subsection{Skeleton decomposition}
Our first goal is to show that every self-energy diagram can be decomposed 
(uniquely, in some sense) as a skeleton diagram with Green's function insertions. 
We now turn to defining the notion of insertion more carefully.

\begin{definition}
Given a truncated Green's function diagram 
$\Gamma$, together with a half-edge pair $\{h_1,h_2\}$ in $\Gamma$\footnote{So 
$\{h_1,h_2\}$ is contained in the pairing $\Pi_{\Gamma}$ of half-edges associated 
with $\Gamma$.}
and another truncated Green's function diagram $\Gamma'$, 
the insertion of $\Gamma'$ into $\Gamma$ at $(h_1,h_2)$, denoted $\Gamma \oplus_{(h_1,h_2)} \Gamma'$ is defined to be the truncated 
Green's function diagram constructed by taking the collection of all vertices and half-edges (along with their 
pairings) from $\Gamma$ and $\Gamma'$, then defining a new half-edge pairing by 
removing $\{h_1,h_2\}$ and adding $\{h_1,e_1\}$ and $\{h_2,e_2\}$, where $e_1$ and $e_2$ are 
the external half-edges of $\Gamma'$ 
labeled $i$ and $j$, respectively.
\end{definition}

Notice that the ordering of $(h_1,h_2)$ in $\Gamma \oplus_{(h_1,h_2)} \Gamma'$ 
matters in this definition because it determines the orientation of the inserted diagram. Here the definition has 
also made use of the fact that 
truncated Green's function diagrams distinguish their external half-edges via the labels $i$ and $j$.

We can define a simultaneous insertion of truncated Green's function
diagrams $\Gamma^{(1)},\ldots,\Gamma^{(K)}$ along several edges of a
diagram, as follows:
\begin{definition}
  Let $\Gamma$ be a truncated Green's function diagram, and consider a collection of distinct
  half-edge pairs $\left\{h_1^{(k)},h_2^{(k)}\right\}$ for $k=1,\ldots,K$. 
  Let $\Gamma_0 = \Gamma$ and recursively define 
  $\Gamma_{k+1} := \Gamma_k \oplus_{(h_1^{(k+1)},h_2^{(k+1)})} \Gamma^{(k+1)}$ 
  for $k=0,\ldots,K-1$.
  Then the resulting $\Gamma_K$ is the insertion of
  $\Gamma^{(1)},\ldots,\Gamma^{(K)}$ into $\Gamma$ along 
  $(h_1^{(1)},h_2^{(1)},\ldots,h_1^{(K)},h_2^{(K)})$, denoted
  \[
  \Gamma \oplus_{(h_1^{(1)},h_2^{(1)},\ldots,h_1^{(K)},h_2^{(K)})} \left[\Gamma^{(1)},\ldots,\Gamma^{(K)} \right].
  \]
\end{definition}

Notice that the simultaneous insertion does not depend on the ordering of the $k$ half-edge pairs, though it does 
depend in general on the ordering of the half-edges within each pair.

\begin{definition}
  We say that a truncated Green's function diagram $\Gamma$ admits an insertion $\Gamma''$ at $(h_1,h_2)$ if it can be written as 
  $\Gamma' \oplus_{(h_1,h_2)} \Gamma''$, where $\{h_1,h_2\}$ is a pair in $\Gamma'$ and 
  $\Gamma''$ is a nonempty truncated Green's function diagram. (Note that $\Gamma$ 
  admits such an insertion if and only if $\Gamma$ can be disconnected by removing the half-edges $h_1$ and $h_2$.)
  We say that this 
  insertion is maximal if $\Gamma'$ does not in turn admit an insertion containing either of the half-edges 
  $h_1,h_2$.
\end{definition}

For example, consider in self-energy diagram of Fig.~\ref{fig:maxInsertions}, which 
admits two maximal insertions, shown in blue and red, respectively. (Note that each of the 
maximal insertions admits insertions itself, i.e., the overall diagram admits several 
insertions that are not maximal.) The remaining half-edges and interaction lines in 
the diagram (shown in black) form the `skeleton' of the diagram.

\begin{figure}[h]
  \begin{center}
    \includegraphics[width=0.6\textwidth]{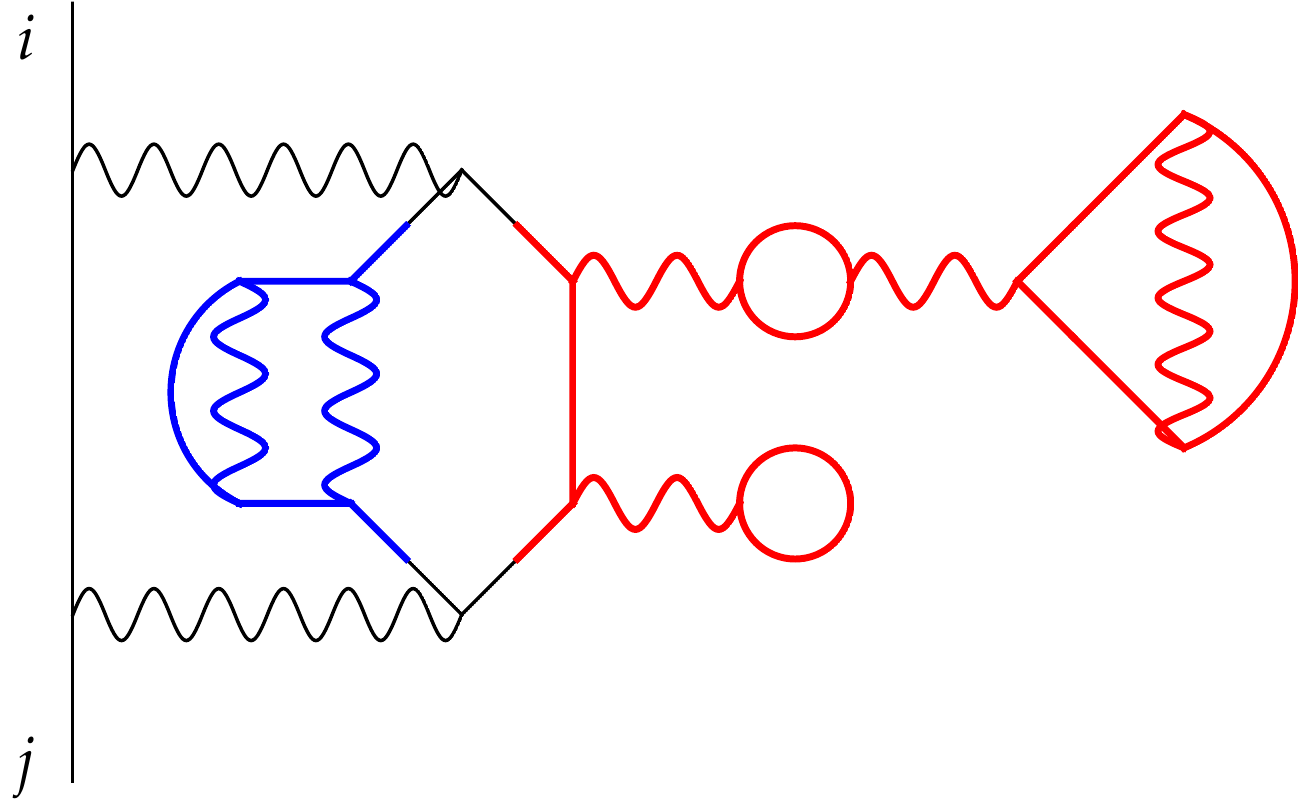}
  \end{center}
  \caption{A self-energy diagram with two maximal insertions, depicted in blue and red, respectively.}
  \label{fig:maxInsertions}
\end{figure}

The following result characterizes every self-energy diagram uniquely in terms of its 
maximal insertions and an underlying 
skeleton diagram (hence the name `skeleton') obtained by collapsing each of these 
insertions into a single propagator line. The proof is given in Appendix \ref{sec:appSkeleton}. 
It somewhat technical and may be skipped on first reading to avoid interrupting the 
flow of the developments that follow.

\begin{proposition}[Skeleton decomposition]
  \label{prop:selfenergycompose}
  Any diagram $\Gamma \in \mf{F}_2^{\mathrm{1PI}}$ can be 
  written as 
  \begin{equation}
    \label{eqn:skeletonInsertion}
    \Gamma = \Gamma_{\mr{s}} \oplus_{(h_1^{(1)},h_2^{(1)},\ldots,h_1^{(K)},h_2^{(K)})} \left[\Gamma^{(1)},\ldots,\Gamma^{(K)} \right],
  \end{equation}
 where $\Gamma_{\mr{s}} \in \mf{F}_2^{\mathrm{2PI}}$ and $\Gamma^{(k)} \in \mf{F}_2^{\mr{c,t}}$ for $k=1,\ldots,K$.
 Moreover such a decomposition is unique up to the external labeling of the $\Gamma^{(k)}$ and the ordering of the pair 
 $(h_1^{(k)},h_2^{(k)})$ for each fixed $k$, and the $\Gamma^{(k)}$ are precisely the maximal insertions 
 admitted by $\Gamma$ 
 (ignoring distinction of insertions based on
  external labelings).
\end{proposition}
\begin{remark}
\label{rem:uniqueness}
Here we record some comments on the meaning of the uniqueness result of 
Proposition \ref{prop:selfenergycompose}.
It is purely an artifact of our \textrm{`$\,\oplus$'} notation for insertions, which privileges 
an ordering of each pair $\{h_1^{(k)},h_2^{(k)}\}$ of half-edges as in \eqref{eqn:skeletonInsertion}, 
that one could just as well write $\Gamma$ in the form of 
\eqref{eqn:skeletonInsertion} by exchanging the roles of $h_1^{(k)}$ and $h_2^{(k)}$ and 
permuting the external labels of the insertion $\Gamma^{(k)}$. The statement 
is then that the decomposition in Proposition \ref{prop:selfenergycompose} is unique 
up to this redundancy, which is resolved by fixing an external labeling for each of the 
maximal insertions admitted by $\Gamma$.

This sort of non-uniqueness (which is really just a notational artifact and reflects no 
interesting topological properties of a diagram), should be contrasted with a 
notion appearing later on, which is to be motivated in section \ref{sec:examples} and 
fully sharpened in section \ref{sec:ways}. Indeed, we will be interested 
in the number of `ways' (in a sense to be clarified) of producing a diagram 
\emph{isomorphic} to $\Gamma \in \mf{F}_2^{\mathrm{1PI}}$ from 
its skeleton $\Gamma_{\mr{s}} \in \mf{F}_2^{\mathrm{2PI}}$ via Green's 
function insertions. By contrast, Proposition \ref{prop:selfenergycompose} 
above concerns the number of ways of producing the \emph{actual} diagram 
$\Gamma$ from its skeleton $\Gamma_{\mr{s}}$, stating that that there is 
in fact only one (up to the notational ambiguity we have discussed).

\end{remark}

\vspace{6mm}

Now we return to the task of developing a bold diagrammatic expansion for 
the self-energy. Proposition \ref{prop:selfenergycompose} tells us that each bare self-energy 
diagram can be constructed from a unique skeleton diagram via 
Green's function insertions. It is not hard to see that, conversely, the result of making 
insertions into a skeleton diagram is a 1PI diagram, i.e., a self-energy diagram. 
If we view skeleton diagrams as \emph{bold} 
diagrams, this implies that by summing over all (bold) skeleton diagrams (and then 
formally replacing each bold line with a sum over Green's function diagrams), 
we recover all of the bare self-energy diagrams.
However, there remains the question 
of whether these diagrams are counted appropriately. To answer this question we need to 
understand three items: (1) how many ways a given (isomorphism class of) self-energy diagram can 
be obtained via insertions from its underlying skeleton, (2) how to represent 
the automorphism groups (hence also symmetry factors) of self-energy diagrams in terms of the 
decomposition of Proposition \ref{prop:selfenergycompose}, and (3) the relation between 
items (1) and (2). These items will be addressed in Sections \ref{sec:ways}, \ref{sec:autSkeleton}, 
and \ref{sec:action}, respectively. First, however, to gain familiarity with what we are 
trying to prove, we 
discuss some motivating examples in section \ref{sec:examples}.

\subsection{Motivating examples}
\label{sec:examples}
Consider the non-skeleton diagram $\Gamma$ in
Fig.~\ref{fig:feynmanSigma_decompose} (a), for which we have $S_{\Gamma} = 2$.  It
can be uniquely decomposed into a skeleton diagram $\Gamma_{\mr{s}}$ in
Fig.~\ref{fig:feynmanSigma_decompose} (b) and the single maximal insertion in 
$S_{\Gamma_{\mr{g}}}$ shown in (c). Evidently 
$S_{\Gamma_{\mr{s}}}=2$ and $S_{\Gamma_{\mr{g}}}=1$. 
Roughly speaking (for now), there is only one `way' in which (c) can be 
inserted into (b) to produce a diagram isomorphic to (a), so we say that 
the \emph{redundancy factor} of $\Gamma$ is $1$ and write $r_{\Gamma}=1$.
This notion will be defined more carefully below. For now we mention that 
we \emph{do not} count separately the oppositely `oriented' insertions of $S_{\Gamma_{\mr{g}}}$
into $S_{\Gamma_{\mr{s}}}$ because $\Gamma_{\mr{g}}$
is \emph{symmetric}, 
i.e., its isomorphism 
class is unchanged by the exchange of its two external labels. Note, e.g., that all 
all diagrams in Fig.~\ref{fig:feynmanG2} are symmetric except $(\mr{b3}'')$ and $(\mr{b3}''')$.

\begin{figure}[h]
  \begin{center}
    \includegraphics[width=0.8\textwidth]{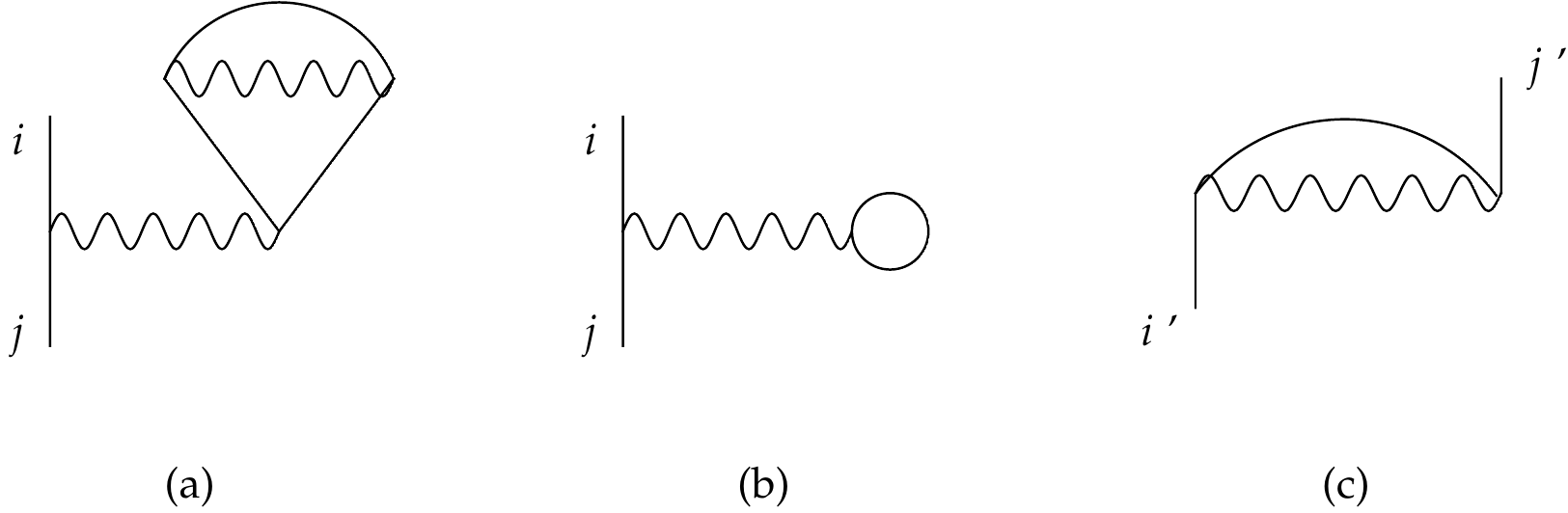}
  \end{center}
  \caption{Decomposition of a non-skeleton diagram $\Gamma$ in (a) into
  a skeleton diagram $\Gamma_{\mr{s}}$ in (b) and a truncated Green's function diagram
  $\Gamma_{\mr{g}}$ in (c).}
  \label{fig:feynmanSigma_decompose}
\end{figure}

In our proposed bold diagram expansion Eq.~\eqref{eqn:bold} for the self-energy, when we 
substitute the bare diagrammatic expansion for the Green's function in for 
each bold line, every bare self-energy diagram will be accounted for once 
for each of the $r_{\Gamma}$ `ways' that it can be produced from its skeleton via insertions, each time 
with a pre-factor equal to the reciprocal of the product of the symmetry 
factors of its skeleton and of all of its insertions. In other words, if 
$\Gamma \in \mf{F}_2^{\mathrm{1PI}}$ is decomposed as in 
Eq.~\eqref{eqn:skeletonInsertion}, then $\Gamma$ will be accounted for 
with a pre-factor of 
\[
\frac{r_{\Gamma}}{S_{\Gamma_{\mr{s}}} \cdot\, \prod_{k=1}^K S_{\Gamma^{(k)}} }.
\]
Then our hope is that 
\begin{equation}
\label{eqn:redundancy}
S_{\Gamma} = \frac{S_{\Gamma_{\mr{s}}} \cdot\, \prod_{k=1}^K S_{\Gamma^{(k)}} }{r_{\Gamma}}.
\end{equation}
This is a key reason for justifying bold diagrams, and is the content of
Corollary \ref{cor:redundancy} below. We can see from the above
discussion that it holds in the case of
Fig.~\ref{fig:feynmanSigma_decompose}.

For now we check Eq.~\eqref{eqn:redundancy} in a few more cases as we 
further develop the notion of the redundancy factor.

Consider Fig.~\ref{fig:feynmanSigma_decompose_3}. Here the insertion 
$\Gamma_{\mr{g}}$ is not symmetric, so we count $r_{\Gamma} = 2$ 
different ways of inserting it into the skeleton $\Gamma_{\mr{s}}$ to make 
a diagram isomorphic to $\Gamma$. Moreover, $S_{\Gamma_{\mr{g}}} = 2$, 
$S_{\Gamma_{\mr{s}}} = 2$, and $S_{\Gamma} = 2$, so Eq.~\eqref{eqn:redundancy} 
holds.

\begin{figure}[h]
  \begin{center}
    \includegraphics[width=0.8\textwidth]{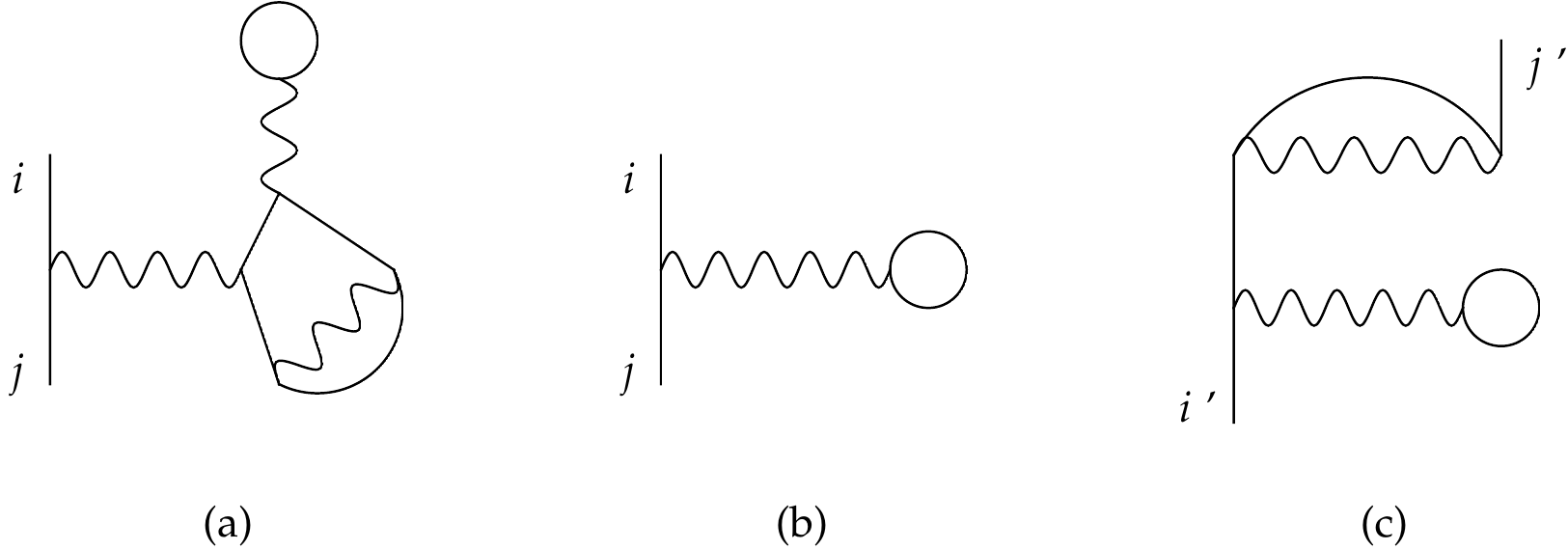}
  \end{center}
  \caption{Decomposition of a non-skeleton diagram $\Gamma$ in (a) into
  a skeleton diagram $\Gamma_{\mr{s}}$ in (b) and a truncated Green's function diagram
  $\Gamma_{\mr{g}}$ in (c).}
  \label{fig:feynmanSigma_decompose_3}
\end{figure}

Next consider Fig.~\ref{fig:feynmanSigma_decompose_2}. The diagram 
$\Gamma_{\mr{g}}$ in (c) can be inserted into the skeleton $\Gamma_{\mr{s}}$ in (d) 
into two different locations, yielding the isomorphic diagrams in (a) and (b). Hence 
$r_{\Gamma} = 2$. Moreover, observe that $S_{\Gamma_{\mr{g}}} = 1$, 
$S_{\Gamma_{\mr{s}}} = 4$, and $S_{\Gamma} = 2$, so Eq.~\eqref{eqn:redundancy} 
holds.

\begin{figure}[h]
  \begin{center}
    \includegraphics[width=0.75\textwidth]{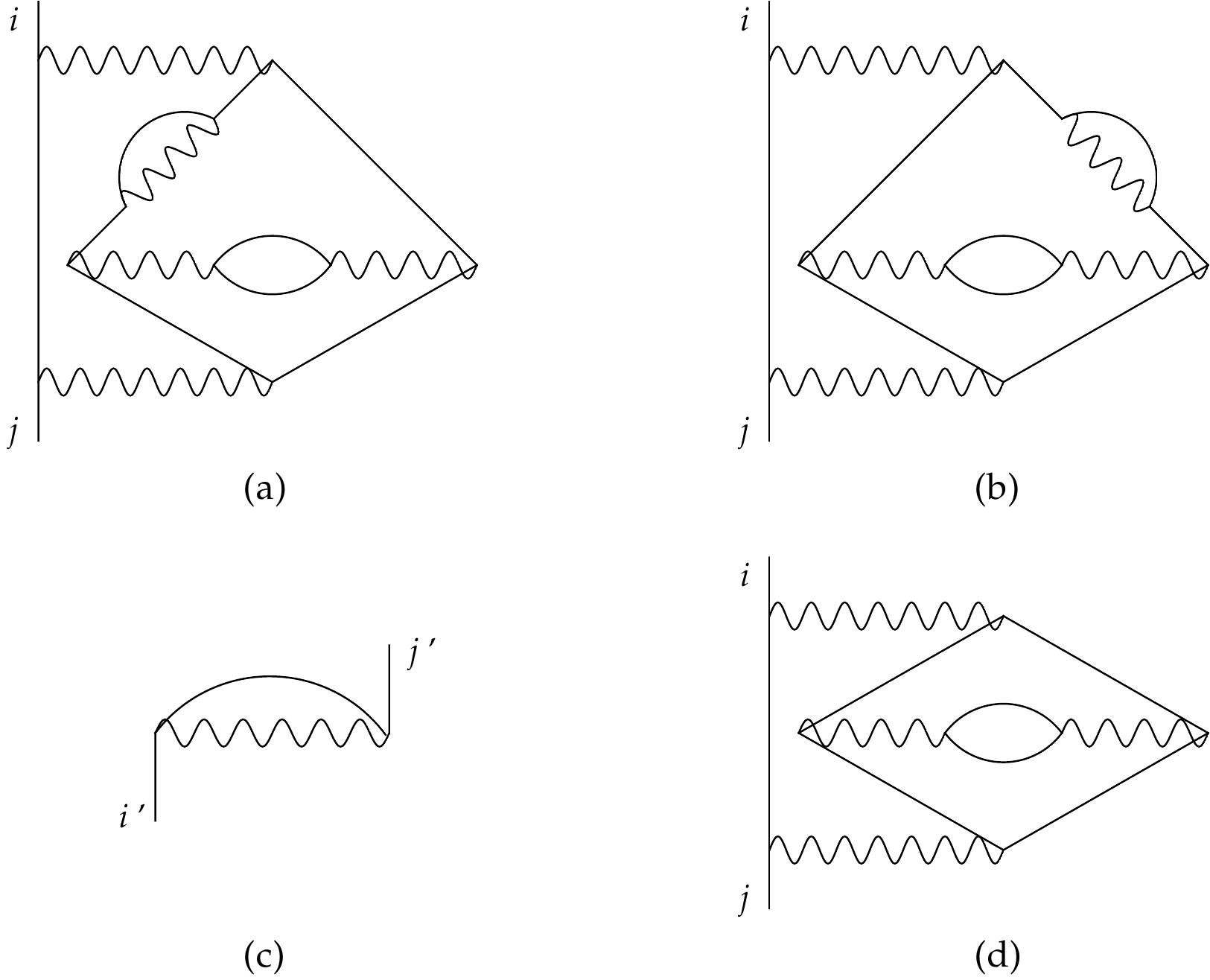}
    \vspace{2mm}
  \end{center}
  \caption{Decomposition of a more complex non-skeleton diagram $\Gamma$ in (a) into
  a skeleton diagram $\Gamma_{\mr{s}}$ in (d) and a truncated Green's function diagram
  $\Gamma_{\mr{g}}$ in (c). The non-skeleton diagram in (b) is isomorphic to 
  (a), but is obtained by the insertion of $\Gamma_{\mr{s}}$ into
  $\Gamma_{\mr{g}}$ at a different location.}
  \label{fig:feynmanSigma_decompose_2}
\end{figure}

\begin{figure}[h]
  \begin{center}
    \includegraphics[width=0.325\textwidth]{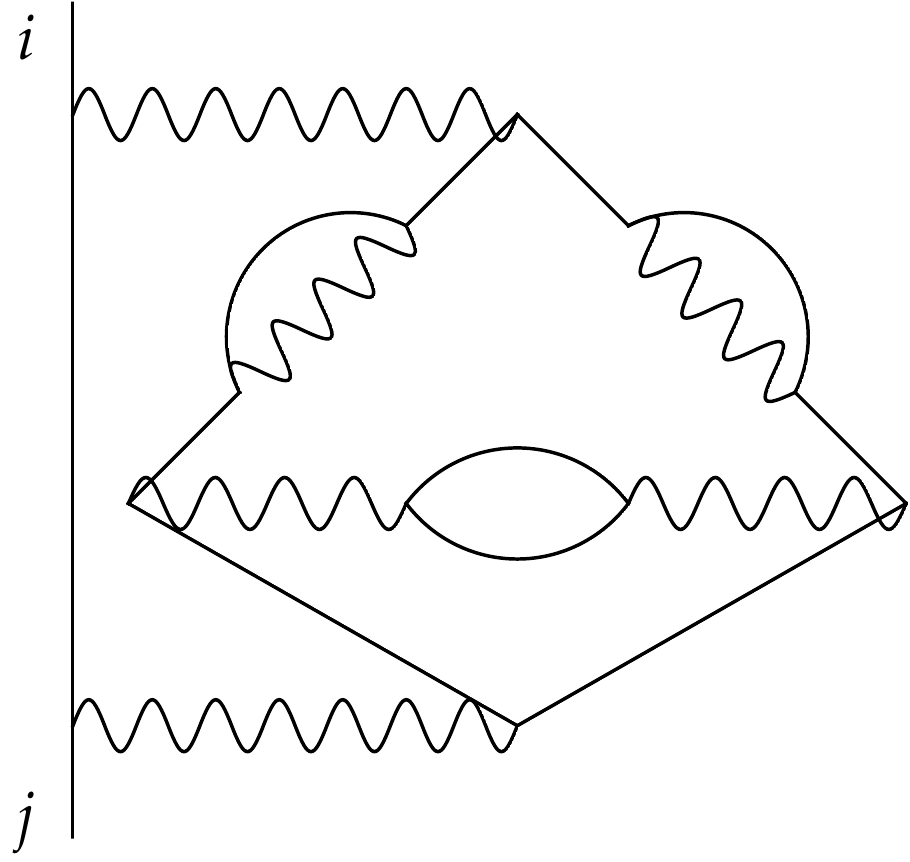}
    \vspace{2mm}
  \end{center}
  \caption{A non-skeleton diagram $\Gamma$ with redundancy factor $r_{\Gamma} = 1$ 
  and the same
  skeleton $\Gamma_{\mr{s}}$ as in Fig.~\ref{fig:feynmanSigma_decompose_2} (d).}
  \label{fig:feynmanSigma_decompose_4}
\end{figure}

By contrast the diagram $\Gamma$ in Fig.~\ref{fig:feynmanSigma_decompose_4}, which 
has the same skeleton $\Gamma_{\mr{s}}$ but 
admits two maximal insertions, has a redundancy factor of only $r_{\Gamma}=1$. 
The left-right symmetry of the diamond does not yield additional redundancy 
because the maximal insertions exchanged by this symmetry are isomorphic to 
each other---and in fact to the insertion 
of Fig.~\ref{fig:feynmanSigma_decompose_2} (c).
Thus in the bold diagram expansion Eq.~\eqref{eqn:bold}, $\Gamma$ is only accounted 
for \emph{once}. Meanwhile, $S_{\Gamma} = 4$, $S_{\Gamma_{\mr{s}}} = 4$, 
and the symmetry factors of the insertions are both one, so Eq.~\eqref{eqn:redundancy} holds.

In Fig.~\ref{fig:feynmanSigma_decompose_6_2} (a), we show a non-skeleton diagram 
$\Gamma$ which has the same skeleton $\Gamma_{\mr{s}}$ 
as in the last two examples. $\Gamma$ admits two (non-symmetric) maximal 
insertions, each isomorphic to
the diagram of Fig.~\ref{fig:feynmanSigma_decompose_3} (c). 
There are two nonequivalent ways of inserting these diagrams into $\Gamma_{\mr{s}}$ to 
yield a diagram isomorphic to $\Gamma$, depicted separately 
in Fig.~\ref{fig:feynmanSigma_decompose_6_2} (a), (b). 
We have $S_{\Gamma}=8$ (with 
a factor of $4$ coming from the two half-dumbbells) and 
$S_{\Gamma_{\mr{s}}} = 4$, and the symmetry factor of each insertion is $2$ (due to 
the half-dumbbell), so 
Eq.~\eqref{eqn:redundancy} holds.

\begin{figure}[h]
  \begin{center}
    \includegraphics[width=1.0\textwidth]{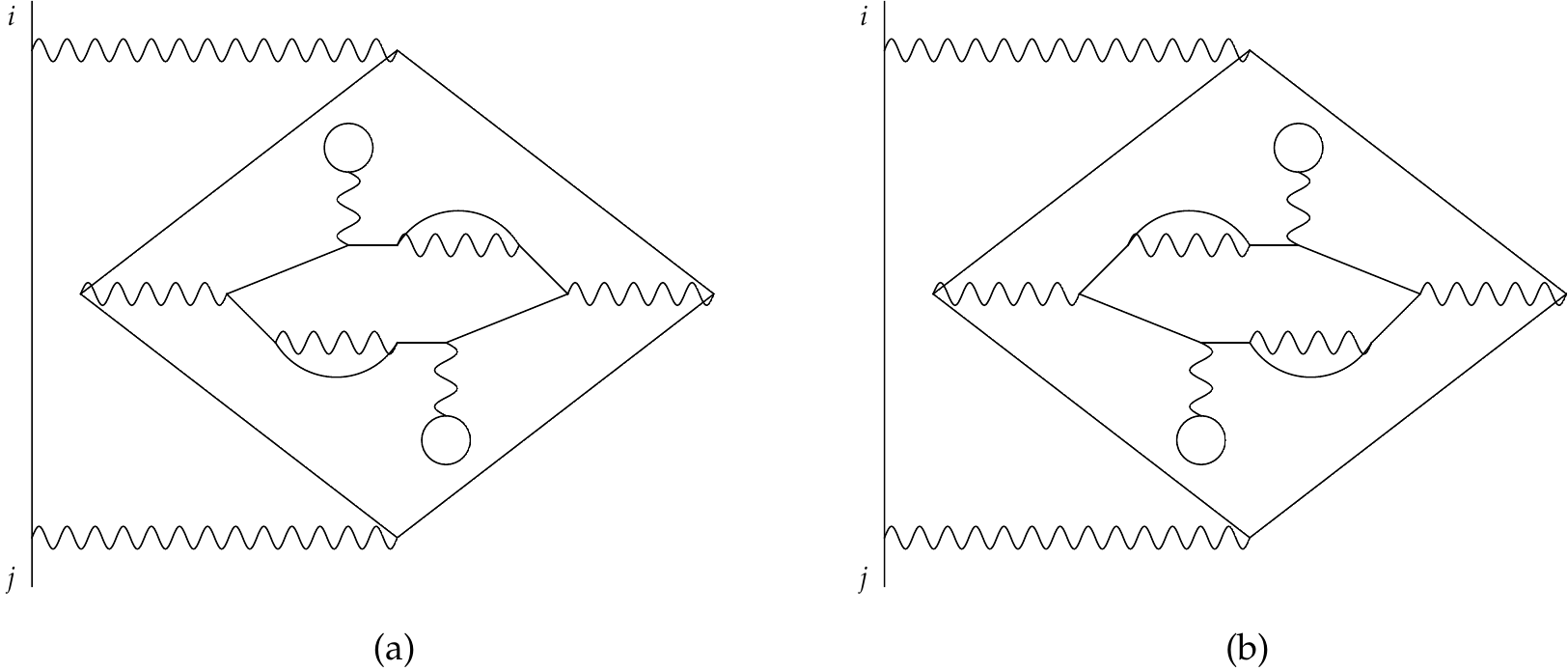}
  \end{center}
  \caption{A non-skeleton diagram $\Gamma$ in (a) with redundancy factor $r_{\Gamma} = 2$ and the same
   skeleton $\Gamma_{\mr{s}}$ as in Fig.~\ref{fig:feynmanSigma_decompose_2} (d). The diagram 
   in (b) is a diagram isomorphic to $\Gamma$ obtained from $\Gamma_{\mr{s}}$ 
   by a nonequivalent set of insertions.}
  \label{fig:feynmanSigma_decompose_6_2}
\end{figure}

Finally, in Fig.~\ref{fig:feynmanSigma_decompose_6_2}, we show 
a diagram $\Gamma$ which once again has the same skeleton $\Gamma_{\mr{s}}$ 
as in the last several examples.
$\Gamma$ admits one (non-symmetric) maximal insertion isomorphic to
the diagram of Fig.~\ref{fig:feynmanSigma_decompose_3} (c). This can be inserted into either 
propagator line of the `bubble' in the center of $\Gamma_{\mr{s}}$ and with either orientation to yield 
$\Gamma$ up to isomorphism, so $r_{\Gamma} = 4$. We have $S_{\Gamma}=2$ (due 
to the half-dumbbell) and 
$S_{\Gamma_{\mr{s}}} = 4$, and the symmetry factor of the insertion is $2$ (due to the half-dumbbell), so 
Eq.~\eqref{eqn:redundancy} holds.

\begin{figure}[h]
  \begin{center}
    \includegraphics[width=0.45\textwidth]{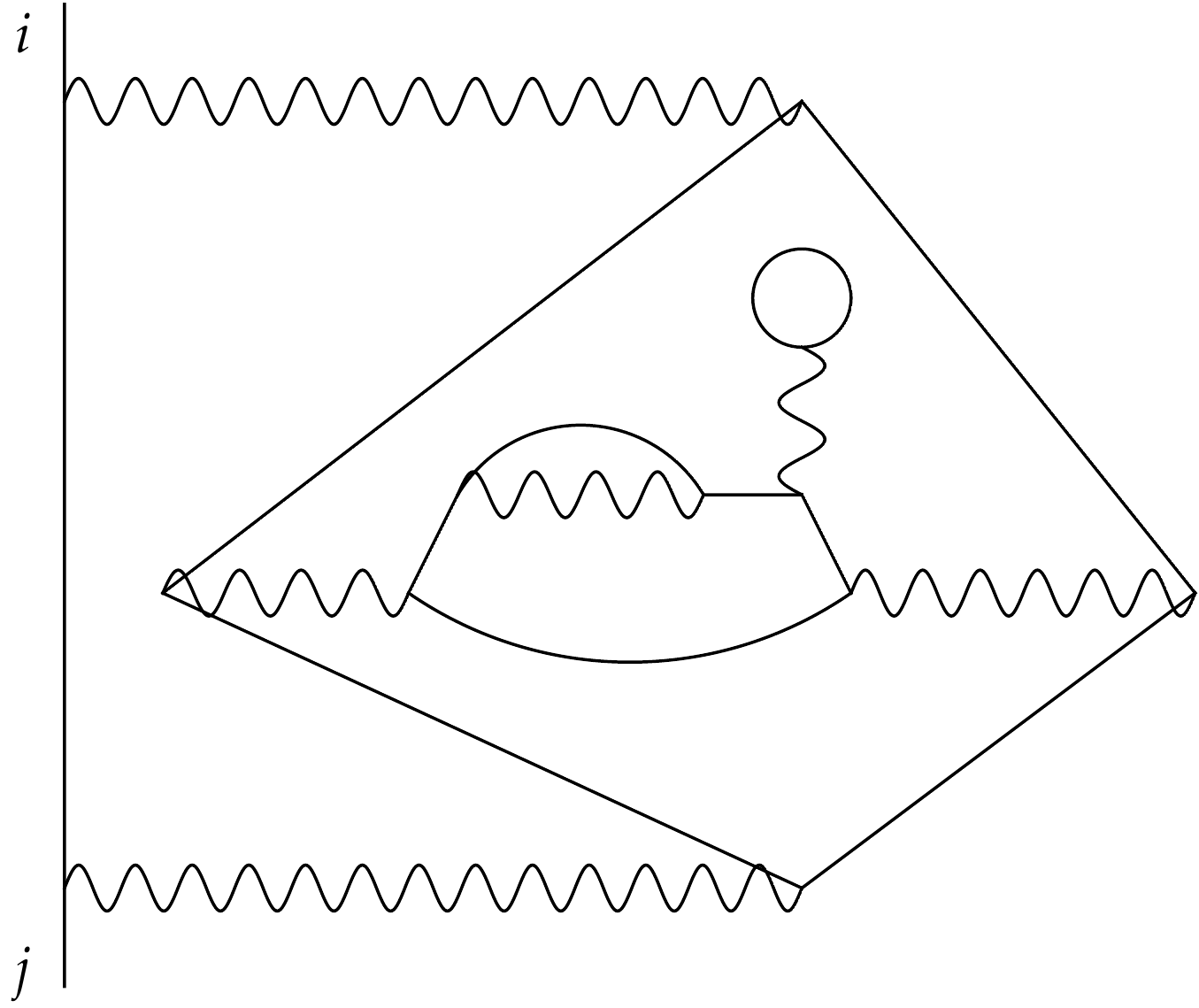}
  \end{center}
  \caption{A non-skeleton diagram $\Gamma$ with redundancy factor $r_{\Gamma} = 4$ and the same
   skeleton $\Gamma_{\mr{s}}$ as in Fig.~\ref{fig:feynmanSigma_decompose_2} (d).}
  \label{fig:feynmanSigma_decompose_5}
\end{figure}

We will refer back to these examples for concreteness in the developments that follow.

\subsection{Ways of producing a self-energy diagram from its skeleton}
\label{sec:ways}
As promised we provide a rigorous definition of the redundancy factor, as well 
as the set of ways of producing a self-energy diagram from its skeleton.
Consider a self-energy diagram $\Gamma$, and write 
\begin{equation}
\label{eqn:gammaDecompose}
\Gamma = \Gamma_{\mr{s}} \oplus_{(h_1^{(1)},h_2^{(1)},\ldots,h_1^{(K)},h_2^{(K)})} \left[\Gamma^{(1)},\ldots,\Gamma^{(K)} \right]
\end{equation}
via Proposition \ref{prop:selfenergycompose}.

\begin{remark}
\label{rem:orientation}
We assume that the ordering within each pair $(h_1^{(k)},h_2^{(k)})$ 
is chosen so that if $\Gamma^{(j)}$, $\Gamma^{(k)}$ are 
non-isomorphic for some $j,k$, then in fact $\Gamma^{(j)}$ and $\Gamma^{(k)}$ are non-isomorphic after 
any external relabeling. In other words, the insertions $\Gamma^{(k)}$ are externally labeled  
such that if any two of them are isomorphic up to external labeling, then they are actually 
isomorphic (with external labeling taken into account). We follow this convention for all 
decompositions of the form of Eq.~\eqref{eqn:gammaDecompose} in the sequel.
\end{remark}

 Implicitly $\Gamma_{\mr{s}}$ is 
a `subdiagram' of $\Gamma$ in that its vertex and half-edge sets are subsets of those of 
$\Gamma$. Roughly speaking, there is only one way to construct $\Gamma$ 
from $\Gamma_{\mr{s}}$ via 
Green's function insertions (namely, via the procedure represented in 
Eq.~\eqref{eqn:gammaDecompose}), 
but there may be many ways to construct diagrams \emph{isomorphic} to $\Gamma$ 
from $\Gamma_{\mr{s}}$ via 
Green's function insertions. 
The uniqueness result of Proposition \ref{prop:selfenergycompose} guarantees that 
any such way must involve (up to isomorphism) the insertion of the 
same set $\{ \Gamma^{(1)},\ldots,\Gamma^{(K)}\}$ of truncated Green's function 
diagrams.
Then let $I(\Gamma,\Gamma_{\mr{s}})$ be the set of ways of replacing $K$ propagator 
lines in $\Gamma_{\mr{s}}$ with $\Gamma^{(1)},\ldots,\Gamma^{(K)}$ such that the 
resulting diagram 
is \emph{isomorphic} to $\Gamma$ as a truncated Green's function diagram. 
(There is some abuse of notation here because $I(\Gamma,\Gamma_{\mr{s}})$ 
additionally depends on the decomposition of Eq.~\eqref{eqn:gammaDecompose}, but 
the meaning will be clear from context.)

More precisely, 
each such `way' consists of the following data: a set of ordered pairs of half-edges 
$(h_1'^{(1)},h_2'^{(1)}),\ldots,(h_1'^{(K)},h_2'^{(K)})$ in $\Gamma_{\mr{s}}$
such that 
\[
\Gamma' := \Gamma_{\mr{s}} \oplus_{(h_1'^{(1)},h_2'^{(1)},\ldots,h_1'^{(K)},h_2'^{(K)})} \left[\Gamma^{(1)},\ldots,\Gamma^{(K)} \right]
\]
is \emph{isomorphic} to $\Gamma$, subject to an \textit{equivalence
relation}. 
Specifically, ways are not distinguished if they differ only by reordering the $K$ 
half-edge pairs by a permutation $\tau \in S_K$ such that $\Gamma^{(\tau(k))}$ is 
isomorphic to $\Gamma^{(k)}$ for all $k$.
Moreover, ways are not distinguished if they differ only by the ordering \emph{within} 
the $k$-th half-edge pair for $k$ such that $\Gamma^{(k)}$ is \emph{symmetric}.
We refer to the equivalence class of $(h_1'^{(1)},h_2'^{(1)}),\ldots,(h_1'^{(K)},h_2'^{(K)})$ 
as the element of $I(\Gamma,\Gamma_{\mr{s}})$ \emph{specified} by these half-edge pairs.

Observe that if we sum over the skeleton diagrams and then  
formally replace each bold line with a sum over Green's function diagrams, then 
in the resulting formal sum over self-energy diagrams, 
each self-energy diagram $\Gamma$ will be counted precisely $\vert I(\Gamma,\Gamma_{\mr{s}})\vert$ 
times, where $\Gamma_{\mr{s}}$ is the skeleton of $\Gamma$. This number 
 $r_{\Gamma} := \vert I(\Gamma,\Gamma_{\mr{s}})\vert$ depends only on 
 $\Gamma \in \mf{F}_2^{\mr{1PI}}$, and as suggested earlier we call 
 it the \emph{redundancy factor} of $\Gamma$.

It is worthwhile to treat the distinction between symmetric and non-symmetric 
insertions a bit more elegantly (and, moreover, in a way that does not so clearly privilege the 
fact that our insertions have \emph{two} external half-edges). For a truncated Green's 
function diagram, the external half-edges have labels `$i$' and `$j$.'
Let the \emph{external symmetry group}, denoted $S(\Gamma_{\mr{g}})$,
of a truncated Green's function diagram $\Gamma_{\mr{g}}$ be the subgroup of $\mr{Sym}\{i,j\} \simeq S_2$ consisting of 
permutations of the labels `$i$' and `$j$' that fix the isomorphism class of the diagram. Therefore 
for symmetric diagrams the external symmetry group is $S_2$, and for non-symmetric diagrams  
it is the trivial group. For future convenience, 
let the action of $\sigma \in S_2$ on a truncated Green's function diagram 
$\Gamma_{\mr{g}}$ defined via permutation of the external labels be
denoted $\sigma\star\Gamma_{\mr{g}}$. (The `$\,\star\,$' notation is meant to distinguish 
from the group action `$\,\cdot\,$' defined earlier.)

Then using this language we can say that ways that 
for any $\sigma \in S(\Gamma^{(k)})$ the modification of  
$(h_{1}'^{(k)},h_{2}'^{(k)})$ to 
$(h_{\sigma(1)}'^{(k)},h_{\sigma(2)}'^{(k)})$ does not yield a distinct element of 
$I(\Gamma,\Gamma_{\mr{s}})$.

\subsection{Understanding automorphisms in terms of the skeleton}
\label{sec:autSkeleton}
Now we turn to item (2), i.e., characterizing the structure of automorphisms of $\Gamma$ in terms of 
its decomposition furnished by Proposition \ref{prop:selfenergycompose}.

With notation as in the section \ref{sec:ways}, 
let $n$ be the order of $\Gamma$, 
and let $p$ be the order of $\Gamma_{\mr{s}}$. Then $q=n-p$ is the order of 
$\Gamma_{\mr{g}} := \bigcup_{k=1}^K \Gamma^{(k)}$. 

We can 
view the skeleton diagram $\Gamma_{\mr{s}}$ as well as the insertions $\Gamma^{(k)}$ 
as labeled truncated Green's function diagrams via the 
labeling of interaction lines and half-edges inherited from $\Gamma$. Let 
$\mr{Aut}(\Gamma,\Gamma_{\mr{s}})$ be the set of 
automorphisms of $\Gamma_{\mr{s}}$ that can be \emph{extended} 
to automorphisms of $\Gamma$ by relabeling the rest of the diagram, i.e., permuting 
the vertex and half-edge labels of $\Gamma_{\mr{g}}$. 

For example, the automorphism of $\Gamma_{\mr{s}}$ of Fig.~\ref{fig:feynmanSigma_decompose_2} (d)
corresponding to the left-right reflection of the `diamond' 
can be extended to an automorphism of the diagram in Fig.~\ref{fig:feynmanSigma_decompose_4}, 
but it \emph{cannot} be extended to an automorphism of any of the diagrams of 
Fig.~\ref{fig:feynmanSigma_decompose_2} (a), (b), nor of any of the diagrams of 
Figs.~\ref{fig:feynmanSigma_decompose_6_2} and \ref{fig:feynmanSigma_decompose_5}. 
Next, consider the automorphism of $\Gamma_{\mr{s}}$ obtained by composing a 
left-right reflection of the diamond with a swap of the two propagator lines in the 
`bubble' at the center of the diamond. This extends to an automorphism 
of each of the diagrams of 
Fig.~\ref{fig:feynmanSigma_decompose_2} (a), (b).

More precisely, viewing 
$\mathbf{R}_q$ as acting on labelings of $\Gamma_{\mr{g}}$, an 
element $g\in \mathrm{Aut}(\Gamma_{\mr{s}})$ is defined to be in $\mr{Aut}(\Gamma,\Gamma_{\mr{s}})$ 
if there exists $h \in \mathbf{R}_q$ such that $gh \in \mr{Aut}(\Gamma)$. (Since $\Gamma_{\mr{s}}$ and 
$\Gamma_{\mr{g}}$ are disjoint, elements $g$ and $h$ 
as in the preceding commute.) Note that $\mr{Aut}(\Gamma,\Gamma_{\mr{s}})$ is 
a subgroup of $\mathrm{Aut}(\Gamma_{\mr{s}})$: indeed, if $g_1, g_2 \in \mr{Aut}(\Gamma,\Gamma_{\mr{s}})$, 
then there exist $h_1,h_2 \in \mathbf{R}_q$ such that $g_1 h_1, g_2 h_2 \in \mr{Aut}(\Gamma)$, but 
then $(g_1 g_2) (h_1 h_2) = (g_1 h_1) (g_2  h_2)$ is in $\mr{Aut}(\Gamma)$, so 
$g_1 g_2 \in \mr{Aut}(\Gamma,\Gamma_{\mr{s}})$.

We have the following characterization of $\mr{Aut}(\Gamma,\Gamma_{\mr{s}})$:

\begin{lemma}
\label{lem:autGamma1}
Let $\Gamma \in \mf{F}_2^{\mr{1PI}}$, and write
\[
\Gamma = \Gamma_{\mr{s}} \oplus_{(h_1^{(1)},h_2^{(1)},\ldots,h_1^{(K)},h_2^{(K)})} \left[\Gamma^{(1)},\ldots,\Gamma^{(K)} \right]
\]
via Proposition \ref{prop:selfenergycompose}. An element 
$g\in \mr{Aut}(\Gamma_{\mr{s}})$ lies in $\mr{Aut}(\Gamma,\Gamma_{\mr{s}})$ 
if and only if for every $k$, there is some $k'$ and some $\sigma \in S_2$ such that 
$\Gamma^{(k)}$ is isomorphic 
to $\sigma \star \Gamma^{(k')}$ and
$\psi_g$ sends 
$(h_1^{(k)},h_2^{(k)})$ to $(h_{\sigma(1)}^{(k')},h_{\sigma(2)}^{(k')})$.
\end{lemma}
\begin{proof}
First we prove the forward direction, so let $g\in \mr{Aut}(\Gamma,\Gamma_{\mr{s}})$. 
Then $g$ extends to an automorphism of $\Gamma$, which we shall also call $g$. 
Let $e_1^{(k)},e_2^{(k)}$ be the external half-edges of the truncated 
Green's function diagram $\Gamma^{(k)}$ paired with $h_1^{(k)},h_2^{(k)}$, respectively, 
in the overall diagram $\Gamma$ (equivalently, the external half-edges labeled 
`$i$' and `$j$,' respectively). Then the maximal insertion $\Gamma^{(k)}$ is 
disconnected from the rest of $\Gamma$ by unpairing $\{e_1^{(k)},h_1^{(k)} \}$ 
and $\{e_2^{(k)},h_2^{(k)}\}$ in $\Gamma$. Since $g$ is an automorphism, 
removing the links  
$\{\psi_g(e_1^{(k)}),\psi_g(h_1^{(k)}) \}$ and $\{\psi_g(e_2^{(k)}),\psi_g(h_2^{(k)}) \}$ 
from $\Gamma$ must also a disconnect a maximal insertion (isomorphic to 
$\Gamma^{(k)}$)
at $(\psi_g(h_1^{(k)}), \psi_g(h_2^{(k)}))$ 
with external half-edges $\psi_g(e_1^{(k)}), \psi_g(e_2^{(k)})$ 
labeled `$i$' and `$j$,' respectively.
Since this diagram is a maximal insertion, by Proposition \ref{prop:selfenergycompose} 
it must be $\sigma \star \Gamma^{(k')}$ for some $k'$, where 
$\sigma \in S_2$, and moreover $\psi_g(h_1^{(k)}) = h_{\sigma(1)}^{(k')}$
and $\psi_g(h_2^{(k)}) = h_{\sigma(2)}^{(k')}$. This concludes the proof 
of the forward direction.

Now assume that $g\in \mr{Aut}(\Gamma_{\mr{s}})$ and 
that for every $k$, there is some $k' = k'(k)$ and some 
$\sigma = \sigma(k) \in S_2$ such that 
$\Gamma^{(k')}$ is isomorphic 
to $\sigma \star \Gamma^{(k)}$ via some isomorphism 
$(\vp^{(k)},\psi^{(k)})$ and
$\psi_g$ sends 
$(h_1^{(k)},h_2^{(k)})$ to $(h_{\sigma(1)}^{(k')},h_{\sigma(2)}^{(k')})$. Then 
we aim to extend $g$ to an automorphism of $\Gamma$, i.e., we aim to 
extend $(\vp_g, \psi_g)$ to an isomorphism from $\Gamma$ to itself.
This can be 
done simply by mapping vertices and half-edges lying in the $\Gamma^{(k)}$ 
via $(\vp^{(k)},\psi^{(k)})$. It is straightforward to check that this indeed defines 
an automorphism.
\end{proof}

We also have the following result characterizing the structure of automorphisms of $\Gamma$ 
in terms of $\mr{Aut}(\Gamma,\Gamma_{\mr{s}})$:

\begin{lemma}
\label{lem:autGamma2}
Let $\Gamma \in \mf{F}_2^{\mr{1PI}}$ be decomposed as 
in Eq.~\eqref{eqn:skeletonInsertion}. 
Then any $g \in \mathrm{Aut}(\Gamma)$ restricts to an automorphism of $\Gamma_{s}$ 
(in particular, only permutes vertex labels \emph{within} the subdiagram 
$\Gamma_{s}$). By definition, this induced automorphism of $\Gamma_{\mr{s}}$ then lies in 
$\mr{Aut}(\Gamma,\Gamma_{\mr{s}})$. Moreover, if $\Gamma$ admits a maximal insertion 
$\Gamma'$ at $(h_1',h_2')$, then $\Gamma$ also admits a maximal insertion $\Gamma''$ isomorphic to
$\Gamma'$ at $(\psi_g (h_1'),\psi_g (h_2'))$. Furthermore,  
$g$ sends all vertex labels from $\Gamma'$ to $\Gamma''$, i.e., $\vp_g$ sends each vertex 
of $\Gamma'$ to a vertex of $\Gamma''$.
\end{lemma}

\begin{proof}
Let $g \in \mathrm{Aut}(\Gamma)$. First we prove that $g$ only permutes vertex labels
\emph{within} $\Gamma_{\mr{s}}$. Indeed, suppose not.
Then $\vp_g$ sends a vertex that is not contained in any insertion 
to a vertex that is contained in some insertion. The property of 
whether or not a vertex 
is contained in an insertion is preserved under diagram isomorphism, 
so we have a contradiction.

Furthermore, 
an isomorphism of unlabeled diagrams sends 
maximal insertions to maximal insertions; i.e., 
if $\Gamma$ admits a maximal insertion 
$\Gamma'$ at $(h_1',h_2')$, then $\Gamma$ also admits a maximal insertion $\Gamma''$ isomorphic to
$\Gamma'$ at $(\psi_g (h_1'),\psi_g (h_2'))$ via $(\vp_g,\psi_g)$.
(In particular $\vp_g$ sends each vertex 
of $\Gamma'$ to a vertex of $\Gamma''$.)

Then it can be readily checked, by collapsing maximal insertions, 
that $g$ descends to an automorphism of the skeleton $\Gamma_{\mr{s}}$. 
Then by definition we can view $g\in \mr{Aut}(\Gamma,\Gamma_{\mr{s}})$.
\end{proof}

The preceding two lemmas can be used to compute the symmetry factor 
of $\Gamma \in \mf{F}_2^{\mr{1PI}}$ via its skeleton decomposition:

\begin{lemma}
\label{lem:autGamma3}
Let $\Gamma \in \mf{F}_2^{\mr{1PI}}$, decomposed as 
in Eq.~\eqref{eqn:skeletonInsertion}. Then
\begin{equation}
\label{eqn:sizeAutSkeleton}
S_{\Gamma} = \vert \mr{Aut}(\Gamma,\Gamma_{\mr{s}}) \vert \cdot \prod_{k=1}^K S_{\Gamma^{(k)}}
\end{equation}
\end{lemma}

\begin{proof}
Lemma \ref{lem:autGamma2} says that every $g\in \mr{Aut}(\Gamma)$ descends to 
$g\in \mr{Aut}(\Gamma,\Gamma_{\mr{s}})$ and moreover defines
an isomorphism from each insertion $\Gamma^{(k)}$ to its image 
under $g$.

Conversely, by Lemma \ref{lem:autGamma1}, for any 
$g\in \mr{Aut}(\Gamma,\Gamma_{\mr{s}})$ and any $k$,
$\psi_g$ sends 
$(h_1^{(k)},h_2^{(k)})$ to $(h_{\sigma(1)}^{(k')},h_{\sigma(2)}^{(k')})$ 
for some $\sigma = \sigma(k) \in S_2$, 
where $k' = k'(k) $ is such that $\Gamma^{(k)}$ is isomorphic to
 $\sigma \star \Gamma^{(k')}$. Any choice of isomorphisms from 
 the $\Gamma^{(k)}$ to the $\sigma \star \Gamma^{(k')}$ defines 
 an extension of $g$ to an automorphism of $\Gamma$. 

Thus any automorphism of $\Gamma$ can be yielded constructively by starting with 
$g \in \mr{Aut}(\Gamma,\Gamma_{\mr{s}})$ and then choosing, for each $k$, an isomorphism
from
$\Gamma^{(k)}$ to $\sigma \star \Gamma^{(k')}$. 
The number of such 
isomorphisms is the same as the number of automorphisms of $\Gamma^{(k)}$, so 
Eq.~\eqref{eqn:sizeAutSkeleton} follows.
\end{proof}

\subsection{The action of $\mr{Aut}(\Gamma_{\mr{s}})$ on $I(\Gamma,\Gamma_{\mr{s}})$}
\label{sec:action}
Finally, we turn to item (3).
Again decompose $\Gamma \in \mf{F}_2^{\mr{1PI}}$ as 
in Eq.~\eqref{eqn:skeletonInsertion}. Notice that Eq.~\eqref{eqn:skeletonInsertion} 
itself defines an element of $I(\Gamma,\Gamma_{\mr{s}})$. Call this element $\iota^*$.

The key observation here is that $\mr{Aut}(\Gamma_{s})$ acts transitively on $I(\Gamma,\Gamma_{\mr{s}})$
and that the stabilizer of any $\iota\in I(\Gamma,\Gamma_{\mr{s}})$ is $\mr{Aut}(\Gamma,\Gamma_{\mr{s}})$.
We define the action as follows.
Let  $g\in \mr{Aut}(\Gamma_{\mr{s}})$, and consider
an element $\iota$ of $I(\Gamma,\Gamma_{\mr{s}})$ specified by a set of ordered 
pairs of half-edges 
$(h_1'^{(1)},h_2'^{(1)}),\ldots,(h_1'^{(K)},h_2'^{(K)})$ in $\Gamma_{\mr{s}}$, 
Then $g \cdot \iota$ is defined to be
the element of $I(\Gamma,\Gamma_{\mr{s}})$ specified by the ordered pairs 
$(\psi_g(h_1'^{(1)}),\psi_g(h_2'^{(1)})),\ldots,(\psi_g (h_1'^{(K)}), \psi_g(h_2'^{(K)}))$.

For example, recall the automorphism of $\Gamma_{\mr{s}}$ of Fig.~\ref{fig:feynmanSigma_decompose_2} (d)
corresponding to the left-right reflection of the `diamond.' The action of this automorphism 
fixes the only element in $I(\Gamma,\Gamma_{\mr{s}})$ represented by 
Fig.~\ref{fig:feynmanSigma_decompose_4}. Meanwhile, it swaps the elements of 
$I(\Gamma,\Gamma_{\mr{s}})$ represented in Fig.~\ref{fig:feynmanSigma_decompose_2} (a) and (b). 
In a slightly more indirect way, it also swaps the elements of $I(\Gamma,\Gamma_{\mr{s}})$ 
represented in Figs.~\ref{fig:feynmanSigma_decompose_6_2} (a) and (b).
Next, consider the automorphism of $\Gamma_{\mr{s}}$ obtained by a swap 
of the two propagator lines in the 
`bubble' at the center of the diamond. The action of this automorphism 
also swaps the elements of $I(\Gamma,\Gamma_{\mr{s}})$ 
represented in Figs.~\ref{fig:feynmanSigma_decompose_6_2} (a) and (b).

\begin{lemma}
\label{lem:autGamma4}
With notation as in the preceding, the action of $\mr{Aut}(\Gamma_{s})$ 
on $I(\Gamma,\Gamma_{\mr{s}})$ is transitive, and 
the stabilizer of $\iota^*$ is $\mr{Aut}(\Gamma,\Gamma_{\mr{s}})$.
\end{lemma}

\begin{proof}
First we establish that the action is transitive. To this end, consider arbitrary elements 
$\iota_1, \iota_2 \in I(\Gamma,\Gamma_{\mr{s}})$, i.e., two different ways of making insertions 
in $\Gamma_{\mr{s}}$ to yield diagrams $\Gamma_1,\Gamma_2$, respectively, that are 
each isomorphic to $\Gamma$. Our isomorphism from 
$\Gamma_1$ to $\Gamma_2$ descends 
(by collapsing the maximal insertions) to an isomorphism from $\Gamma_{\mr{s}}$ to itself, 
i.e., an automorphism $g \in \mathrm{Aut}(\Gamma_{\mr{s}})$, 
and evidently 
this automorphism satisfies $g\cdot \iota_1 = \iota_2$. This establishes transitivity.

Now we turn to the claim about the stabilizer. Let $g\in \mr{Aut}(\Gamma,\Gamma_{s})$. 
We want to show that $g \cdot \iota^* = \iota^*$. 
By Lemma \ref{lem:autGamma1}, there exists $\tau \in S_K$ and 
$\sigma_k \in S_2$ for $k=1,\ldots,K$ such that 
$\Gamma^{(k)}$ is isomorphic 
to $\sigma_k \star \Gamma^{(\tau(k))}$ and
$\psi_g$ sends 
$(h_1^{(k)},h_2^{(k)})$ to $(h_{\sigma_k(1)}^{(\tau(k))},h_{\sigma_k(2)}^{(\tau(k))})$. 
Hence $g\cdot \iota^*$ is specified by the ordered pairs 
$(h_{\sigma_k(1)}^{(\tau(k))},h_{\sigma_k(2)}^{(\tau(k))})$, $k=1,\ldots,K$.

By Remark \ref{rem:orientation}, since $\Gamma^{(k)}$ and $\Gamma^{(\tau(k))}$ 
are isomorphic up to external labeling, they are in fact isomorphic. But since 
 $\Gamma^{(k)}$ and $\sigma_k \star \Gamma^{(\tau(k))}$ are isomorphic, 
 this means that in turn $\sigma_k \star \Gamma^{(\tau(k))}$ is isomorphic to 
 $\Gamma^{(\tau(k))}$, i.e., $\sigma_k \in S(\Gamma^{(k)})$ for all $k$. 
 
Then recalling the equivalence relation used to define $I(\Gamma,\Gamma_{\mr{s}})$, we 
see that $g \cdot \iota^*$ is equivalently specified by the ordered pairs 
$(h_{1}^{(k)},h_{2}^{(k)})$, $k=1,\ldots,K$, i.e., $g\cdot \iota^* = \iota^*$.

Conversely,  suppose that $g \cdot \iota^* = \iota^*$ for some $g\in \mathrm{Aut}(\Gamma_{\mr{s}})$. 
Then there exist $\tau \in S_K$ and $\sigma_k \in S(\Gamma^{(k)})$ for $k=1,\ldots,K$ 
such that $\psi_g$ sends $(h_1^{(k)},h_2^{(k)})$ to $(h_{\sigma_k(1)}^{(\tau(k))},h_{\sigma_k(2)}^{(\tau(k))})$ 
for $k=1,\ldots,K$, and moreover $\Gamma^{(k)}$ is isomorphic to $\Gamma^{(\tau(k))}$ for all
$k$.
Since $\sigma_k \in S(\Gamma^{(k)})$, this means 
that $\Gamma^{(k)}$ is isomorphic to $\sigma_k \star \Gamma^{(\tau(k))}$ for all $k$. 
But then by Lemma \ref{lem:autGamma1} we have that $g\in \mr{Aut}(\Gamma,\Gamma_{\mr{s}})$.
\end{proof}

Then the orbit-stabilizer theorem, together with Lemmas \ref{lem:autGamma3} 
and \ref{lem:autGamma4}, yields the following corollary:

\begin{corollary}
\label{cor:redundancy}
For $\Gamma \in \mf{F}_2^{\mr{1PI}}$, the redundancy factor of $\Gamma$ is given by 
\[
r_{\Gamma}
= \frac{ S_{\Gamma_{\mr{s}}} \cdot \prod_{k=1}^K S_{\Gamma^{(k)}} }{ S_{\Gamma} }.
\]
\end{corollary}
\begin{proof}
Applying the orbit-stabilizer theorem via Lemma \ref{lem:autGamma4} we obtain 
\[
r_{\Gamma} = \vert I(\Gamma,\Gamma_{\mr{s}}) \vert 
= \frac{\vert \mr{Aut}(\Gamma_{\mr{s}}) \vert }{\vert \mr{Aut}(\Gamma, \Gamma_{\mr{s}}) \vert}
= \frac{ S_{\Gamma_{\mr{s}}} }{\vert \mr{Aut}(\Gamma, \Gamma_{\mr{s}}) \vert}.
\]
The result then follows from Lemma \ref{lem:autGamma3}.
\end{proof}

\subsection{Bold diagrammatic expansion for the self-energy}
\label{sec:finallyBold}
At last we can prove the bold diagrammatic expansion for the self-energy, 
stated as follows:

\begin{theorem}
\label{thm:boldExpand}
For $1\leq i,j \leq N$, we have the equality of formal power series (in the coupling constant)
\begin{equation}
\label{eqn:boldExpand}
\Sigma_{ij} = \sum_{\Gamma_{\mr{s}} \in \mf{F}_2^{\mr{2PI}}} \frac{\mathbf{F}_{\Gamma_{\mr{s}}}(i,j)}{S_{\Gamma_{\mr{s}}}},
\end{equation}
where $\Sigma_{ij}$ is interpreted as a power series via Theorem \ref{thm:Sigmaexpand} and 
where, 
for every $\Gamma_{s} \in \mf{F}_2^{\mr{2PI}}$, 
the expression $\mathbf{F}_{\Gamma_{s}}(i,j)$ is interpreted as the power series obtained by 
applying the Feynman rules for $\Gamma_{s}$ with propagator $G$, where $G$ is in turn interpreted 
as a formal power series via Theorem \ref{thm:Gexpand}.
\end{theorem}

\begin{remark}
\label{rem:boldExpand}
The interpretation of the bold diagrammatic expansion of the self-energy is at this point
somewhat cryptic. For the moment it can only be interpreted as a reorganization of 
the terms in the asymptotic series for the self-energy. However, since the terms on the 
right-hand side of 
Eq.~\eqref{eqn:boldExpand} depend only on $G$ and $v$ (and not on $G^0$), one 
might conjecture based on the expansion that the self-energy depends only on 
$G,v$. This is indeed a major goal of Part II, where we will indeed construct the self-energy (non-perturbatively) 
as a (matrix-valued) functional $\Sigma[G,v]$ of $G$ and $v$ only, and interpret the 
bold diagrammatic expansion as an asymptotic series in the coupling constant for 
the self-energy at fixed $G$, with terms given by the $\Sigma^{(k)}[G,v]$ to be specified below. 
The non-perturbative perspective will guarantee 
the existence of such an asymptotic series, but Theorem \ref{thm:boldExpand} 
will be used to show that this series is in fact given by Eq.~\eqref{eqn:boldExpand}.
\end{remark}

\begin{proof}
In the following all expressions should be suitably interpreted as in the statement of the theorem. 
For $\Gamma_{\mr{s}} \in \mf{F}_2^{\mr{2PI}}$, the series $\mathbf{F}_{\Gamma_{\mr{s}}}(i,j)$ 
counts $r_{\Gamma}$ times every self-energy 
diagram $\Gamma \in \mf{F}_2^{\mr{1PI}}$ with skeleton isomorphic to $\Gamma_{\mr{s}}$, 
each with factor $F_{\Gamma}(i,j) / \prod_{k=1}^K S_{\Gamma^{(k)}}$, where the 
$\Gamma^{(k)}$ are the maximal insertions of $\Gamma$. Then by 
Corollary \ref{cor:redundancy}, $\mathbf{F}_{\Gamma_{\mr{s}}}(i,j) / S_{\Gamma_{\mr{s}}}$
equals the sum of $F_{\Gamma}(i,j) / S_{\Gamma}$ 
over self-energy diagrams $\Gamma$ with skeleton isomorphic to $\Gamma_{\mr{s}}$.
Therefore the right-hand side of Eq.~\eqref{eqn:boldExpand} 
is the sum of $F_{\Gamma}(i,j) / S_{\Gamma}$ over all self-energy diagrams $\Gamma$.
\end{proof}

Following Remark \ref{rem:boldExpand}, each term in the bold 
diagrammatic expansion can be thought of as a functional of $G$ and $v$.  
We indicate by $\Sigma^{(k)}[G,v]$ the $k$-th order bold contribution to the 
self-energy, i.e., the contribution of the terms in the diagrammatic expansion 
that are of order $k$ in the interaction $v$.
In particular, the `first-order' bold contribution to the self-energy is
given by 
\begin{equation}
  \left(\Sigma^{(1)}[G,v]\right)_{ij} = 
  -\frac12 \left(\sum_{k} v_{ik} G_{kk}\right)
  \delta_{ij} - v_{ij} G_{ij}.
  \label{eqn:sigma1bold}
\end{equation}

These two terms are represented in Fig.~\ref{fig:feynmanSigmaG2} (a), (b), and 
we denote them by $\Sigma_{\mr{H}}[G,v]$ and $\Sigma_{\mr{F}}[G,v]$ for 
`Hartree' and `Fock,' respectively. 
 The associated diagrams happen to be the same as the first-order bare self-energy 
diagrams, but with thin lines replaced by bold lines.

\begin{figure}[h]
  \begin{center}
    \includegraphics[width=0.55\textwidth]{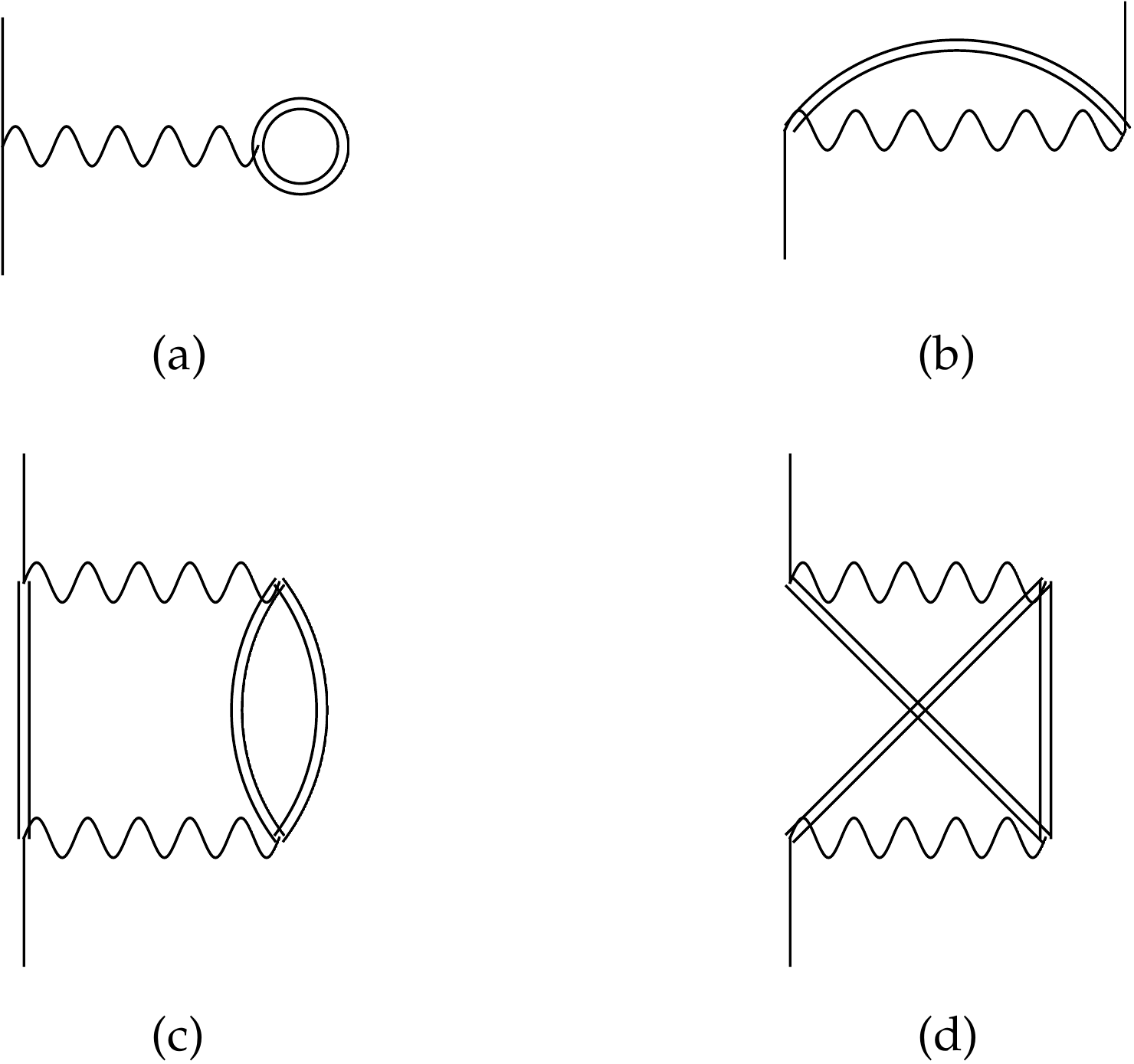}
  \end{center}
  \caption{Bold diagrammatic expansion of the self-energy up to second order (external labelings omitted).}
  \label{fig:feynmanSigmaG2}
\end{figure}

The approximation $\Sigma[G,v]\approx \Sigma^{(1)}[G,v]$ 
is known as the \emph{Hartree-Fock approximation}. One can likewise approximate
$\Sigma[G,v]\approx \Sigma^{(1)}[G,v]+\Sigma^{(2)}[G,v]$, where the second-order 
contribution can be written 
\begin{equation}
  \left(\Sigma^{(2)}[G]\right)_{ij} = 
  \frac12 G_{ij} \sum_{k,l} v_{ik} G^2_{kl}
  v_{lj} + \sum_{k,l} v_{ik} G_{kj} G_{kl} G_{li} v_{jl},
  \label{eqn:sigma2bold}
\end{equation}
and the second-order bold diagrams are shown in 
Fig.~\ref{fig:feynmanSigmaG2} (c), (d). The latter is known as the 
`second-order exchange,' while the former is an example of 
what is called a `ring diagram,' for 
reasons to be made clear later.

\subsection{Green's function methods}
\label{sec:greenMethod}
A \emph{Green's function method} using bold diagram expansion is a method for computing the Green's function 
via an ansatz for the self-energy $\Sigma_{\mr{ans}}[G,v] \approx \Sigma[G,v]$. 
This ansatz should be viewed as some sort of approximation of the full self-energy, 
usually consisting of diagrammatic contributions meant to incorporate certain physical 
effects.

After choosing an ansatz, one substitutes $\Sigma \leftarrow \Sigma_{\mr{ans}}[G,v]$ in the 
Dyson equation \eqref{eqn:Sigmadef} and attempts to solve it 
\emph{self-consistently} for $G$. In other words, one solves
\[
G = (A- \Sigma_{\mr{ans}} [G,v] )^{-1}
\]
for $G$, where $A$ and $v$ are specified in advance.

The most immediate technique for solving this system is a fixed-point iteration
that we refer to as  
the \emph{Dyson iteration}, which is defined by the update 
\[
G^{(k+1)} = (A- \Sigma_{\mr{ans}} [G^{(k)},v] )^{-1}.
\]
This iteration can be combined with damping techniques to yield better 
convergence in practice, but in general the convergence behavior of 
the Dyson iteration (even with damping) 
may depend strongly on the ansatz for the self-energy.

The Green's function method obtained by the ansatz 
$\Sigma_{\mr{ans}}[G,v] = \Sigma^{(1)}[G,v]$ is known as the 
\emph{Hartree-Fock method}. One is likewise free to consider higher-order
approximations for the self-energy. However it should be 
emphasized that it is not obvious \emph{a priori} which order of 
approximation is optimal for a given problem specification. 
Some numerical comparison of 
methods will be undertaken in Section~\ref{sec:numer}, but further 
comparisons are a subject of detailed study to be left for future work.

It should be noted that once the Green's function is computed, 
it can be used to compute the internal energy of the system via the Galitskii-Migdal 
formula of Theorem \ref{thm:migdal}. It can also be used to compute the 
Gibbs free energy via the Luttinger-Ward formalism in a way to be 
explained below. More remarkably, the framework of Green's function methods 
can even be used, as in the dynamical mean field theory
(DMFT)~\cite{GeorgesKotliarKrauthEtAl1996}, to compute effective Hamiltonians 
on smaller subsystems (`fragments') that self-consistently account for their interaction 
with their environments.  This will be studied in future work.

\subsubsection{The GW approximation}
\label{sec:GW}
In order to provide a broader perspective on both Green's function methods and 
diagrammatic manipulations, we include here a diagrammatic derivation of the GW 
approximation~\cite{Hedin1965} for the self-energy, which corresponds 
to an ansatz for the self-energy yielded by a further summation over an \emph{infinite} subset of the 
bold diagrams. The GW method is the corresponding Green's function method.

This summation, which is represented graphically in Fig.~\ref{fig:GWsum}, includes 
the Hartree diagram, together with all of the so-called ring diagrams.

\begin{figure}[h]
  \begin{center}
    \includegraphics[width=0.9\textwidth]{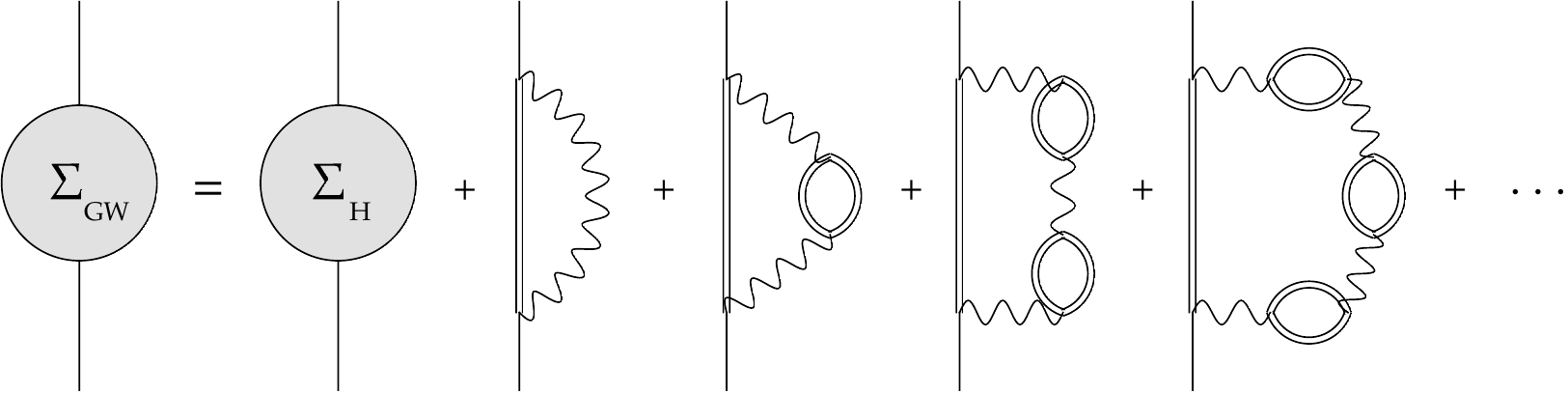}
  \end{center}
  \caption{Diagrammatic depiction of the GW self-energy.}
  \label{fig:GWsum}
\end{figure}

The Fock exchange diagram of Fig.~\ref{fig:feynmanSigmaG2} (b) can be thought of as 
the `zeroth' ring diagram, and Fig.~\ref{fig:feynmanSigmaG2} (c) is the first ring diagram. 
Notice that the $k$-th ring diagram $\Gamma_k$ has a symmetry factor of 
$S_{\Gamma_k} = 2^k$, with a factor of 
2 deriving from the each symmetry that exchanges the two propagators in one of the  
$k$ `bubbles.'

Furthermore, the $k$-the ring diagram has Feynman amplitude given by 
\[
\mathbf{F}_{\Gamma_k}(i,j) = - \left([-v (G\odot G)]^k v \right)_{ij} G_{ij}, 
\]
or equivalently, 
\[
\mathbf{F}_{\Gamma_k}= - G \odot \left([-v (G\odot G)]^k v \right), 
\]
where `$\odot$' indicates the entrywise matrix product.

Then formally summing the geometric series 
we obtain 
\[
\sum_{k=0}^\infty \frac{\mathbf{F}_{\Gamma_k}(i,j)}{S_{\Gamma_k}} 
=  - G \odot \left( \left[ I + \frac{1}{2} v (G\odot G) \right]^{-1}  v \right),
\]
where the factor $\frac12$ arises from the symmetry factor.
Incorporating the Hartree term, we arrive at the GW ansatz for the self-energy:
\[
\Sigma_{\mr{GW}}[G,v] = \Sigma_{\mr{H}}[G,v] - G \odot W[G,v],
\]
where 
\[
W[G,v] := \left( \left[ I + \frac{1}{2} v (G\odot G) \right]^{-1}  v \right) = \left[ v^{-1} + \frac{1}{2} (G\odot G) \right]^{-1},
\]
is known as the \emph{screened Coulomb interaction}. Thus the GW 
self-energy (whose name derives from the $G\odot W$ term appearing therein) 
looks very much like the Hartree-Fock self energy, but with the Fock exchange 
replaced by a \emph{screened exchange}, in which $W$ assumes the role of $v$.

\subsection{A preview of the Luttinger-Ward formalism}
\label{sec:LWpre}

There is in fact a more fundamental formalism underlying the 
self-energy ansatzes and Green's function methods outlined above, 
which can be recovered from ansatzes for the so-called 
\emph{Luttinger-Ward (LW) functional}.

It will turn out (as will be demonstrated in Part II) that the 
exact self-energy $\Sigma[G,v]$, viewed as a matrix-valued functional of 
$G$ for fixed $v$, can be written as the \emph{matrix derivative} of a scalar-valued 
functional of $G$, as in 
\begin{equation}
\label{eqn:LWpre}
\Sigma[G,v] = \frac{\partial \Phi}{\partial G} [G,v].
\end{equation}
Here $\Phi[G,v]$ is the LW functional, which additionally  
satisfies $\Phi[0,v] = 0$.

We must include some technical commentary to make precise sense of 
Eq.~\eqref{eqn:LWpre}. In fact, $\Phi[\,\cdot\,,v]$ should be thought of as a 
function $\pd \ra \R$, where $\pd$ is the set of symmetric positive definite 
$N\times N$ matrices. Then the derivative $\frac{\partial}{\partial G}$ 
should be defined in terms of variations \emph{within} $\pd$. Letting 
$E^{(ij)} \in \symm$ be defined by $E^{(ij)}_{kl} = \delta_{ik}\delta_{jl} + \delta_{il}\delta_{jk}$, 
we then define $\frac{\partial}{\partial G} := \left(\frac{\partial}{\partial G_{ij}}\right) := \frac12 \left(D_{E^{(ij)}}\right)$, 
where $D_{E^{(ij)}}$ indicates the directional derivative in the direction 
$E^{(ij)}$. If $f : \pd \ra \R$ is obtained by the restriction of a function 
$f: \R^{N\times N} \ra \R$ that is specified by a formula $X\mapsto f(X)$ 
in which the roles of $X_{kl}$ and $X_{lk}$ are the same for all $l,k$, then 
$\frac{\partial}{\partial G}$ simply coincides with the usual matrix 
derivative. This is the basic scenario underlying various computations below.

The 
LW functional relates to the free energy in the following manner.
Consider the free energy as a functional of $A$ and $v$ via 
$\Omega = \Omega[A,v]$. Then for $A$ and $v$ that  
yield Green's function $G$, we will derive in Part II the relation 
\begin{equation}
  \Omega[A,v] = \frac12 \Tr[AG] -\frac12 \Tr \log G - \frac12 \left( \Phi[G,v] + \Phi_0 \right),
  \label{eqn:omegaLW}
\end{equation}
where $\Phi_0 = N \log(2\pi e)$ is a constant.

Moreover, much like $\Sigma[G,v]$, for fixed $G$ the LW functional $\Phi[G,v]$ 
will admit an asymptotic series in the coupling constant with terms denoted 
$\Phi^{(k)}[G,v]$. It will follow, on the basis of Eq.~\eqref{eqn:LWpre}, that 
\begin{equation}
\label{eqn:PhiBold}
\Phi^{(k)}[G,v] = \frac{1}{2k} \mr{Tr} \left( G \Sigma^{(k)}[G,v] \right).
\end{equation}

Therefore $\Phi$ admits a bold diagrammatic expansion obtained 
by linking the external 
half-edges of every skeleton diagram with a bold propagator to obtain 
a closed diagram, as in Fig.~\ref{fig:feynmanSigmaG2}.

\begin{figure}[h]
  \begin{center}
    \includegraphics[width=0.55\textwidth]{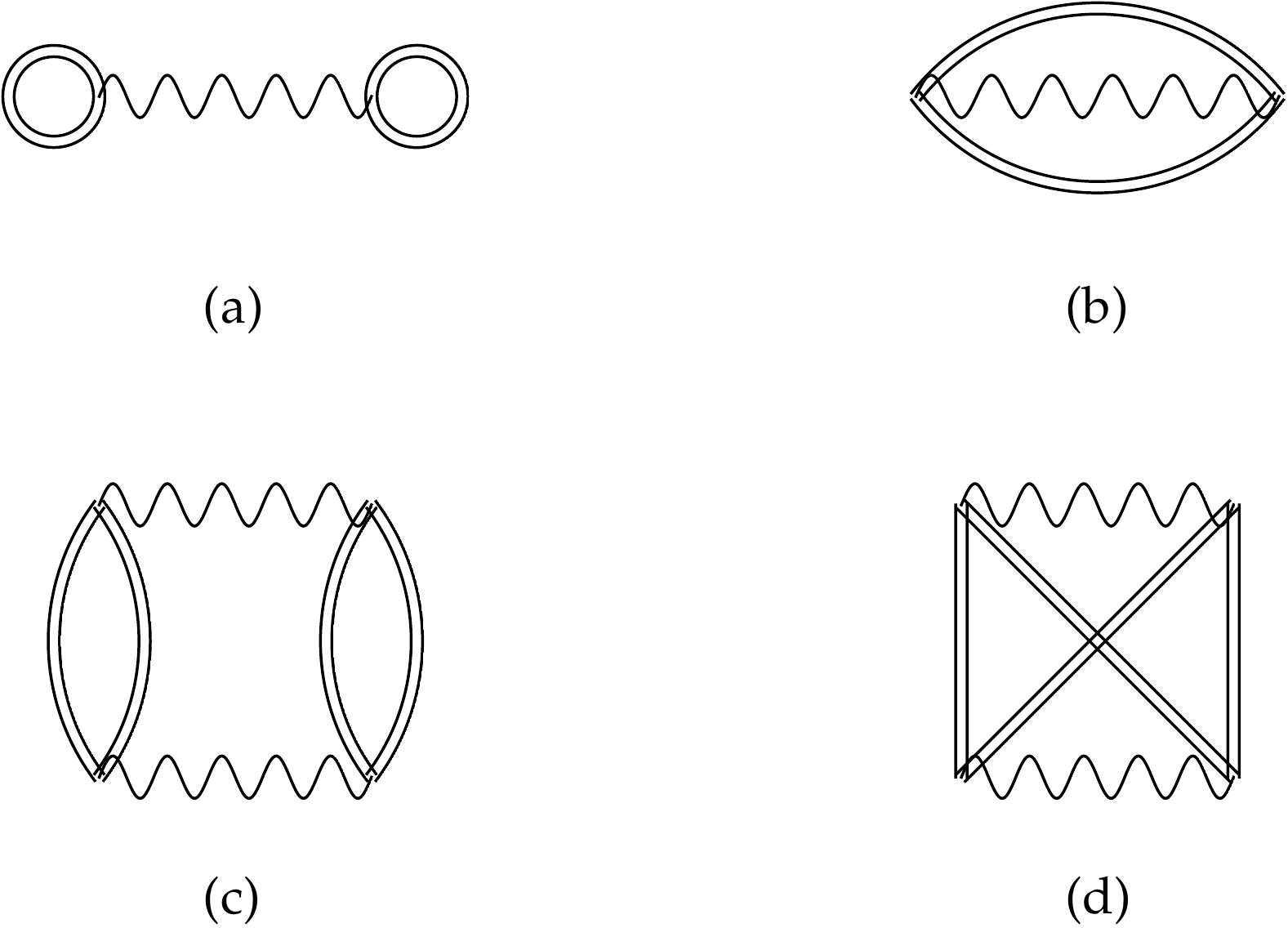}
  \end{center}
  \caption{Bold diagrammatic expansion of the LW functional up to second order.}
  \label{fig:feynmanLW2}
\end{figure}

It is important to realize that the pre-factors for these diagrams are obtained 
rather differently than the diagrams we have already seen (though the 
Feynman amplitudes are computed as usual). Indeed, for each 
LW diagram, one must sum over the prefactors of all skeleton diagrams 
from which it can be produced, then divide by $2k$, where $k$ is the number of 
interaction lines.

For example, we obtain 
\begin{equation}
\label{eqn:LW1bold}
\Phi^{(1)}[G,v] = - \frac{1}{4} \sum_{i,j} v_{ij} G_{ii} G_{jj} - \frac12 \sum_{i,j} v_{ij} G_{ij} G_{ij}
\end{equation}
and 
\begin{equation}
\label{eqn:LW2bold}
  \Phi^{(2)}[G] = 
  \frac18 \sum_{i,j,k,l} G_{ij}^2 v_{ik} G^2_{kl}
  v_{lj} + \frac14 \sum_{i,j,k,l} v_{ik} v_{jl} G_{ij} G_{kj} G_{kl} G_{li} ,
\end{equation}
for the first- and second-order contributions from Eqs.~\eqref{eqn:sigma1bold}
and \eqref{eqn:sigma2bold}, together with Eq.~\eqref{eqn:PhiBold}. Moreover 
we denote by $\Phi_{\mr{H}}[G,v]$ and $\Phi_{\mr{F}}[G,v]$ the first and second 
terms of Eq.~\eqref{eqn:LW1bold}, for `Hartree' and `Fock,' respectively.

\subsubsection{$\Phi$-derivability}

An ansatz $\Sigma_{\mr{ans}}[G,v]$ for $\Sigma[G,v]$ specifies a Green's function method. 
Among these ansatzes, some can be written as
matrix derivatives, i.e., can be viewed 
as being obtained from an ansatz $\Phi_{\mr{ans}}[G,v]$ for the  Luttinger-Ward functional via
\[
\Sigma_{\mr{ans}}[G,v] := \frac{\partial \Phi_{\mr{ans}}}{\partial G} [G,v].
\]
These approximations are known as \emph{$\Phi$-derivable}
or \emph{conserving} approximations. In the context of 
quantum many-body physics, the resulting $\Phi$-derivable Green's function 
methods are physically motivated in that they respect certain
conservation laws~\cite{StefanucciVanLeeuwen2013}.
 
Notice that for fixed $A$ and $v$, once an estimate $G_{\mr{ans}}$ for the Green's function 
has been computed via such a method, Eq.~\eqref{eqn:omegaLW} 
suggests a way to approximate the free energy, i.e., 
\begin{equation}
  \Omega \approx \frac12 \Tr[AG_{\mr{ans}}] -\frac12 \Tr \log (G_{\mr{ans}})- \frac12 \left( \Phi_{\mr{ans}}[G,v] + \Phi_0 \right).
  \label{eqn:omegaLW2}
\end{equation}

In fact all of the self-energy approximations considered thus far, 
namely the first-order (Hartree-Fock) 
and second-order approximations and the GW 
approximation, are $\Phi$-derivable. 

In fact the first- and second-order self-energy approximations of 
Eqs.~\eqref{eqn:sigma1bold} and Eqs.~\eqref{eqn:sigma2bold} can be 
obtained from the first- and second-order LW approximations of 
Eqs.~\eqref{eqn:LW1bold} and Eqs.~\eqref{eqn:LW2bold}.

Meanwhile, the GW approximation can be obtained from 
\[
\Phi_{\mr{GW}} [G,v] = \Phi_{\mr{H}} [G,v] - \Tr \log \left[ I + \frac12 v (G \odot G) \right].
\]
Here the matrix derivative of the first term yields the Hartree contribution $\Sigma_{\mr{H}}[G,v]$, and the matrix derivative 
of the second term yields $G\odot W[G,v]$.

In fact, $\Phi_{\mr{GW}}$ can be viewed as being obtained from an infinite summation 
of the (closed) ring diagrams from the bold diagrammatic expansion of the LW functional, 
together 
with the closed Hartree diagram. These are the obtained by closing up the diagrams of 
Fig.~\ref{fig:GWsum} and are themselves depicted in Fig.~\ref{fig:GWsumLW}.

\begin{figure}[h]
  \begin{center}
    \includegraphics[width=1.0\textwidth]{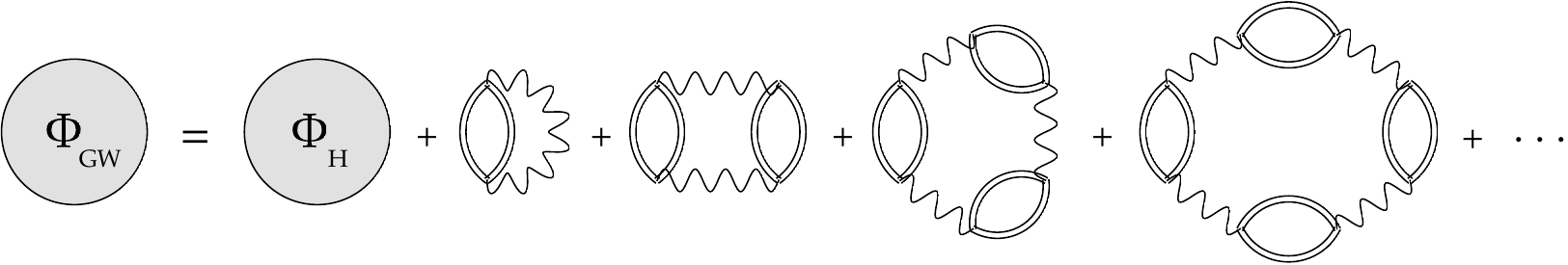}
  \end{center}
  \caption{Diagrammatic depiction of the GW Luttinger-Ward functional.}
  \label{fig:GWsumLW}
\end{figure}

What kind of (reasonable) self-energy approximation would \emph{fail} to be $\Phi$-derivable?
Roughly speaking, by the usual product rule for derivatives, taking the matrix derivative in $G$ of a LW diagram 
yields a sum over all skeleton diagrams that can be obtained from the LW diagram 
by cutting a single propagator. This means that if one includes some skeleton diagram 
in the approximate self-energy but \emph{does not} include another skeleton diagram 
that closes up to the same LW diagram as the first, then the approximation should 
not be $\Phi$-derivable.

Notice that for each of the diagrams in Fig.~\ref{fig:feynmanLW2}, any choice of 
the propagator line to be removed yields the same skeleton diagram, so the 
above scenario cannot apply to these terms. This is 
one way of seeing why each of the bold skeleton diagrams up to second 
order is \emph{individually} $\Phi$-derivable. (The same is true for each of the ring 
diagrams.)

To find a non-$\Phi$-derivable bold skeleton diagram, we have to go to 
the third order. In Fig.~\ref{fig:notPhi} (a), we depict a LW diagram that can 
be cut in different ways to obtain the distinct skeleton diagrams of
Fig.~\ref{fig:notPhi} (b) and (c).

\begin{figure}[h]
  \begin{center}
    \includegraphics[width=0.7\textwidth]{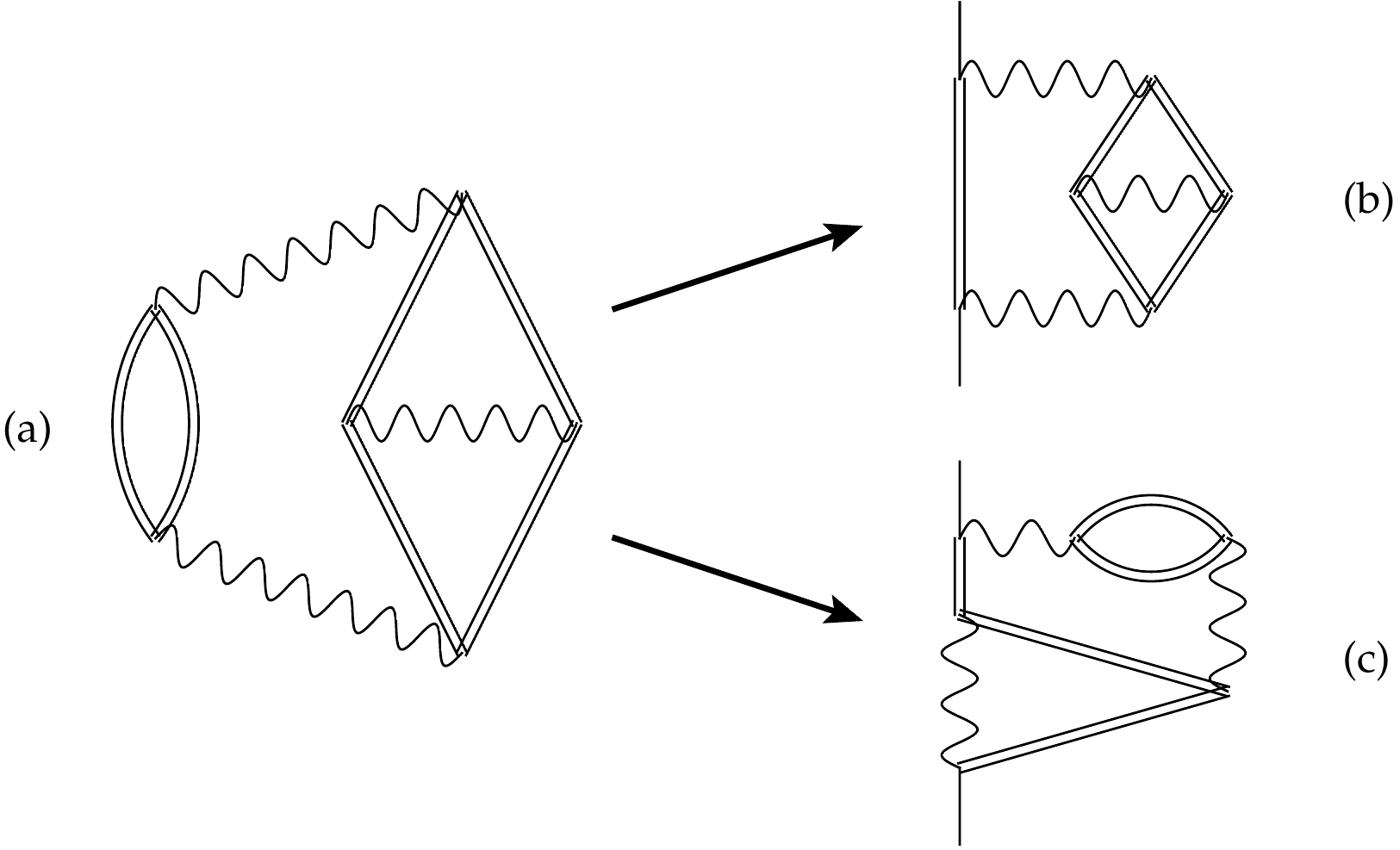}
  \end{center}
  \caption{A LW diagram (a) that yields distinct classes of skeleton diagram (b), (c). (External labelings of 
  the skeleton diagrams are ignored.)}
  \label{fig:notPhi}
\end{figure}

Although some skeleton diagrams such as Fig.~\ref{fig:notPhi} (b) and (c) 
are not \emph{individually} $\Phi$-derivable, the sum over all skeleton 
diagrams of a given order \emph{is} $\Phi$-derivable.

\subsection{Quantities that \emph{do not} admit a bold diagrammatic expansion}
\label{sec:noBold}

We stress that here the bold diagrammatic expansion has only been established 
for the self-energy (and, by extension, the Luttinger-Ward functional). 
The same concept cannot be generalized to other
quantities without verification. 

For example, let us consider what goes wrong in the case of the free energy. 
Theorem \ref{thm:logZexpand} expresses the free energy as a sum over 
bare connected closed diagrams. Why can't we just replace the thin lines 
with bold lines and then remove redundant diagrams (i.e., the diagrams 
with Green's function insertions)?

First let us see what happens if we try this. Later we will discuss where 
the proof of the bold diagrammatic expansion of the self-energy breaks 
down when we attempt to adapt it to the free energy.

Recall that the first-order free energy diagrams are simply the Hartree and 
Fock diagrams (i.e. the dumbbell and the oyster). Neither of these 
admit Green's function insertions, so converting them bold diagrams 
and retain them.

Now recall that the second-order bare free energy diagrams 
are depicted in Fig.~\ref{fig:feynmanLink2} (b), (c). In particular, 
consider diagram (b2), reproduced below in Fig.~\ref{fig:feynmanZ_decompose} (a). 
This diagram admits Green's function insertions, 
hence will not be retained as a bold diagram.

It can be obtained by inserting Green's function diagrams into the bold 
oyster diagram, shown in Fig.~\ref{fig:feynmanZ_decompose} (b), in \emph{two} different ways; 
indeed, either of the bold propagators in diagram (b) 
can be replaced with the insertion depicted in Fig.~\ref{fig:feynmanZ_decompose} (c).

Since diagram (b) has a symmetry factor of 4 and (c) has a symmetry factor of 1, 
our attempted bold diagrammatic 
expansion for the free energy then counts diagram (a) with a pre-factor of $\frac{1}{2}$. However, 
diagram (a) has a symmetry factor of 4, hence has a pre-factor of $\frac14$ in the bare 
diagrammatic expansion for the free energy!

\begin{figure}[h]
  \begin{center}
    \includegraphics[width=0.75\textwidth]{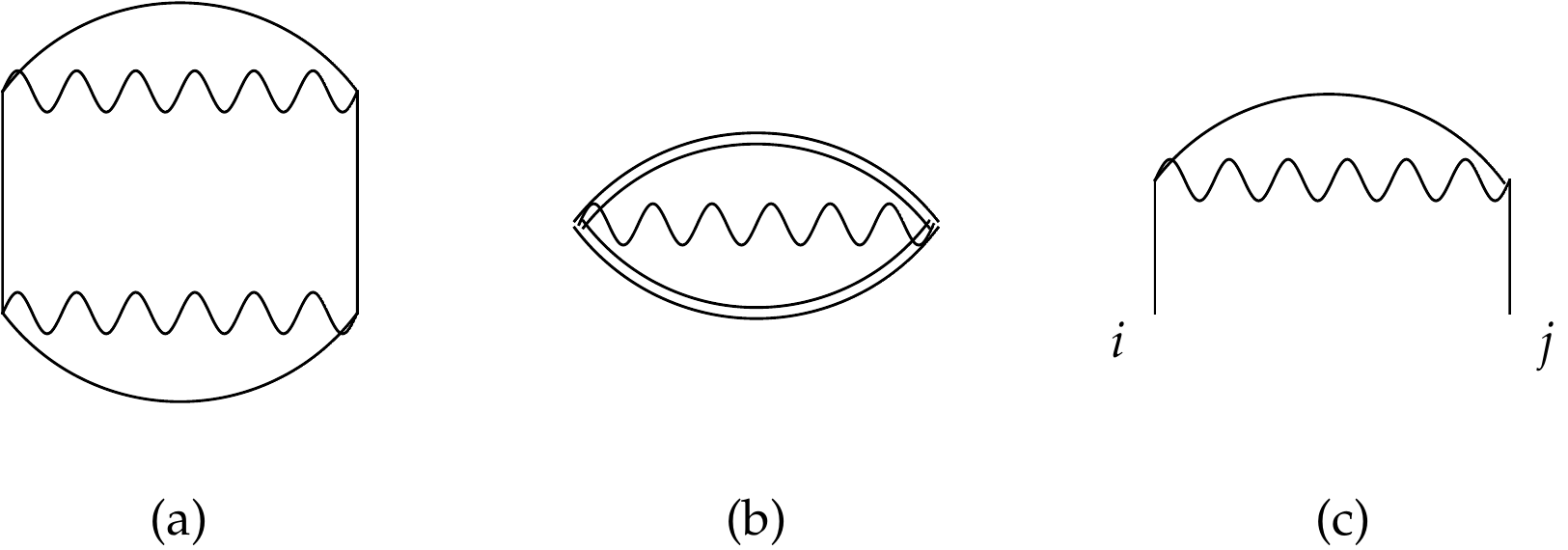}
  \end{center}
  \caption{Decomposition of a closed diagram $\Gamma$ in (a) into
  a closed diagram $\Gamma_{s}$ in (b), and a truncated Green's function diagram
  $\Gamma_{g}$ in (c).}
  \label{fig:feynmanZ_decompose}
\end{figure}

What went wrong? The problem is that there is no analog of the unique 
skeleton decomposition (Proposition \ref{prop:selfenergycompose}) 
for \emph{closed} diagrams.\footnote{The proof of 
Proposition \ref{prop:selfenergycompose} fails for closed diagrams 
in that the case $\vert E^{(j,k)} \vert =0 $ cannot be ruled out. (In the 
original proof, this case could be ruled out because it implied the existence 
of a closed subdiagram of a connected self-energy diagram, which is 
impossible.)} Indeed, the diagram of 
Fig.~\ref{fig:feynmanZ_decompose} (a) can be built up 
from two different `skeleton' subdiagrams, one containing the 
upper interaction line and the other containing the lower one.

\section{Numerical experiments}\label{sec:numer}
For the Gibbs model, in contrast to the quantum many-body problem, 
Green's function methods as in section \ref{sec:greenMethod} 
can be implemented within a few lines of MATLAB code. 
In this section we provide a small snapshot of the application of
the Gibbs model to the investigation of the numerical performance of MBPT. 

In particular, we demonstrate that the use of the LW functional to compute the free energy via Eq.~\eqref{eqn:omegaLW2} 
can yield more accurate results than the use of the bare diagrammatic expansion of 
Theorem \ref{thm:logZexpand}.

For simplicity we consider dimension $N=4$, where all integrals can be
evaluated directly via a quadrature scheme, such as Gauss-Hermite
quadrature.  The quadratic and the quartic terms of the Hamiltonian 
are specified by 
\begin{equation}
  A = \begin{pmatrix}
  2 & -1 & 0  & 0\\
  -1 & 2 & -1 & 0\\
  0 & -1 & 2 & -1\\
  0 & 0 & -1 & 2\\
  \end{pmatrix},
  \quad 
  v = \lambda \begin{pmatrix}
  1 & 0 & 0  & 0\\
  0 & 1 & 0 & 0\\
  0 & 0 & 1 & 0\\
  0 & 0 & 0 & 1\\
  \end{pmatrix},
  \label{}
\end{equation}
respectively. 

Once the self-consistent Green's function $G$ is obtained from 
some $\Phi$-derivable Green's function method, we
evaluate the free energy using the LW functional via Eq.~\eqref{eqn:omegaLW2}. 
Since the non-interacting free energy $\Omega_0$ is not physically 
meaningful (as it corresponds to an additive constant in the 
Hamiltonian), we measure the relative error of $\Omega - \Omega_0$, 
i.e., we compute 
$\vert \Omega_{\mr{ans}} - \Omega_{\mr{exact}} \vert / \vert \Omega_{\mr{exact}} - \Omega_0 \vert $, 
where $\Omega_{\mr{ans}}$ is the free energy obtained via 
our approximation and $\Omega_{\mr{exact}}$ is the exact free energy.
We carry out this procedure for the first-order (Hartree-Fock) and
second-order approximations (GF2) of the LW functional, 
denoted `Bold 1st' and `Bold 2nd,' respectively, 
in Fig.~\ref{fig:omegaN4}, which plots the relative 
error against the coupling 
constant $\lambda$.

For comparison, we also consider the approximations of the free energy obtained 
directly from 
the first-order and second-order truncations of the bare diagrammatic 
expansion for the free energy of Theorem \ref{thm:logZexpand}, denoted 
`Bare 1st' and `Bare 2nd,' respectively, in Fig.~\ref{fig:omegaN4}.

\begin{figure}[h]
  \begin{center}
    \includegraphics[width=0.8\textwidth]{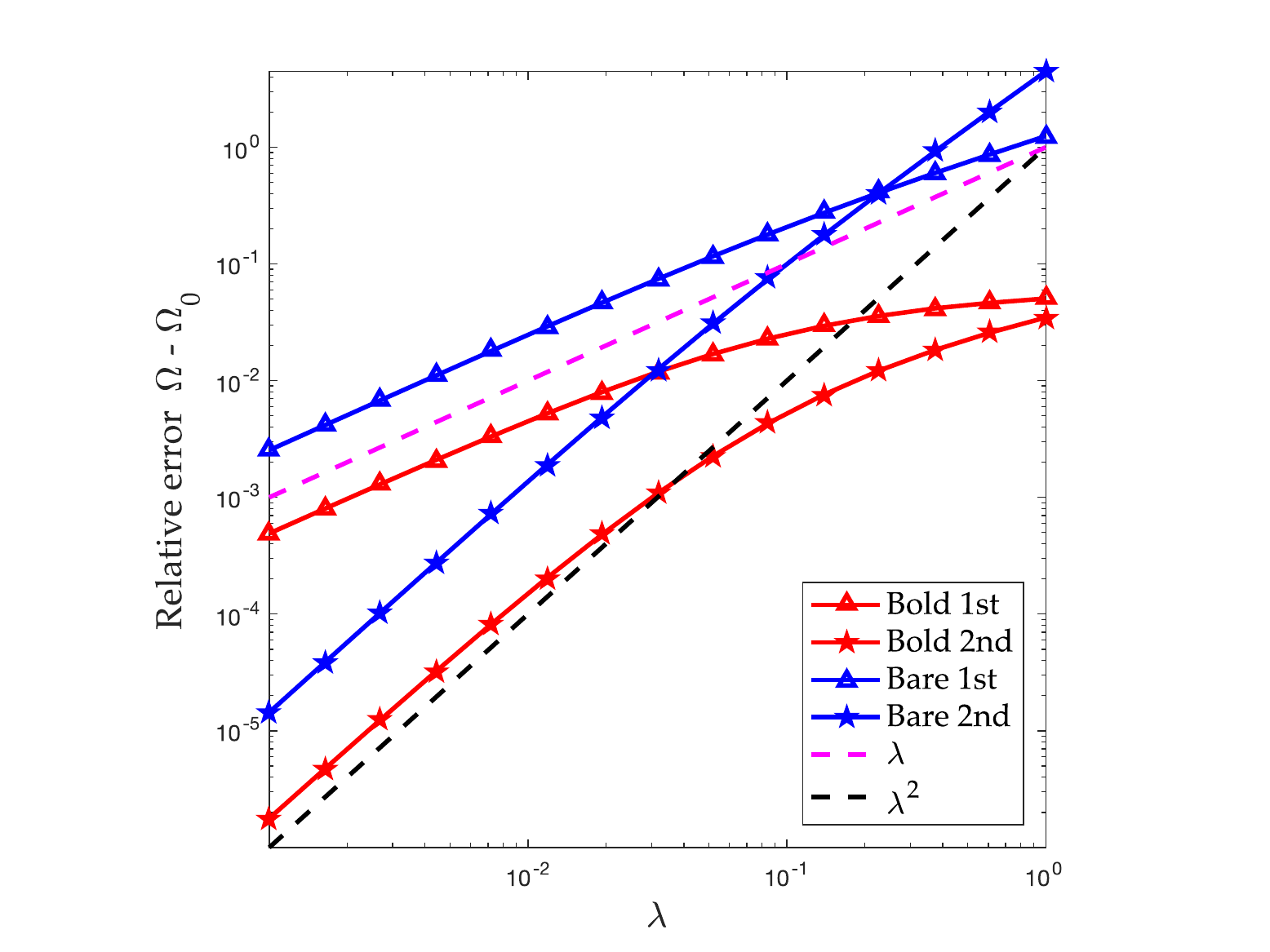}
  \end{center}
  \caption{Comparison of different approximation schemes for the free
  energy.}
  \label{fig:omegaN4}
\end{figure}

Observe that the relative errors of the first- and second-order 
expansions scale as $\lambda$ and $\lambda^2$, respectively, 
as $\lambda \ra 0$. This makes sense because the `truncation' error 
of these expansions should be of order $\lambda^2$ and $\lambda^3$, 
respectively, but to obtain the relative error we are dividing by 
$\Omega_{\mr{exact}}-\Omega_0$, which is of order $\lambda$ as 
$\lambda \ra 0$.

Note, however, that the pre-constant of the scaling is more
favorable for the bold method in both cases. Furthermore, the bold 
methods scale more gracefully than their bare counterparts when 
$\lambda$ is relatively large.

To demonstrate the simplicity of the implementation, 
we provide below a MATLAB code for computing the self-consistent Green's function 
and the free energy using the first- and second-order bold diagrammatic expansions. 
The exact solution is evaluated directly using a quadrature code which is omitted here.

\vspace{4 mm}
\begin{verbatim}
% luttingerward.m
d = 4;
Phi0 = d*(log(2*pi)+1);
A = [ 2 -1  0  0
     -1  2 -1  0;
      0 -1  2 -1;
      0  0 -1 2];
Umat = 0.1 * eye(d);


% First order bold diagram
maxiter = 100;
G = inv(A);
for iter = 1 : maxiter
  rho = diag(G);
  Sigma = -1/2*diag(Umat * rho) - (Umat.*G);
  GNew = inv(A-Sigma);
  err = norm(G-GNew)/norm(G);
  if( norm(G-GNew)/norm(G) < 1e-10 ) break; end
  G = GNew;
end
Phi = 1/2*trace(Sigma*G);
Omega = 0.5*(trace(A*G) - log(det(G)) - (Phi + Phi0));
fprintf('Free energy 1st order = %g\n', Omega);

% Second order bold diagram
G = inv(A);
for iter = 1 : maxiter
  rho = diag(G);
  % First order term
  Sigma1 = -1/2*diag(Umat * rho) - (Umat.*G);
  % Ring diagram
  Sigma2 = 1/2 * G.*(Umat * (G.*G) * Umat);
  % Second order exchange diagram
  for i = 1 : d
    for j = 1 : d
      Sigma2(i,j) = Sigma2(i,j) + (Umat(:,i).*G(:,j))'*G*(Umat(:,j).*G(:,i));
    end
  end
  Sigma = Sigma1 + Sigma2;
  GNew = inv(A-Sigma);
  nrmerr = norm(G-GNew)/norm(G);
  if( norm(G-GNew)/norm(G) < 1e-10 ) break; end
  G = GNew;
end
Phi = 1/2*trace(Sigma1*G) + 1/4*trace(Sigma2*G);
Omega = 0.5*(trace(A*G) - log(det(G)) - (Phi + Phi0));
fprintf('Free energy 2nd order = %g\n', Omega);

OmegaExact = -2.7510737;
fprintf('Free energy exact     = %g\n', OmegaExact);


>> luttingerward
Free energy 1st order = -2.74745
Free energy 2nd order = -2.75209
Free energy exact     = -2.75107
\end{verbatim}
\vspace{4 mm}


\section*{Acknowledgments} 

This work was partially supported by the Department of Energy under
grant DE-AC02-05CH11231 (L.L., M.L.), by the Department of Energy under
grant No. DE-SC0017867 and by the Air Force Office of Scientific
Research under award number FA9550-18-1-0095 (L.L.), and by the NSF
Graduate Research Fellowship Program under Grant DGE-1106400 (M.L.).  We
thank Fabien Bruneval, Garnet Chan, Alexandre Chorin, Lek-Heng Lim,
Nicolai Reshetikhin, Chao Yang and Lexing Ying for helpful discussions.


\appendix

\section{Proof of Theorem \ref{thm:wick} (Isserlis-Wick)}
\label{sec:appWick}
\begin{proof}
  For any $b\in \RR^{N}$, define
  \begin{equation}
    Z_{0}(b) = \int_{\RR^{N}} e^{-\frac12 x^{T} Ax + b^{T}x} \ud x.
    \label{eqn:Z0b}
  \end{equation}
  Clearly $Z_{0}=Z_{0}(0)$.
  Performing the change of variable $\wt{x}=x+A^{-1}b$, we find
  \begin{equation}
    Z_{0}(b) = e^{\frac12 b^{T} A^{-1} b} Z_{0}.
    \label{eqn:Z0beval}
  \end{equation}
  Observe that, for integers $1\le \alpha_{1},\ldots,\alpha_{m}\le N$,
  \begin{equation}
    \frac{\partial^{m} Z_{0}(b)}{\partial b_{\alpha_{1}}\cdots \partial b_{\alpha_{m}}}\Big\vert_{b=0}
    = \int_{\RR^{N}} x_{\alpha_{1}}\cdots x_{\alpha_{m}}
    e^{-\frac12 x^{T} Ax} \ud x.
    \label{}
  \end{equation}
  Combining with Eq.~\eqref{eqn:Z0beval}, we have that
  \begin{equation}
    \average{x_{\alpha_{1}}\cdots x_{\alpha_{m}}}_{0} =
    \frac{\partial^{m} }{\partial b_{\alpha_{1}}\cdots \partial
    b_{\alpha_{m}}} e^{\frac12 b^{T} A^{-1} b} \Big\vert_{b=0}.
    \label{eqn:wick1}
  \end{equation}
  Now write the Taylor expansion
  \begin{equation}
    e^{\frac12 b^{T} A^{-1} b} = \sum_{n=0}^{\infty} \frac{1}{n! \, 2^{n}}
    \left( b^{T}A^{-1} b \right)^{n}. 
    \label{}
  \end{equation}
  Since the expansion contains no odd powers of the components of $b$, 
  the right-hand side of Eq.~\eqref{eqn:wick1} vanishes whenever $m$ is odd. If
  $m$ is an even number, only the term $n=m/2$ in the Taylor expansion
  will contribute to the right-hand side of Eq.~\eqref{eqn:wick1}. This gives
  \begin{equation} 
   \frac{\partial^{m} }{\partial b_{\alpha_{1}}\cdots \partial
    b_{\alpha_{m}}} e^{\frac12 b^{T} A b} \Big\vert_{b=0}
    = \frac{1}{2^{m/2} (m/2)!}  \frac{\partial^{m} }{\partial b_{\alpha_{1}}\cdots \partial
    b_{\alpha_{m}}} \left( b^{T}A^{-1} b \right)^{m/2}.
    \label{eqn:wickDeriv}
  \end{equation}
  Now one can write 
  \[
   \frac{\partial^{m} }{\partial b_{\alpha_{1}}\cdots \partial
    b_{\alpha_{m}}}
    = 
     \frac{\partial^{m} }{\partial^{m_1} b_{\beta_{1}}\cdots \partial^{m_p}
    b_{\beta_{p}}}
  \]
  where the indices $\beta_1,\ldots,\beta_p$ are distinct and $\sum_{j=1}^p m_j = m$. Then to further 
  simplify the right-hand side of Eq.~\eqref{eqn:wickDeriv}, we are interested in computing the 
  coefficient of the
  $b_{\beta_1}^{m_1} \cdots b_{\beta_p}^{m_p}$ term of the polynomial $(b^T A^{-1} b)^{m/2}$.  
  Expand $(b^T A^{-1} b)^{m/2}$ into a sum of monomials as 
  
  \[
   (b^T A^{-1} b)^{m/2} = \sum_{j_1, k_1, \ldots, j_{m/2}, k_{m/2} } \prod_{i=1}^{m/2} (A^{-1})_{j_i, k_i} b_{j_i} b_{k_i}.
  \]
  Thus each distinct permutation $(j_1, k_1, \ldots, j_{m/2}, k_{m/2})$ of the multiset
  $\{\alpha_1, \ldots ,\alpha_m \}$\footnote{For instance,
  if $\alpha_1 = 1, \alpha_2 = 1, \alpha_3 = 2$, there are only three distinct multiset permutations: $(1,1,2)$, $(1,2,1)$, and $(2,1,1)$.}
  contributes $\prod_{i=1}^{m/2} (A^{-1})_{j_i, k_i}$ to our desired coefficient.
  
  Since 
  \[
  \frac{\partial^{m} }{\partial^{m_1} b_{\beta_{1}}\cdots \partial^{m_p}
    b_{\beta_{p}}} b_{\beta_1}^{m_1} \cdots b_{\beta_p}^{m_p} = m_1! \cdots m_p!,
  \]
  we obtain from Eq.~\eqref{eqn:wick1} and Eq.~\eqref{eqn:wickDeriv}
  \begin{equation}
  \label{eqn:wickMultisetSum}
  \average{x_{\alpha_{1}}\cdots x_{\alpha_{m}}}_{0} 
    = 
    \frac{m_1! \cdots m_p!}{2^{m/2} (m/2)!} 
    \sum_{(j_1, k_1, \ldots, j_{m/2}, k_{m/2})}
    \prod_{i=1}^{m/2} (A^{-1})_{j_i, k_i},
  \end{equation}
  where the sum is understood to be taken over multiset permutations of $\{\alpha_1, \ldots ,\alpha_m \}$.

  Now every permutation of $\mc{I}_m$ can be thought of as inducing a permutation 
  of the multiset $\{\alpha_1, \ldots ,\alpha_m \}$ via 
  the map $(l_1, \ldots, l_m) \mapsto (\alpha_{l_1}, \ldots, \alpha_{l_m})$. By this map each multiset 
  permutation of $\{\alpha_1, \ldots ,\alpha_m \}$ is associated with $m_1! \cdots m_p!$ different permutations 
  of $\mc{I}_m$.
  
  Moreover, every permutation of $\mc{I}_m$ can be thought of as inducing
  a pairing on $\mc{I}_m$ 
  by pairing the first two indices in the permutation, then the next two, etc. For example, if $m=4$, then 
  the permutation $(1,4,3,2)$ induces the pairing $(1,4)(3,2)$. By this map, each pairing on $\mc{I}_m$ 
  is associated with $2^{m/2} (m/2)!$ permutations of $\mc{I}_m$ since a pairing does not distinguish 
  the ordering of the pairs, nor the order of the indices within each pair.
  
  Finally, observe that if we take any two permutations of $\mc{I}_m$ associated with a pairing on
  $\mc{I}_m$ and then consider the (possibly distinct) multiset permutations 
   $(j_1, k_1, \ldots, j_{m/2}, k_{m/2})$ and  $(j_1', k_1', \ldots, j_{m/2}', k_{m/2}')$ of $\{\alpha_1, \ldots ,\alpha_m \}$ 
  associated to these permutations, then in fact $\prod_{i=1}^{m/2} (A^{-1})_{j_i, k_i} = \prod_{i=1}^{m/2} (A^{-1})_{j_i', k_i'}$.
  
  Thus we can replace the sum over multisets in Eq.~\eqref{eqn:wickMultisetSum} with a sum over 
  pairings on $\mc{I}_m$. Each term must be counted $\frac{2^{m/2} (m/2)!}{m_1! \cdots m_p!}$ times 
  because each pairing on $\mc{I}_m$ induces $2^{m/2} (m/2)!$ permutations of $\mc{I}_m$, each of 
  which is redundant by a factor of $m_1! \cdots m_p!$. This factor cancels with the factor in 
  Eq.~\eqref{eqn:wickMultisetSum} to yield
  \[
  \average{x_{\alpha_{1}}\cdots x_{\alpha_{m}}}_{0} = \sum_{\sigma \in \Pi(\mc{I}_{m})} \prod_{i \in \mc{I}_{m}/\sigma}
    (A^{-1})_{\alpha_{i},\alpha_{\sigma(i)} }.
  \]
  Recalling that $G^{0}=A^{-1}$, this completes the proof of Theorem~\ref{thm:wick}.
\end{proof}

\begin{remark}
  In field theories, the auxiliary variable $b$ is often interpreted as a
  coupling external field.
\end{remark}

\section{Proof of Proposition \ref{prop:selfenergycompose} (Skeleton decomposition)}
\label{sec:appSkeleton}
 \begin{proof}
 As suggested in the statement of the proposition, let $\Gamma^{(k)}$ be the maximal insertions 
 admitted by $\Gamma$, where we do not count separately any insertions 
 that differ only by their external labeling (see Remark \ref{rem:uniqueness}). Let 
 $h_1^{(k)}$ and $h_2^{(k)}$ be half edges such that $\Gamma$ admits the insertion $\Gamma^{(k)}$ at 
 $(h_1^{(k)},h_2^{(k)})$. Let $e_1^{(k)},e_2^{(k)}$ be the external half-edges of the truncated 
 Green's function diagram $\Gamma^{(k)}$ paired with $h_1^{(k)},h_2^{(k)}$, respectively, 
 in the overall diagram $\Gamma$.
 
  First we aim to show that the $\Gamma^{(k)}$ are disjoint, i.e., share no half-edges. In this case 
  we say that the $\Gamma^{(k)}$ do not \emph{overlap}. In fact we 
  will show additionally that the external half-edges of the $\Gamma^{(k)}$ do not touch one another 
  (i.e., are not paired in $\Gamma$), and accordingly we say that the $\Gamma^{(k)}$ do not \emph{touch}.
  
  To this end, let $j \neq k$. The idea is to consider the diagrammatic structure formed by collecting 
  all of the half-edges appearing in $\Gamma^{(j)}$ and $\Gamma^{(k)}$ and then arguing by maximality 
  that $\Gamma^{(j)}$ and $\Gamma^{(k)}$ cannot overlap or touch, let this structure admit 
  $\Gamma^{(j)}$ and $\Gamma^{(k)}$ as insertions.
  
 We now formalize this notion. Let $H^{(j)}$ and $H^{(k)}$ be the half-edge sets of 
 $\Gamma^{(j)}$ and $\Gamma^{(k)}$, respectively, and consider the union 
 $H^{(j,k)} := H^{(j)} \cup H^{(k)}$, together with 
 a \emph{partial} pairing $\Pi^{(j,k)}$ on $H^{(j,k)}$, i.e., a collection of disjoint pairs in $H^{(j,k)}$, 
 defined by the rule that $\{h_1,h_2\} \in \Pi^{(j,k)}$ if and only if $\{h_1,h_2\} \in \Pi_{\Gamma}$ and 
 $h_1,h_2 \in H^{(j,k)}$. In other words, the structure $(H^{(j,k)},\Pi^{(j,k)})$ is formed by taking all 
 of the half-edges appearing in $\Gamma^{(j)}$ and $\Gamma^{(k)}$ and then pairing the ones that 
 are paired in the overall diagram $\Gamma$.
 
 Let $E^{(j,k)}$ be the subset of unpaired half-edges 
 in $H^{(j,k)}$, i.e., those half-edges that \emph{do not} appear in any pair in $\Pi^{(j,k)}$. Since 
 all half-edges in $\Gamma^{(k)}$ besides $e_1^{(k)},e_2^{(k)}$ are paired in the diagram 
 $\Gamma^{(k)}$, we must have $E^{(j,k)} \subset \{ e_1^{(j)},e_2^{(j)}, e_1^{(k)},e_2^{(k)} \}$.
 
 \vspace{2mm}
 
 \begin{lemma}
 \label{lem:sizeE}
 $\vert E^{(j,k)} \vert = 4$. 
 \end{lemma}
 \begin{proof}
 We claim that $\vert E^{(j,k)} \vert$ is even. To see this, we first establish that $\vert H^{(j,k)} \vert$ 
 is even. Notice that a truncated Green's function diagram (in particular, $\Gamma^{(j)}$ and 
 $\Gamma^{(k)}$) contains an even number of half-edges (more specifically $4n$, 
 where $n$ is the order of the diagram). Thus $\vert H^{(j)} \vert, \vert H^{(k)}\vert $ are even. 
 Moreover, $\Gamma^{(j)}$ and 
 $\Gamma^{(k)}$ share a half-edge if and only if they share the vertex (or interaction line) 
 associated with this half-edge, in which case $\Gamma^{(j)}$ and 
 $\Gamma^{(k)}$ share all four half-edges emanating from this vertex. Thus the number
 $\vert H^{(j)} \cap H^{(k)} \vert $ of 
 half-edges common to $\Gamma^{(j)}$ and $\Gamma^{(k)}$ is four times the number of
 common interaction lines, in particular an even number. So 
 $\vert H^{(j,k)} \vert = \vert H^{(j)} \vert + \vert H^{(k)} \vert - \vert H^{(j)} \cap H^{(k)} \vert$ is 
 even, as desired. Now any partial pairing on $H^{(j,k)}$ includes an even number of 
 distinct elements, so the number of leftover elements, i.e., $\vert E^{(j,k)} \vert$, 
 must be even as well, as claimed.
 
 Next we rule out the cases $\vert E^{(j,k)} \vert \in \{ 0, 2 \}$.
 
 Suppose that $\vert E^{(j,k)} \vert = 0$. Then the structure $(H^{(j,k)},\Pi^{(j,k)})$ 
 defines a \emph{closed} sub-diagram of $\Gamma$, disconnected from the 
 rest of $\Gamma$. Since $\Gamma$ is not closed, 
 the sub-diagram cannot be all of $\Gamma$. This conclusion contradicts the 
 connectedness of $\Gamma$.\footnote{Interestingly, if one were to try to 
 construct a unique skeleton decomposition for \emph{closed} connected diagrams, i.e., 
 free energy diagrams, this is the place where the argument would fail; see 
 section \ref{sec:noBold}.} 
 
 Next suppose that $\vert E^{(j,k)} \vert = 2$. Then the structure $(H^{(j,k)},\Pi^{(j,k)})$ 
 has two unpaired half-edges, hence defines a truncated Green's function diagram 
 $\Gamma^{(j,k)}$ that contains both $\Gamma^{(j)}$ and $\Gamma^{(k)}$
 and admits both as insertions. But since $\Gamma^{(j)} \neq \Gamma^{(k)}$, 
 $\Gamma^{(j,k)}$ is neither $\Gamma^{(j)}$ nor $\Gamma^{(k)}$ (i.e., the containment 
 is proper). This conclusion contradicts the maximality of 
 $\Gamma^{(j)}$ and $\Gamma^{(k)}$.
 \end{proof}
 
  \vspace{2mm}
 
 From Lemma \ref{lem:sizeE} it follows that 
 $E^{(j,k)} = \{ e_1^{(j)},e_2^{(j)}, e_1^{(k)},e_2^{(k)} \}$. (And moreover 
 $ e_1^{(j)},e_2^{(j)}, e_1^{(k)},e_2^{(k)}$ are distinct, which was not assumed!) 
 This establishes 
 one of our original claims, i.e., that the external half-edges of $\Gamma^{(j)}$ and 
 $\Gamma^{(k)}$ do not touch (i.e., are not paired.)
 
 We need to establish the other part of our original claim, i.e., that the half-edge 
 sets $H^{(j)}$ and $H^{(k)}$ are disjoint. To see this, suppose otherwise, 
 so $\Gamma^{(j)}$ and $\Gamma^{(k)}$ share some half-edge $h$. 
 Since $\Gamma^{(j)} \neq \Gamma^{(k)}$, one of $H^{(j)}$ 
 and $H^{(k)}$ does not contain the other, so assume without loss of 
 generality that $H^{(j)}$ does not contain $H^{(k)}$, so there is some 
 half-edge $h' \in H^{(k)} \backslash H^{(j)}$. Since $\Gamma^{(k)}$ 
 is connected there 
 must be some path in $\Gamma^{(k)}$ connecting $h$ with $h'$.\footnote{By 
 such a `path' we mean a sequence of half-edges 
 $(h_1,h_2,\ldots,h_{2m-1},h_{2m})$ in $\Gamma^{(k)}$ such that $h_1 = h$; $h_{2m} = h'$;  
 $h_l,h_{l+1}$ share an interaction line for $l$ odd; and $h_l,h_{l+1}$ 
 are paired by $\Gamma^{(k)}$.}
 Now 
 $\Gamma^{(j)}$ can be disconnected from the rest of $\Gamma$ by 
 deleting the links $\{e_1^{(j)},h_1^{(j)}\}$ and $\{e_2^{(j)},h_2^{(j)}\}$ from 
 the pairing $\Pi_{\Gamma}$, so evidently our path in $\Gamma^{(k)}$ 
 connecting $h$ and $h'$ 
 must contain either $(e_1^{(j)},h_1^{(j)})$ or $(e_2^{(j)},h_2^{(j)})$ 
 as a `subpath.' But then either $e_1^{(j)}$ or $e_2^{(j)}$ is paired in $\Gamma^{(k)}$, 
 hence also by $\Pi^{(j,k)}$, contradicting its inclusion in $E^{(j,k)}$.

 In summary we have shown that the $\Gamma^{(k)}$ are disjoint, i.e., 
 share no half-edges, and that the external half-edges of the $\Gamma^{(k)}$ 
 do not touch one another, i.e., are not paired in $\Gamma$. We can 
 define a truncated Green's function diagram $\Gamma_{\mr{s}}$ by considering
 $\Gamma$, then omitting all half-edges and vertices appearing in the 
 $\Gamma^{(k)}$. Since the $\Gamma^{(k)}$ are disjoint and do not touch, 
 this leaves behind the half-edges $h_1^{(k)},h_2^{(k)}$ for all $k$, which 
 are now left unpaired. Then we complete the construction of 
 $\Gamma_{\mr{s}}$ by adding the pairings $\{h_1^{(k)},h_2^{(k)}\}$. In short, 
 we form $\Gamma_{\mr{s}}$ from $\Gamma$ by replacing each insertion 
 $\Gamma^{(k)}$ with the edge $\{h_1^{(k)},h_2^{(k)}\}$. Eq.~\eqref{eqn:skeletonInsertion} 
 then holds by construction, for a suitable 
external labeling of the $\Gamma^{(k)}$.

Moreover, we find that $\Gamma_{\mr{s}}$ 
is 2PI. It is not hard to check first that $\Gamma_{\mr{s}}$ is 1PI. Indeed, the 
unlinking of any half-edge pair in $\Gamma_{\mr{s}}$  that is \emph{not} one of 
the $\{h_1^{(k)},h_2^{(k)}\}$ can be lifted to the unlinking of the same 
half-edge pair in the original diagram $\Gamma$. Since $\Gamma$ is 1PI, this 
unlinking does not disconnect $\Gamma$. Re-collapsing the maximal insertions once 
again does not affect the connectedness of the result, so $\Gamma_{\mr{s}}$ 
does not become disconnected by the unlinking. On the other hand, the unlinking of a 
half-edge pair $\{h_1^{(k)},h_2^{(k)}\}$ were to disconnect $\Gamma_{\mr{s}}$, then 
necessarily the unlinking of either $\{e_1^{(k)},h_1^{(k)}\}$ or $\{e_2^{(k)},h_2^{(k)}\}$ 
would disconnect $\Gamma$, which contradicts the premise that $\Gamma$ is 1PI.

Thus $\Gamma_{\mr{s}}$ is 1PI, and two-particle irreducibility is equivalent to the 
absence of any insertions. But if $\Gamma_{\mr{s}}$ admits an insertion containing either 
$h_1^{(k)}$ or $h_2^{(k)}$, then this contradicts the maximality of the insertion 
$\Gamma^{(k)}$ in $\Gamma$. Moreover, if $\Gamma_{\mr{s}}$ admits an insertion 
containing \emph{none} of the $h_1^{(k)},h_2^{(k)}$, then this insertion lifts 
to an insertion in the original diagram $\Gamma$, hence this insertion (i.e., 
all of its interaction lines 
and half-edges) must have been omitted from $\Gamma_{\mr{s}}$ (contradiction). 
We conclude that $\Gamma_{\mr{s}}$ admits no insertions, hence is 2PI.
 
 Finally it remains to prove the uniqueness of the decomposition of 
 Eq.~\eqref{eqn:skeletonInsertion}. To this end, let 
 $\Gamma_{\mr{s}} \in \mf{F}_2^{\mathrm{2PI}}$ and $\Gamma^{(k)} \in \mf{F}_2^{\mr{c,t}}$ 
 for $k=1,\ldots,K$, and let $\left\{h_1^{(k)},h_2^{(k)}\right\}$ be distinct half-edge pairs 
 in $\Gamma_{\mr{s}}$ for $k=1,\ldots,K$. Then \emph{define} $\Gamma$ via 
 Eq.~\eqref{eqn:skeletonInsertion}. The uniqueness claim then follows if we can show 
 that the $\Gamma^{(k)}$ are the maximal insertions in $\Gamma$.
 
 Suppose that $\Gamma^{(k)}$ is not maximal for some $k$. Then 
 by definition the diagram $\Gamma'$ 
 formed from $\Gamma$ by collapsing the insertion $\Gamma^{(k)}$ admits
 an insertion containing $h_1^{(k)}$ or $h_2^{(k)}$ (assume $h_1^{(k)}$ without 
 loss of generality). In fact let $\Gamma^{(k)}_{\mr{m}}$ be a \emph{maximal} 
 insertion containing $h_1^{(k)}$. 
 Note also that $\Gamma'$ still admits the $\Gamma^{(j)}$ for 
 $j\neq k$ as insertions.
 
 Then for $j\neq k$ form the structure $(H^{(j,k)},\Pi^{(j,k)})$ by `merging' 
 $\Gamma^{(j)}$ and $\Gamma^{(k)}_{\mr{m}}$ via the same construction as above 
 (now within the overall diagram $\Gamma'$). By the same reasoning as in Lemma \ref{lem:sizeE}, 
 the set of unpaired half-edges $E^{(j,k)}$ must be of even cardinality, and we 
 can rule out the case $\vert E^{(j,k)} \vert =0$. 
 
 In the case $\vert E^{(j,k)} \vert = 2$, the structure $(H^{(j,k)},\Pi^{(j,k)})$ once 
 again defines a truncated Green's function diagram $\Gamma^{(j,k)}$ 
 admitting both $\Gamma^{(j)}$ and $\Gamma^{(k)}_{\mr{m}}$ as insertions. 
 But since $\Gamma^{(k)}_{\mr{m}}$ is maximal, $\Gamma^{(j,k)} = \Gamma^{(k)}_{\mr{m}}$, 
 and $\Gamma^{(j)}$ is contained in $\Gamma^{(k)}_{\mr{m}}$.
 
 In the case $\vert E^{(j,k)} \vert = 4$, since $\Gamma^{(j)}$ does 
 not contain $\Gamma^{(k)}_{\mr{m}}$ (i.e., the half-edge set of the former 
 does not contain that of the latter), the same reasoning as above guarantees that 
 $\Gamma^{(j)}$ and $\Gamma^{(k)}_{\mr{m}}$ do not overlap or touch.
 
 Then consider the diagram $\Gamma''$ formed from $\Gamma'$ 
 by collapsing the insertion $\Gamma^{(j)}$. In both of our cases (namely, 
 that $\Gamma^{(j)}$ is contained in $\Gamma^{(k)}_{\mr{m}}$ and 
 that $\Gamma^{(j)}$ and $\Gamma^{(k)}_{\mr{m}}$ do not overlap or touch), 
 the insertion $\Gamma^{(k)}_{\mr{m}}$ \emph{descends} to an insertion in 
 $\Gamma''$ containing $h_1^{(k)}$.
 
 Iteratively repeating this procedure for all $j\neq k$ (i.e., collapsing all 
 of the insertions $\Gamma^{(j)}$ to obtain the original skeleton diagram 
 $\Gamma_{\mr{s}}$), we find that 
 $\Gamma^{(k)}_{\mr{m}}$ descends to an insertion in $\Gamma_{\mr{s}}$ 
 containing $h_1^{(k)}$, contradicting the fact that $\Gamma_{\mr{s}}$ is 2PI.
 \end{proof}

\bibliographystyle{siam}
\bibliography{lwref}

\begin{thebibliography}{10}

\bibitem{AltlandSimons2010}
{\sc A.~Altland and B.~D. Simons}, {\em Condensed matter field theory},
  Cambridge Univ. Pr., 2010.

\bibitem{AmitMartin-Mayor2005}
{\sc D.~J. Amit and V.~Martin-Mayor}, {\em Field theory, the renormalization
  group, and critical phenomena: graphs to computers}, World Scientific
  Publishing Co Inc, 2005.

\bibitem{AryasetiawanGunnarsson1998}
{\sc F.~Aryasetiawan and O.~Gunnarsson}, {\em The {GW} method}, Rep. Prog.
  Phys., 61 (1998), p.~237.

\bibitem{BaymKadanoff1961}
{\sc G.~Baym and L.~P. Kadanoff}, {\em Conservation laws and correlation
  functions}, Phys. Rev., 124 (1961), p.~287.

\bibitem{BenlagraKimPepin2011}
{\sc A.~Benlagra, K.-S. Kim, and C.~P{\'{e}}pin}, {\em {The Luttinger-Ward
  functional approach in the Eliashberg framework: a systematic derivation of
  scaling for thermodynamics near the quantum critical point}}, J. Phys.
  Condens. Matter, 23 (2011), p.~145601.

\bibitem{BernardiPalummoGrossman2013}
{\sc M.~Bernardi, M.~Palummo, and J.~C. Grossman}, {\em Extraordinary sunlight
  absorption and one nanometer thick photovoltaics using two-dimensional
  monolayer materials}, Nano Lett., 13 (2013), pp.~3664--3670.

\bibitem{BiermannAryasetiawanGeorges2003}
{\sc S.~Biermann, F.~Aryasetiawan, and A.~Georges}, {\em First-principles
  approach to the electronic structure of strongly correlated systems:
  Combining the {GW} approximation and dynamical mean-field theory}, Phys. Rev.
  Lett., 90 (2003), p.~086402.

\bibitem{BullaCostiPruschke2008}
{\sc R.~Bulla, T.~A. Costi, and T.~Pruschke}, {\em {Numerical renormalization
  group method for quantum impurity systems}}, Rev. Mod. Phys., 80 (2008),
  pp.~395--450.

\bibitem{DahlenVanVon2005}
{\sc N.~E. Dahlen, R.~{Van Leeuwen}, and U.~{Von Barth}}, {\em {Variational
  energy functionals of the green function tested on molecules}}, in Int. J.
  Quantum Chem., vol.~101, 2005, pp.~512--519.

\bibitem{DeligneEtingofFreedEtAl1999}
{\sc P.~Deligne, P.~Etingof, D.~S. Freed, L.~C. Jeffrey, D.~Kazhdan, J.~W.
  Morgan, D.~R. Morrison, and E.~Witten}, eds., {\em {Quantum fields and
  strings: A course for mathematicians. Vol. 1, 2}}, Amer. Math. Soc., 1999.

\bibitem{Elder2014}
{\sc R.~Elder}, {\em {Comment on ``Non-existence of the Luttinger-Ward
  functional and misleading convergence of skeleton diagrammatic series for
  Hubbard-like models''}}, arXiv:1407.6599,  (2014).

\bibitem{FetterWalecka2003}
{\sc A.~L. Fetter and J.~D. Walecka}, {\em Quantum theory of many-particle
  systems}, Courier Corp., 2003.

\bibitem{GeorgesKotliarKrauthEtAl1996}
{\sc A.~Georges, G.~Kotliar, W.~Krauth, and M.~J. Rozenberg}, {\em Dynamical
  mean-field theory of strongly correlated fermion systems and the limit of
  infinite dimensions}, Rev. Mod. Phys., 68 (1996), p.~13.

\bibitem{GunnarssonRohringerSchaeferEtAl2017}
{\sc O.~Gunnarsson, G.~Rohringer, T.~Sch{\"{a}}fer, G.~Sangiovanni, and
  A.~Toschi}, {\em {Breakdown of Traditional Many-Body Theories for Correlated
  Electrons}}, Phys. Rev. Lett., 119 (2017), p.~056402.

\bibitem{Hedin1965}
{\sc L.~Hedin}, {\em New method for calculating the one-particle {Green's}
  function with application to the electron-gas problem}, Phys. Rev., 139
  (1965), p.~A796.

\bibitem{Ismail-Beigi2010}
{\sc S.~Ismail-Beigi}, {\em {Correlation energy functional within the GW -RPA:
  Exact forms, approximate forms, and challenges}}, Phys. Rev. B, 81 (2010),
  pp.~1--21.

\bibitem{Isserlis1918}
{\sc L.~Isserlis}, {\em On a formula for the product-moment coefficient of any
  order of a normal frequency distribution in any number of variables},
  Biometrika, 12 (1918), pp.~134--139.

\bibitem{KotliarSavrasovHauleEtAl2006}
{\sc G.~Kotliar, S.~Y. Savrasov, K.~Haule, V.~S. Oudovenko, O.~Parcollet, and
  C.~A. Marianetti}, {\em Electronic structure calculations with dynamical
  mean-field theory}, Rev. Mod. Phys., 78 (2006), p.~865.

\bibitem{KozikFerreroGeorges2015}
{\sc E.~Kozik, M.~Ferrero, and A.~Georges}, {\em {Nonexistence of the
  Luttinger-Ward Functional and Misleading Convergence of Skeleton Diagrammatic
  Series for Hubbard-Like Models}}, Phys. Rev. Lett., 114 (2015), p.~156402.

\bibitem{LiLu2017}
{\sc Y.~Li and J.~Lu}, {\em Bold diagrammatic {M}onte {C}arlo in the lens of
  stochastic iterative methods}, arXiv:1710.00966,  (2017).

\bibitem{LinLindsey2018}
{\sc L.~Lin and M.~Lindsey}, {\em Variational structure of {Luttinger-Ward}
  formalism and bold diagrammatic expansion for {Euclidean} lattice field
  theory}, Proc. Natl. Acad. Sci., 115 (2018), p.~2282.

\bibitem{LuttingerWard1960}
{\sc J.~M. Luttinger and J.~C. Ward}, {\em {Ground-state energy of a
  many-fermion system. II}}, Phys. Rev., 118 (1960), p.~1417.

\bibitem{MartinReiningCeperley2016}
{\sc R.~M. Martin, L.~Reining, and D.~M. Ceperley}, {\em Interacting
  Electrons}, Cambridge Univ. Pr., 2016.

\bibitem{NegeleOrland1988}
{\sc J.~W. Negele and H.~Orland}, {\em Quantum many-particle systems},
  Westview, 1988.

\bibitem{OnidaReiningRubio2002}
{\sc G.~Onida, L.~Reining, and A.~Rubio}, {\em {Electronic excitations:
  density-functional versus many-body Green's-function approaches}}, Rev. Mod.
  Phys., 74 (2002), p.~601.

\bibitem{PeskinSchroeder1995}
{\sc M.~E. Peskin and D.~V. Schroeder}, {\em An introduction to quantum field
  theory},  (1995).

\bibitem{Polyak2004}
{\sc M.~Polyak}, {\em {Feynman diagrams for pedestrians and mathematicians}},
  in Graphs and Patterns in Mathematics and Theoretical Physics, 2004,
  pp.~1--27.

\bibitem{Potthoff2006}
{\sc M.~Potthoff}, {\em {Non-perturbative construction of the Luttinger-Ward
  functional}}, Condens. Matter Phys., 9 (2006), pp.~557--567.

\bibitem{ProkofevSvistunov2008}
{\sc NV~Prokof'ev and BV~Svistunov}, {\em {Bold diagrammatic Monte Carlo: A
  generic sign-problem tolerant technique for polaron models and possibly
  interacting many-body problems}}, Physical Rev. B, 77 (2008), p.~125101.

\bibitem{RentropMedenJakobs2016}
{\sc J.~F. Rentrop, V.~Meden, and S.~G. Jakobs}, {\em {Renormalization group
  flow of the Luttinger-Ward functional: Conserving approximations and
  application to the Anderson impurity model}}, Phys. Rev. B, 93 (2016),
  p.~195160.

\bibitem{StaarMaierSchulthess2013}
{\sc P.~Staar, T.~Maier, and T.~C. Schulthess}, {\em Dynamical cluster
  approximation with continuous lattice self-energy}, Phys. Rev. B, 88 (2013),
  p.~115101.

\bibitem{StefanucciVanLeeuwen2013}
{\sc G.~Stefanucci and R.~Van~Leeuwen}, {\em Nonequilibrium many-body theory of
  quantum systems: a modern introduction}, Cambridge Univ. Pr., 2013.

\bibitem{SunChan2016}
{\sc Q.~Sun and G.~K.-L. Chan}, {\em Quantum embedding theories}, Acc. Chem.
  Res., 49 (2016), pp.~2705--2712.

\bibitem{TarantinoRomanielloBergerEtAl2017}
{\sc W.~Tarantino, P.~Romaniello, J.~A. Berger, and L.~Reining}, {\em
  {Self-consistent Dyson equation and self-energy functionals: An analysis and
  illustration on the example of the Hubbard atom}}, Phys. Rev. B, 96 (2017),
  p.~045124.

\bibitem{Zinn-Justin2002}
{\sc J.~Zinn-Justin}, {\em Quantum Field Theory and Critical Phenomena},
  Clarendon Pr., 2002.

\end{thebibliography}

\end{document}